\documentclass[11pt]{article}

\usepackage[usenames,dvipsnames]{xcolor}
\definecolor{Gred}{RGB}{219, 50, 54}
\definecolor{ToCgreen}{RGB}{0, 128, 0}

\usepackage[margin=1in]{geometry}
\usepackage[T1]{fontenc}

\usepackage[scale=0.97]{XCharter} 

\usepackage[libertine,bigdelims,vvarbb,scaled=1.05]{newtxmath} 

\usepackage{babel}
\usepackage[spacing=true,kerning=true,babel=true,tracking=true]{microtype}

\DeclareMathAlphabet{\pazocal}{OMS}{zplm}{m}{n} 
\renewcommand{\mathcal}[1]{\pazocal{#1}}

\usepackage{makecell}
\usepackage{dsfont}

\usepackage{multirow}
\usepackage{amsmath,amsthm}
\usepackage{bm}
\usepackage{bbm}
\usepackage{textgreek}
\usepackage{mathtools}
\usepackage{enumitem}
\usepackage[numbers,comma,sort&compress]{natbib}
\usepackage{authblk}
\usepackage{graphicx}
\usepackage[font=small]{caption}
\usepackage[labelformat=simple]{subcaption}

\usepackage{float}
\usepackage[linesnumbered,ruled,vlined]{algorithm2e}
\SetKwInput{KwInput}{Input}
\SetKwInput{KwOutput}{Output}
\SetKwInOut{Promise}{Promise}
\SetKwInput{Goal}{Goal}
\SetKwProg{Fn}{function}{}{}
\SetKwFor{RepTimes}{repeat}{times}{}
\SetKwFunction{Ver}{Verify}
\SetKwFunction{Prep}{PrepareState}
\usepackage{algorithmic}

\usepackage{physics}
\usepackage{footnote}
\usepackage{xcolor}
\usepackage{mathrsfs}
\usepackage{bbm}
\usepackage{braket}
\usepackage[colorlinks]{hyperref}
\usepackage{cleveref}
\hypersetup{
      colorlinks=true,
  citecolor=ToCgreen,
  linkcolor=Sepia,
  filecolor=Gred,
  urlcolor=Gred
  }
\usepackage{placeins}

\newtheorem{theorem}{Theorem}[section]
 
\newtheorem{definition}[theorem]{Definition}
\newtheorem{lemma}[theorem]{Lemma}

\newtheorem{corollary}[theorem]{Corollary}
\newtheorem{remark}[theorem]{Remark}
\newtheorem{problem}[theorem]{Problem}

\newcommand{\algo}[1]{\hyperref[algo:#1]{Algorithm~\ref*{algo:#1}}}

\renewcommand{\geq}{\geqslant}
\renewcommand{\leq}{\leqslant}
\renewcommand{\ge}{\geqslant}
\renewcommand{\le}{\leqslant}

\newcommand{\E}{\mathbb{E}}

\newcommand{\cB}{\mathcal{B}}

\newcommand{\cC}{\mathcal{C}}
\newcommand{\cP}{\mathcal{P}}
\newcommand{\cD}{\mathcal{D}}

\newcommand{\cS}{\mathcal{S}}

\newcommand{\Weyl}{\operatorname{Weyl}}
\newcommand{\Stab}{\operatorname{Stab}}
\usepackage{textgreek}
\newcommand{\sgn}{\varsigma}

\DeclareMathOperator{\poly}{poly}
\DeclareMathOperator{\spn}{span}

\DeclareMathOperator*{\argmax}{argmax}

\renewcommand{\tr}{\mathrm{tr}}
\renewcommand{\Tr}{\mathrm{tr}}
\newcommand{\SWAP}{\textnormal{SWAP}}

\newcommand{\0}{\mathbf{0}}
\renewcommand{\emptyset}{\varnothing}
\def\Tr{\operatorname{tr}}\def\:{\hbox{\bf:}}

\renewcommand{\epsilon}{\varepsilon}

\allowdisplaybreaks

\begin{document}

\title{Stabilizer bootstrapping: \\
A recipe for efficient agnostic tomography and magic estimation}

\author{
Sitan Chen
\thanks{SEAS, Harvard University. Email: \href{mailto:sitan@seas.harvard.edu}{sitan@seas.harvard.edu}.}
\qquad\qquad
Weiyuan Gong
\thanks{SEAS, Harvard University. Email: \href{mailto:wgong@g.harvard.edu}{wgong@g.harvard.edu}.}
\qquad\qquad
Qi Ye
\thanks{IIIS, Tsinghua University. Email: \href{mailto:yeq22@mails.tsinghua.edu.cn}{yeq22@mails.tsinghua.edu.cn}.}
\qquad\qquad
Zhihan Zhang
\thanks{IIIS, Tsinghua University. Email: \href{mailto:zhihan-z21@mails.tsinghua.edu.cn}{zhihan-z21@mails.tsinghua.edu.cn}.}
}

\date{August 13, 2024}
\maketitle

\begin{abstract}
      We study the task of agnostic tomography: given copies of an unknown $n$-qubit state $\rho$ which has fidelity $\tau$ with some state in a given class $\mathcal{C}$, find a state which has fidelity $\ge \tau - \epsilon$ with $\rho$. We give a new framework, \emph{stabilizer bootstrapping}, for designing computationally efficient protocols for this task, and use this to get new agnostic tomography protocols for the following classes:
      \begin{itemize}[leftmargin=*]
            \item \textbf{Stabilizer states}: We give a protocol that runs in time $\mathrm{poly}(n,1/\epsilon)\cdot (1/\tau)^{O(\log(1/\tau))}$, answering an open question posed by Grewal, Iyer, Kretschmer, Liang~\cite{grewal2024improved} and Anshu and Arunachalam~\cite{anshu2023survey}. Previous protocols ran in time $\mathrm{exp}(\Theta(n))$ or required $\tau>\cos^2(\pi/8)$.
            \item \textbf{States with stabilizer dimension $n - t$}: We give a protocol that runs in time $n^3\cdot(2^t/\tau)^{O(\log(1/\epsilon))}$, extending recent work on learning quantum states prepared by circuits with few non-Clifford gates, which only applied in the realizable setting where $\tau = 1$~\cite{grewal2023efficient,leone2024learning,chia2024efficient,hangleiter2024bell}.
            \item \textbf{Discrete product states}: If $\mathcal{C} = \mathcal{K}^{\otimes n}$ for some $\mu$-separated discrete set $\mathcal{K}$ of single-qubit states, we give a protocol that runs in time $(n/\mu)^{O((1 + \log (1/\tau))/\mu)}/\epsilon^2$. This strictly generalizes a prior guarantee which applied to stabilizer product states~\cite{grewal2024agnostic}. For stabilizer product states, we give a further improved protocol that runs in time $(n^2/\epsilon^2)\cdot (1/\tau)^{O(\log(1/\tau))}$.
      \end{itemize}
      As a corollary, we give the first protocol for estimating stabilizer fidelity, a standard measure of magic for quantum states, to error $\epsilon$ in $n^3 \mathrm{quasipoly}(1/\epsilon)$ time.
\end{abstract}

\clearpage

\tableofcontents
\clearpage


\newpage

\section{Introduction}

State tomography is a core primitive in quantum information, especially in the characterization, benchmarking, and verification of quantum systems~\cite{banaszek2013focus}. Yet without additional assumptions on the quantum state in consideration, the complexity of recovering a description of it from measurements scales exponentially with the system size~\cite{bruss1998optimal,odonnell2016efficient,haah2016sample}, in terms of both statistical and computational complexity.

In practice, however, the states of interest tend to be far more benign and structured than these worst-case lower bounds would suggest. There has been a growing body of work showing that when the underlying state $\rho$ satisfies some simple ansatz \--- e.g. $\rho$ is prepared by a quantum circuit of low depth~\cite{huang2024learning} or with few non-Clifford gates~\cite{chia2024efficient,leone2024learning,grewal2023efficient}, has low entanglement like a matrix product state~\cite{arad2017rigorous,landau2015polynomial,cramer2010efficient}, or is the thermal state of a local Hamiltonian~\cite{haah2024learning,bakshi2024learning,anshu2021sample} \--- tomography can be performed with only a \emph{polynomial} amount of data and compute. 

Unfortunately, the algorithms developed in these works crucially require that the ansatz hold \emph{exactly} \--- in classical learning theory parlance, these algorithms only work in the \emph{realizable} setting. In reality, however, the ansatz is at best an approximation to $\rho$, which will deviate from it either because it has been corrupted by noise or because the ansatz is insufficiently expressive.

Motivated by this, several recent works~\cite{grewal2024agnostic,anshu2023survey,buadescu2021improved,grewal2024improved} have proposed \emph{agnostic tomography} as a powerful new solution concept:

\begin{definition}[Agnostic tomography]\label{def:agnostic_learning}
Given $0 < \epsilon, \delta < 1$ and an ansatz corresponding to a class $\cC$ of $n$-qubit mixed states, \emph{agnostic tomography of $\cC$} is the following task. Given copies of a mixed state $\rho$, output a classical description of a state $\sigma$ for which the fidelity $F(\sigma,\rho)$ satisfies
\begin{equation*}
F(\sigma,\rho)\geq\max_{\sigma'\in\cC}F(\sigma',\rho)-\epsilon
\end{equation*}
with probability at least $1 - \delta$.
If the algorithm outputs a state $\sigma\in\cC$, then it is said to be \emph{proper}; otherwise, it is said to be \emph{improper}.
\end{definition}

\noindent This is the natural quantum analog of \emph{agnostic learning} from classical learning theory~\cite{kearns1992toward}, which was originally introduced for the very same reason of going beyond realizable learning.

Our understanding of the computational complexity of agnostic tomography is currently limited. In fact, as we explain below, when $\tau = o_n(1)$, there are no known agnostic tomography algorithms for interesting families of states that run in $\mathrm{poly}(n)$ time. In this work, we ask:
\begin{center}
	{\em Are there efficient, general-purpose algorithms for agnostic tomography?}
\end{center}
We answer this in the affirmative by developing a new algorithmic framework for this task that we call \emph{stabilizer bootstrapping}. Before describing this recipe in Section~\ref{sec:recipe}, we detail four new learning results that we obtain as applications, for agnostic tomography of 1) stabilizer states, 2) states with high stabilizer dimension, 3) discrete product states, and 4) stabilizer product states.

\subsection{Our results}

\subsubsection{Stabilizer states}
\label{sec:stabilizer_our_results} Widely studied within quantum information is the class of \emph{stabilizer states}, namely states preparable by quantum circuits consisting solely of Clifford gates. It is well-known that such circuits are efficiently classically simulable~\cite{gottesman1998theory,aaronson2004improved}. In general, the computational cost of existing methods for classical simulation increases as the quantum state in question deviates from the set of stabilizer states~\cite{bravyi2016trading,bravyi2016improved,bravyi2019simulation}. This deviation can be quantified in various ways, collectively referred to as \emph{magic}.

The question of agnostic tomography for stabilizer states has a natural interpretation from this perspective: given copies of a state $\rho$ with a bounded amount of magic, can we learn an approximation which is competitive with the best approximation by a stabilizer state? Previously, Montanaro~\cite{montanaro2017learning} gave the first complete proof that this can be done when $\rho$ has zero magic, i.e. in the realizable setting where $\rho$ is exactly a stabilizer state. Subsequently, Grewal, Iyer, Kretschmer, and Liang~\cite{grewal2024improved} extended this to show that when $\rho$ has fidelity larger than $\tau = \cos^2(\pi/8)$ with respect to some stabilizer state, there is a polynomial-time algorithm for agnostic tomography. They also gave an algorithm for general $\tau$ that runs in time exponential in the system size $n$.

As our first application of stabilizer bootstrapping, we give an algorithm for general $\tau$ that scales polynomially in $n$, answering an open question of Ref.~\cite{grewal2024improved,anshu2023survey}:

\begin{theorem}[Informal, see \Cref{thm:all_gamma_approximate_local_maximizer} and \Cref{cor:agnostic_learning_stabilizer}]\label{thm:agnostic_learning_stabilizer_informal}
	Fix any $1 \ge \tau\geq\epsilon\geq 0$. There is an algorithm that, given access to copies of a mixed state $\rho$ with $\max_{\ket{\phi'}\in \cC}\bra{\phi'}\rho\ket{\phi'} \ge \tau$ for $\cC$ the class of stabilizer states, outputs a stabilizer state $\ket{\phi}$ such that $\bra{\phi}\rho\ket{\phi}\geq\tau-\epsilon$ with high probability. The algorithm performs single-copy and two-copy measurements on at most $n(1/\tau)^{O(\log 1/\tau)}+O((1+\log^2(1/\tau))/\epsilon^2)$ copies of $\rho$ and runs in time $O(n^2(n+1/\epsilon^2))\cdot (1/\tau)^{O(\log 1/\tau)}$.
\end{theorem}

\noindent Note that this matches the sample and time complexity of Montanaro's algorithm in the realizable case, and as long as $\tau \ge \mathrm{exp}(-c\sqrt{\log n})$ for sufficiently small constant $c$, this algorithm runs in time $\mathrm{poly}(n,1/\epsilon)$. Additionally, the guarantee holds for mixed states $\rho$, whereas the prior result of Grewal et al.~\cite{grewal2024improved} only applied to pure states.

A powerful consequence of our agnostic tomography result is that it gives a way to estimate magic, a task of practical interest given the need to characterize proposed quantum devices' capacity for quantum advantage~\cite{oliviero2022measuring,tarabunga2024nonstabilizerness,tirrito2024quantifying,haug2023quantifying}. A natural measure of magic is \emph{stabilizer fidelity}, the fidelity between the state and the closest stabilizer state; this quantity is also closely related to other notions of magic and to the cost of known algorithms for classical simulation of quantum circuits~\cite{bravyi2019simulation}. Previously there were no algorithms for estimating this quantity in time better than exponential in $n$. Our agnostic tomography algorithm implies the following:
 
\begin{theorem}[Informal, see Corollary~\ref{cor:fid_est}]\label{thm:magic_est_informal}
 There is an algorithm for estimating the stabilizer fidelity of any $n$-qubit mixed state $\rho$ to error $\epsilon$ in time $n^3 (1/\epsilon)^{O(\log 1/\epsilon)}$ using $n(1/\epsilon)^{O(\log 1/\epsilon)}$ copies of $\rho$.
\end{theorem}
 
\noindent In Section~\ref{sec:related}, we situate this result within the broader literature on quantum resource theory that has defined various notions of magic and formulated protocols for estimating them. 

Finally, we mention that the guarantee we prove for agnostic tomography is actually more powerful than Theorem~\ref{thm:agnostic_learning_stabilizer_informal}. In fact, we give a \emph{list-decoding algorithm} that returns a list of length $(1/\tau)^{O(\log 1/\tau)}$ containing \emph{all} states with fidelity at least $\tau$ with $\rho$ which additionally are ``approximate local maximizers'' of fidelity (see Corollary~\ref{cor:stab_list}). Roughly speaking, a stabilizer state is an approximate local maximizer if it achieves approximately the maximum fidelity with $\rho$ among all of its nearest neighbors in the set of stabilizer states (see \Cref{def:localmax} for a formal definition). 

\begin{remark}[Implication for optimization landscape]
	This list-decoding guarantee gives an algorithmic proof of a new structural result about the optimization landscape over stabilizer states. For context, it is well-known that when $\rho$ is a stabilizer state, there are exponentially many nearest neighbors, each of which achieves fidelity $1/2$ with $\rho$ and is thus a $1/2$-approximate local maximizer. The list-decoding algorithm implies that in contrast, there are only $(\xi\tau)^{-O(\log 1/\tau)}$ many $(1/2 + \xi)$-approximate local maximizers with fidelity $\tau$ with $\rho$, for any mixed state $\rho$ (see \Cref{cor:structural} for details).
\end{remark}

\subsubsection{States with high stabilizer dimension}
\label{sec:dim_our_results}

Given that circuits consisting solely of Clifford gates cannot achieve universal quantum computation, it is natural to consider circuits which also contain some number of non-Clifford gates. These are called \emph{doped circuits}, and because arbitrary quantum gates can be decomposed into Clifford gates and a bounded number of non-Clifford gates, e.g. T gates, the number of non-Clifford gates offers a natural sliding scale for interpolating between classically simulable circuits and universal quantum computation.

Recently, several algorithms have been proposed for learning states generated by such circuits, where the time complexity of the algorithms scales exponentially in the number of non-Clifford gates~\cite{grewal2023efficient,leone2024learning,chia2024efficient,hangleiter2024bell} but polynomially in the system size. More generally, some of these algorithms can learn states with \emph{stabilizer dimension} $n - t$ in time $\mathrm{poly}(n,2^t)$. These are states that are stabilized by a commuting set of $2^{n-t}$ Pauli operators, and they include states prepared by circuits doped with at most $t/2$ non-Clifford gates as a special case. These results strictly generalize the realizable learning result of Montanaro~\cite{montanaro2017learning} for stabilizer states, which corresponds to the case of $t = 0$.

As discussed previously, assuming that the underlying state is \emph{exactly} preparable by a circuit of Clifford gates and $t$ non-Clifford gates is quite constraining \--- even a small amount of depolarizing noise in the output of such a circuit will result in a state deviating from this assumption. Yet to our knowledge, no guarantees were known for agnostic tomography of such states, motivating our next application of stabilizer bootstrapping:

\begin{theorem}[Informal, see \Cref{thm:agnostic_learning_high_stabilizer_dimension_states}]\label{thm:agnostic_learning_high_stabilizer_dimension_states_informal}
	Fix any $1 \ge \tau\geq\epsilon\geq 0$. There is an algorithm that, given access to copies of a mixed state $\rho$ with $\max_{\ket{\phi'}\in \cC}\bra{\phi'}\rho\ket{\phi'} \ge \tau$ for $\cC$ the class of states with stabilizer dimension at least $n - t$, outputs a state $\ket{\phi}$ with stabilizer dimension at least $n - t$ such that $\bra{\phi}\rho\ket{\phi} \ge \tau - \epsilon$ with high probability. The algorithm performs single- and two-copy measurements on at most $n(2^t/\tau)^{O(\log(1/\epsilon))}$ copies of $\rho$ and runs in time $n^3(2^t/\tau)^{O(\log(1/\epsilon))}$.
\end{theorem}

\noindent This gives the first nontrivial extension of the aforementioned works on realizable learning of states with high stabilizer dimension to the agnostic setting. 
It achieves the same runtime as Theorem~\ref{thm:agnostic_learning_stabilizer_informal} with the caveat that the dependence on $\epsilon$ is worse, and in fact it is also worse than what is achieved in the realizable learning results. 
In \Cref{remark: high dim 1}, we introduce a simple modification that recovers the $\poly(n, 2^t, 1/\epsilon)$ runtime from the realizable learning results as long as $\tau$ is close to 1. Whether there is a $\poly(n, 2^t, 1/\epsilon)$ algorithm for a wider range of $\tau$ remains open.

\subsubsection{Discrete product states}
\label{sec:discrete_our_results}

The stabilizer bootstrapping framework is also useful beyond the realm of stabilizer states and relaxations thereof. As a proof of concept, here we consider the set of \emph{discrete product states}: let $\mathcal{K}$ denote some discrete set of single-qubit states, and let $\cC$ consist of all states obtained by $n$-fold tensor products of elements of $\mathcal{K}$. Previously, a different work of Grewal, Iyer, Kretschmer, and Liang~\cite{grewal2024agnostic} showed that when $\mathcal{K} = \{\ket{0}, \ket{1}, \ket{+}, \ket{-}, \ket{+i}, \ket{-i}\}$, i.e. when $\cC$ is the set of \emph{stabilizer product states}, there is an agnostic tomography algorithm that runs in time $n^{O(1 + \log 1/\tau)}/\epsilon^2$.

Our first result is to give a very simple instantiation of stabilizer bootstrapping that achieves the same runtime, but for \emph{any} family of discrete product states for which the states in $\mathcal{K}$ are angularly separated:

\begin{theorem}[Informal, see \Cref{thm:product_base} and \Cref{cor:product_base}]\label{thm:product_base_informal}
	Fix any $1 \ge \tau\geq\epsilon\geq 0$ and $\mu > 0$. Let $\mathcal{K}$ be a set of single-qubit pure states for which $|\braket{\phi_1|\phi_2}|^2 \leq 1 - \mu$ for any distinct $\ket{\phi_1},\ket{\phi_2}\in\mathcal{K}$. There is an algorithm that, given access to copies of a mixed state $\rho$ with $\max_{\ket{\phi'}\in \cC}\bra{\phi'}\rho\ket{\phi'} \ge \tau$ for $\cC = \mathcal{K}^{\otimes n}$, outputs a product state $\ket{\phi}\in\cC$ such that $\bra{\phi}\rho\ket{\phi}\geq\tau-\epsilon$ with high probability. The algorithm only performs single-copy measurements and has time and sample complexity $(n|\mathcal{K}|)^{O((1+\log(1/\tau))/\mu)}/\epsilon^2$.
\end{theorem}

\noindent When $\mathcal{K}$ is the set of single-qubit stabilizer states, we can take $\mu$ above to be $1/2$, thus recovering the runtime of Grewal et al.~\cite{grewal2024agnostic}. In fact, Theorem~\ref{thm:product_base_informal} gives a slight improvement: the protocol only uses single-copy measurements whereas prior work relied on Bell difference sampling which involves two-copy measurements.

For other choices of $\mathcal{K}$, note that $|\mathcal{K}| \le \mathrm{poly}(1/\mu)$, and in particular when the angular separation $\mu$ is a constant, any $\mathcal{K}$ will have constant size and thus result in the same runtime, up to polynomial overheads, as in the special case of stabilizer product states.

As with our result on learning stabilizer states, we actually prove a stronger list-decoding guarantee in which the algorithm returns a list of length $(n|\mathcal{K}|)^{O((1+\log 1/\tau)/\mu)}$ containing \emph{all} states in $\cC$ with fidelity at least $\tau$ with $\rho$ (see~\Cref{cor:prod_list} for details).

\vspace{0.5em} \noindent \textbf{Concurrent work.} \enspace Concurrently with and independently of this work, \cite{bakshi2024learninga} also obtained an agnostic tomography algorithm in the setting of~\Cref{thm:product_base_informal} with time and sample complexity $(n|\mathcal{K}|)^{O((1 + \log(1/\tau))/\mu)}$ using a different set of techniques.

\subsubsection{Stabilizer product states}
\label{sec:stabilizer_prod_our_results}
While \Cref{thm:product_base_informal} generalizes prior work on learning stabilizer product states, it turns out that a more careful application of stabilizer bootstrapping can also be used to give an improved algorithm in the special case of stabilizer product states. Indeed, a shortcoming of the runtime from the prior work~\cite{grewal2024agnostic} is that if $\tau = o_n(1)$, the runtime is super-polynomial in $n$. We address this by showing the following:

\begin{theorem}[Informal, see \Cref{thm:stab_product_base} and \Cref{cor:stab_product_base}]\label{thm:stabilizer_product_informal}
	Fix any $1 \ge \tau\geq\epsilon\geq 0$. There is an algorithm that, given access to copies of a mixed state $\rho$ with $\max_{\ket{\phi'}\in \cC}\bra{\phi'}\rho\ket{\phi'} \ge \tau$ for $\cC$ the set of stabilizer product states, outputs a stabilizer product state $\ket{\phi}\in\cC$ such that $\bra{\phi}\rho\ket{\phi}\geq\tau-\epsilon$ with high probability. The algorithm only performs single-copy and two-copy measurements on $\log n (1/\tau)^{O(\log 1/\tau)} + O(\log^2(1/\tau)/\epsilon^2)$ copies of $\rho$ and runs in time $n^2(1/\tau)^{O(\log 1/\tau)}/\epsilon^2$.
\end{theorem}

\subsubsection{Lower bounds}

Our algorithm for agnostic tomography of stabilizer states is polynomial-time in the regime where $\tau$ is at least slightly inverse sub-polynomial. We observe that this cannot be improved significantly because as soon as $\epsilon\leq\tau \ll 1/\mathrm{poly}(n)$, agnostic tomography is not possible with a polynomial number of samples, regardless of runtime, by the following information-theoretic lower bound:

\begin{theorem}[Informal, see \Cref{thm:info_lower}]\label{thm:info_lower_informal}
Assume $0<\epsilon < \tau/3$ and $\tau \ge \Omega(2^{-n})$. The sample complexity of agnostic tomography of stabilizer states within accuracy $\epsilon$ for input states $\rho$ with stabilizer fidelity at most $\tau$ is at least $\Omega(n/\tau)$. 
\end{theorem}

\noindent This applies to the general setting where $\rho$ can be mixed. Even if $\rho$ is pure however, we can show a lower bound of $\Omega(\tau^{-1})$ (see~\Cref{lem:magic_lower_agnostic}). Using a similar technique, we also show a sample complexity lower bound of $\Omega(\epsilon^{-1})$ for estimating the stabilizer fidelity of a pure state to error $\epsilon$ (see~\Cref{lem:magic_lower}). We note that the subset phase state construction of Ref.~\cite{ji2018pseudorandom,aaronson2024quantum}, which was used to show the existence of ``pseudo-magic'' states~\cite{gu2024pseudomagic} which have low magic but which are hard to distinguish from Haar-random, already implies a weaker sample complexity lower bound of $\Omega(\epsilon^{-1/2})$ for estimating stabilizer fidelity, but for technical reasons (see~\Cref{remark:compare_pseudo}), this does not immediately imply a lower bound for agnostic tomography and we have to construct a different hard instance to obtain our results.

These lower bounds however still leave open whether the quasipolynomial dependence on $1/\tau$ in Theorem~\ref{thm:agnostic_learning_stabilizer_informal} can be improved. 

In this work, we do not prove a lower bound ruling this out based on a standard hardness assumption. Instead, we note some interesting implications if such an efficient solution exists, even in the special case where $\rho$ is a subset state of a polynomial-sized subset $A\subseteq \mathbb{F}_2^n$. In this case, finding the closest stabilizer state up to error $\epsilon$ implies an algorithm for finding an $t= O(\log n)$-dimensional affine space with the largest intersection with $A$ (see \Cref{sec:lower_poly} for more details). The $t=n-1$ case of the latter problem is widely believed to be hard for quantum algorithms based on the \emph{learning parity with noise (LPN)} assumption~\cite{pietrzak2012cryptography}, and $t=\beta n$ case for constant $\beta$ is believed to be hard based on the \emph{learning subspace with noise (LSN)} assumption~\cite{dodis2009cryptography}.\footnote{The work~\cite{dodis2009cryptography} which introduced LSN did not conjecture quantum hardness, although no quantum attack is known.} An efficient solution to the agnostic tomography of stabilizer states may thus shed new light on this line of cryptographic assumptions. 

\subsection{Stabilizer bootstrapping}\label{sec:recipe}

Here we give a high-level description of the method we develop to obtain the results above. 

The starting point is the following observation. Let $\rho$ denote the input state, and let $\ket{\phi}$ denote the state from $\mathcal{C}$ with highest fidelity $\tau$ with $\rho$ (breaking ties arbitrarily). Now suppose one could produce a family of projectors $\Pi_1,\ldots,\Pi_n$ which (A) stabilize $\ket{\phi}$, (B) mutually commute, and (C) admit a sufficiently succinct representation that one can efficiently prepare a measurement in their joint eigenbasis. We will refer to such a family of projectors as \emph{complete}. Note that by (A) and (B), $\ket{\phi}$ is an element of this eigenbasis up to phase, so we can simply measure a copy of $\rho$ in this basis and we will obtain the outcome $\ket{\phi}$ with probability $\tau$. By repeating this $\Theta(\log(1/\delta)/\tau)$ times to accumulate a list of candidate states and estimating fidelities with each of them, we can identify a good approximation to $\ket{\phi}$.

The key challenge is how to produce such a complete family of projectors. In this work, we consider the following iterative procedure:
\begin{enumerate}[leftmargin=*,itemsep=0pt]
    \item \textbf{Find ``high-correlation'' projectors}: Compile a family of as many mutually commuting projectors $\Pi$ as possible which satisfy $\Tr(\Pi\rho) \ge \Omega(1)$. How this is implemented is specific to the application at hand.
    \item \textbf{If complete, terminate}: If the projectors form a complete family, then exit the loop and measure $\rho$ in their joint eigenbasis as described above.
    \item \textbf{If incomplete, sample a low-correlation projector $\Pi_{\sf low}$}: Use the incompleteness to argue that with non-negligible probability, one can find a \emph{low-correlation} projector $\Pi_{\sf low}$ which stabilizes $\ket{\phi}$. Like Step 1, how this step is implemented is specific to the application at hand.
    \item \textbf{Transform $\rho$ by measuring it with $\Pi_{\sf low}$}: Return to Step 1 but now with all subsequent copies of $\rho$ replaced by the post-measurement state given by measuring with $\{\Pi_{\sf low}, \mathrm{Id} - \Pi_{\sf low}\}$ and post-selecting on the former outcome.
\end{enumerate}

\noindent The ``bootstrapping'' in stabilizer bootstrapping happens in Step 4. To provide intuition for this step, consider the following calculation. Let 
\begin{equation}
    \rho' \triangleq \frac{\Pi_{\sf low} \rho \Pi_{\sf low}}{\Tr(\Pi_{\sf low}\rho)}
\end{equation}
denote the post-measurement state given we observe the outcome corresponding to $\Pi_{\sf low}$. Note that because $\rho$ has some fidelity with respect to $\ket{\phi}$, and $\Pi_{\sf low}$ stabilizes $\ket{\phi}$, this outcome happens with non-negligible probability. Now because $\Pi_{\sf low}$ is a low-correlation projector that stabilizes $\ket{\phi}$, we have 
\begin{equation}
    \bra{\phi}\rho'\ket{\phi} = \frac{\bra{\phi}\Pi_{\sf low} \rho \Pi_{\sf low}\ket{\phi}}{\Tr(\Pi_{\sf low}\rho)} = \frac{\bra{\phi}\rho\ket{\phi}}{\Tr(\Pi_{\sf low}\rho)} \ge \frac{\tau}{\Tr(\Pi_{\sf low}\rho)} \ge c\tau\,,
\end{equation}
where $c > 1$ is an absolute constant because $\Pi_{\sf low}$ is a low-correlation projector.
In other words, the fidelity between $\rho'$ and $\ket{\phi}$ gets amplified compared to the fidelity between $\rho$ and $\ket{\phi}$! We can then recurse, noting that this procedure will terminate after at most $O(\log 1/\tau)$ rounds because the fidelity cannot continue amplifying indefinitely.

Finally, we note that each of the first three steps above will have some probability of failure over the course of these recursive calls. Importantly, in our applications, each iteration of Step 3 only succeeds with some small but non-negligible probability. So in order for the final output of the algorithm to be correct, Step 3 must succeed across potentially as many as $O(\log 1/\tau)$ rounds. As a result, the entire algorithm may need to be repeated multiple times, resulting in a list of candidate states that one can then select from by measuring all of their fidelities with $\rho$. For this reason, the computational complexity of this procedure should be dominated by the cost of this repetition, and thus dictated by the probability with which each iteration of Step 3 successfully samples a low-correlation projector.

At this juncture, it should not be clear \emph{a priori} why Steps 1 - 4 can be implemented for interesting classes of states $\mathcal{C}$. In Section~\ref{sec:overview}, we give an overview of how we execute these for the classes discussed above.

We also caution that the recipe is not a silver bullet. For instance, as we will see for learning states with high stabilizer dimension in Theorem~\ref{thm:agnostic_learning_high_stabilizer_dimension_states_informal}, we need to slightly deviate from this recipe. Nevertheless, the general outline remains the same, and we hope that this approach will find future applications with the right problem-specific modifications.

\subsection{Outlook}

In this work, we obtained new agnostic tomography results for several different classes of states, e.g. improving significantly upon prior work for agnostic tomography of stabilizer states and stabilizer product states, using our new framework of stabilizer bootstrapping. As a corollary, we also obtain new structural characterizations of the optimization landscape for fidelity with respect to stabilizer states, as well as the first efficient algorithm for estimating the stabilizer fidelity of quantum states. Looking ahead, it will be interesting to see how this framework can be further developed to agnostically learn other interesting and physically relevant classes of states. Below, we mention some concrete open questions closely related to the current work.

\vspace{0.5em}\noindent\textbf{Agnostic tomography of stabilizer states with $1/\poly(n)$ stabilizer fidelity.} \enspace 
The most immediate open question is to better understand whether the $(1/\tau)^{O(\log 1/\tau)}$ dependence in any of our runtime bounds is necessary, either by improving upon it or establishing hardness under a standard cryptographic assumption. In particular, when $\tau=\max_{\ket{\phi'}\in\cS}\bra{\phi'}\rho\ket{\phi'}=1/\poly(n)$ for the input state $\rho$ and $\cS$ the class of stabilizer states, is there a $\mathrm{poly}(n)$-time algorithm? Note that for sample complexity, $\mathrm{poly}(n)$ is already known to be possible in this regime by the aforementioned connection to shadow tomography.

\vspace{0.5em}\noindent\textbf{Proper agnostic tomography of $t$-doped quantum states.} \enspace  While our algorithm for states with high stabilizer dimension is proper in the sense that it outputs a state with high stabilizer dimension, if $\rho$ has fidelity at least $\tau$ specifically with a $t$-doped state, then our algorithm does not necessarily output a $t$-doped state. That is to say, our algorithm solves the task of improper agnostic tomography but does not solve the task of proper agnostic tomography of $t$-doped states, even though proper learning of such states is known to be possible in the realizable setting~\cite{grewal2023efficient,leone2024learning,chia2024efficient,hangleiter2024bell}.

\vspace{0.5em}\noindent\textbf{Learning states with bounded stabilizer rank. } \enspace 
A state is said to have stabilizer rank at most $r$ if it can be written as the superposition of at most $r$ stabilizer states. The cost of existing classical simulation methods~\cite{bravyi2019simulation} scales with this quantity, and it is open~\cite{anshu2023survey} to obtain algorithms, even in the realizable setting, for learning states with bounded stabilizer rank.

\vspace{0.5em}\noindent\textbf{Protocols that use single-copy measurements} \enspace 
Our algorithms use Bell difference sampling, which requires two-copy measurements. In the realizable case, it is known how to learn stabilizer states and even $t$-doped quantum states using single-copy measurements~\cite{grewal2023efficient}, via a procedure called \emph{computational difference sampling}. It would be interesting to see whether computational difference sampling could be used in Steps 1 and 3 of the stabilizer bootstrapping recipe to obtain analogous results in the agnostic tomography setting.

\subsection{Roadmap}

In Section~\ref{sec:overview} we give an overview of our techniques for implementing stabilizer bootstrapping to obtain our main results. In Section~\ref{sec:related} we discuss related work in classical simulation, quantum state learning, and magic estimation. In Section~\ref{sec:prelim} we provide technical preliminaries. In Section~\ref{sec:bell_properties}, we prove some useful properties of Bell difference sampling that will be crucial to our proofs of Theorems~\ref{thm:agnostic_learning_stabilizer_informal}, \ref{thm:agnostic_learning_high_stabilizer_dimension_states_informal}, and~\ref{thm:stabilizer_product_informal}. In Sections~\ref{sec:stabilizer}-\ref{sec:stabilizer_product}, we prove Theorems~\ref{thm:agnostic_learning_stabilizer_informal}-\ref{thm:stabilizer_product_informal}. In Section~\ref{sec:lower}, we present our lower bounds on the statistical and computational complexity of agnostic tomography of stabilizer states. Appendix~\ref{app:defer} contains proofs deferred from the main body of the paper.

\section{Overview of techniques}
\label{sec:overview}

In this section, we provide an overview of how we instantiate our stabilizer bootstrapping recipe for the different classes of states that we consider. Recall that the application-specific parts of the recipe are Steps 1 and 3, namely accumulating a collection of high-correlation projectors, and sampling a low-correlation projector with non-negligible probability if the collection is incomplete. Our applications to learning stabilizer states, discrete product states, and stabilizer product states hew closely to the recipe, and in Sections~\ref{sec:stabilizer_overview}-\ref{sec:stabilizer_prod_overview}, we focus on explaining how to implement Steps 1 and 3. For our application to learning states with high stabilizer dimension, we need to deviate somewhat from the recipe, and in Section~\ref{sec:dim_overview} we explain what needs to be changed and how to implement the resulting steps.

\subsection{Stabilizer states}
\label{sec:stabilizer_overview}

In this discussion, let $\ket{\phi}$ denote the stabilizer state which is closest to $\rho$, breaking ties arbitrarily, and suppose $F(\rho,\ket{\phi}) \ge \tau$. In stabilizer bootstrapping, the goal is to produce a family of mutually commuting projectors that all stabilize $\ket{\phi}$. As $\ket{\phi}$ is a stabilizer state, it is natural to try to find the \emph{stabilizer group} of $\ket{\phi}$, denoted $\Weyl(\ket{\phi})$, i.e. the family of $2^n$ Pauli operators $P\in\{I,X,Y,Z\}^{\otimes n}$ which stabilize $\phi$ up to sign. The complete family of projectors we will try to find will then consist of projectors of the form $\frac{I + \sgn P}{2}$ for all $P\in\Weyl(\ket{\phi})$ and for $\sgn\in\{\pm 1\}$ satisfying $P\ket{\phi} = \sgn\ket{\phi}$.

As in prior work~\cite{grewal2024improved}, our starting point is \emph{Bell difference sampling}, a procedure which measures two pairs of copies of $\rho$ in the Bell basis and adds the resulting outcomes to get a string $\mathbb{F}^{2n}_2$ corresponding to some Pauli operator. Let $\cB_\rho$ denote the distribution over the final output. Previously it was shown that $\cB_\rho$ places mass at least $\Omega(\tau^4)$ on $\Weyl(\ket{\phi})$ and is sufficiently evenly spread over it that by repeatedly sampling from $\cB_\rho$ for $O(n/\tau^4)$ times, the resulting list contains a generating set for $\Weyl(\ket{\phi})$.
The issue that this strategy runs into is that the list also contains many elements outside $\Weyl(\ket{\phi})$, and it is hard to tell them apart without resorting to a brute-force enumeration, leading to the exponential-time algorithm of~\cite{grewal2024improved}.

Stabilizer bootstrapping offers a way around this issue. Instead of trying to find $\Weyl(\ket{\phi})$ in one shot, it makes a new attempt every time it returns to Step 1, with each new attempt having an increased likelihood of success because of the amplification in fidelity in Step 4.

To implement Step 1, we simply run Bell difference sampling many times, and only keep the Pauli operators $P$ for which the correlation $\Tr(P\rho)^2$ \--- which we can estimate by measuring copies of $\rho$ \--- exceeds $1/2$ by a sufficient margin. By the uncertainty principle (see Lemma~\ref{lem:high-corr-commute}), these operators are guaranteed to commute with each other. Furthermore, provided $\Omega(n/\tau^4)$ Bell difference samples are taken, one can ensure that the high-correlation Pauli operators accumulated in this fashion generate nearly all of the mass in $\cB_\rho$ coming from high-correlation Pauli operators (see Definition~\ref{def:high_cor_basis} and Lemma~\ref{lem:high_correlation_space}).

We proceed by a win-win argument. If $F(\rho,\ket{\phi})$ is sufficiently large, then the family of Pauli operators accumulated in this way will generate the stabilizer group of $\ket{\phi}$ (Theorem 6.7 in Ref.~\cite{grewal2024improved}), and we can terminate successfully in Step 2. But if $F(\rho,\ket{\phi})$ is too small, we have no such guarantee. If the family produced in Step 1 is incomplete, however, the upshot is that this certifies that the collection of (nearly) all high-correlation Pauli operators does not generate $\Weyl(\ket{\phi})$, and in particular that there are Pauli operators $P$ in $\Weyl(\ket{\phi})$ whose correlation with $\rho$ is low. By the fact that $B_\rho$ is not too concentrated on any proper subspace of $\Weyl(\ket{\phi})$ (see~\Cref{thm:stab_progress}), this ensures that with $\Omega(\tau^4)$ probability, a sample from $B_\rho$ will be a low-correlation Pauli from $\Weyl(\ket{\phi})$. If we then guess correctly, with probability $1/2$, the sign $\sgn\in\{\pm 1\}$ for which $P\ket{\phi} = \sgn \ket{\phi}$, then we can obtain a projector $\frac{I + \sgn P}{2}$ which has low correlation with $\rho$ and stabilizes $\ket{\phi}$, thus implementing Step 3 successfully with overall probability $\Omega(\tau^4)$ (see Lemma~\ref{lem:sample_low_correlation_Pauli}).

We can then bootstrap by measuring with this in Step 4, amplifying the fidelity between $\ket{\phi}$ and $\rho$ by a constant factor as explained in Section~\ref{sec:recipe}.

As the fidelity can only increase a total of $O(\log 1/\tau)$ times, and each iteration of Step 3 has success probability $\Omega(\tau^4)$, this results in an overall success probability of $\tau^{O(\log 1/\tau)}$ in finding $\Weyl(\ket{\phi})$ and correctly outputting $\ket{\phi}$. We can then repeat this algorithm many times, accumulating a list of estimates that contains $\ket{\phi}$ with high probability and pick out the one with the highest fidelity with $\rho$.

\begin{remark}
    It turns out that the above analysis only uses the fact that $\ket{\phi}$ globally maximizes the fidelity with $\rho$ across stabilizer states in a very weak way, namely that among its \emph{nearest-neighbor stabilizer states}, it has the highest fidelity with $\rho$ (see the proof of Lemma~\ref{lem:progress}). In the above discussion, $\ket{\phi}$ could thus have been replaced with any such ``local maximizer'' with fidelity at least $\tau$ with $\rho$, and the final list of states output by repeating stabilizer bootstrapping $\tau^{O(\log 1/\tau)}$ times will contain all such states. This explains why we are able to get a stronger list-decoding guarantee for local maximizers, rather than just an agnostic tomography guarantee. Moreover, as we show in Section~\ref{sec:bell}, all of these ideas extend naturally to the case of \emph{approximate} local maximizers.
\end{remark}

\subsection{Discrete product states}
\label{sec:discrete_overview}

For discrete product states, the analysis simplifies considerably. Let $\mathcal{K}$ be some set of single-qubit states for which $|\braket{\phi_1|\phi_2}|^2 \le 1 - \mu$ for any distinct $\ket{\phi_1}, \ket{\phi_2}\in\mathcal{K}$. Let $\ket{\phi} = \bigotimes^n_{j=1} \ket{\phi^j}$ denote the product state in $\mathcal{K}^{\otimes n}$ closest to $\rho$, again breaking ties arbitrarily. As before, suppose $F(\rho,\ket{\phi}) \ge \tau$.

Instead of looking for projectors of the form $\frac{I + \sgn P}{2}$ as in the previous section, we look for projectors $\Pi^i_{\ket{\psi^i}}$ for $\ket{\psi^i} \in \mathcal{K}$ and $i\in[n]$, where $\Pi^i_{\ket{\psi}}$ projects the $i$-th qubit in the direction of $\ket{\psi}$ and acts as the identity on all other qubits. In this setting, a family of projectors is complete if for each $i\in[n]$, there is exactly one projector $\Pi^i_{\ket{\psi^i}}$, and moreover $\ket{\psi^i} = \ket{\phi^i}$.

To accumulate a family of high-correlation projectors in Step 1, we perform a simple local optimization: for each $i$ we include in the family the projector $\Pi^i_{\ket{\psi^i}}$ for $\ket{\psi^i}\in\mathcal{K}$ which maximizes the fidelity with the partial trace $\tr_{-i}(\rho)$ of $\rho$ onto the $i$-th qubit. The key property of the family obtained in this fashion is that any projector outside of the family has low correlation with $\rho$ (see Lemma~\ref{lem:unique-high-corr}), because of the fact that the states in $\mathcal{K}$ are well-separated.

If this local optimization already happens to produce a complete family, then we can successfully terminate in Step 2. Of course, this will not in general be the case, but it will happen provided $\tau$ is sufficiently large.

Otherwise, we need to implement Step 3 to sample a low-correlation projector. Because the family produced in Step 1 is incomplete, there exists $i\in[n]$ for which $\ket{\psi^i} \neq \ket{\phi^i}$. Now if we simply guess a random $i\in[n]$ and a random state $\ket{\psi'}\in\mathcal{K}$ distinct from the state $\ket{\psi^i}$ found through local optimization, we have at least a $\frac{1}{n(|\mathcal{K}| - 1)}$ chance of guessing an $i$ for which $\ket{\psi^i} \neq \ket{\phi^i}$ and correctly guessing $\ket{\psi'} = \ket{\phi^i}$. So the resulting guessed projector $\Pi^i_{\ket{\psi'}}$ stabilizes $\ket{\phi}$. Furthermore, $\Pi^i_{\ket{\psi'}}$ has low correlation with $\rho$ for all $\psi'$ by the key property mentioned above.

As in the previous application of stabilizer bootstrapping, we can then bootstrap by measuring with this projector in Step 4, amplifying the fidelity between $\ket{\phi}$ and $\rho$ by a constant factor, in this case scaling with $\Theta(1/\mu)$.

As the fidelity can only increase a total of $O(\log(1/\tau)/\mu)$ times, and each iteration of Step 3 has success probability $\Omega(1/(n|\mathcal{K}|))$, this results in an overall success probability of $(n|\mathcal{K}|)^{-O((1+\log (1/\tau))/\mu)}$ in finding a complete family and correctly outputting $\ket{\phi}$. As before, we can repeat this many times, accumulating a list of estimates that are guaranteed to contain $\ket{\phi}$ with high probability and picking out the one with the highest fidelity with $\rho$.

\subsection{Stabilizer product states}
\label{sec:stabilizer_prod_overview}

In the special case where $\mathcal{K}$ consists of the single-qubit stabilizer states, the approach in the previous section already recovers the previously best runtime of $n^{O(1 + \log(1/\tau))}/\epsilon^2$. In this section, we describe our approach for further improving upon this bound by combining ideas from the previous two sections. Here, let $\ket{\phi} = \bigotimes^n_{j=1} \ket{\phi^j}$ denote the stabilizer product state closest to $\rho$, and suppose $F(\rho,\ket{\phi}) \ge \tau$.

First, instead of parametrizing the projectors as $\Pi^i_{\ket{\psi}}$ as outlined above, we will parametrize them as $\Pi^i_Q$, where $Q \in \pm \{X,Y,Z\}$ is a signed Pauli operator and $\Pi^i_Q$ acts via $\frac{I + Q}{2}$ on the $i$-th qubit and via the identity on all other qubits. The set of such projectors is identical to the set of projectors of the form $\Pi^i_{\ket{\psi}}$ for single-qubit stabilizer states $\ket{\psi}$, but the parametrization in terms of Pauli operators will make it more convenient to draw upon our tools related to Bell difference sampling. In this context, a family of such projectors is complete if for each $i\in[n]$, there is exactly one projector $\Pi^i_{Q^i}$, and furthermore $Q^i\ket{\phi^i} = \ket{\phi^i}$.

To implement Step 1, we still use a local optimization: for each $i\in[n]$, we take $P^i$ to be the Pauli operator in $\{X,Y,Z\}$ for which $\Tr(P^i\tr_{-j}(\rho))^2$ is largest. By the uncertainty principle, this ensures that all projectors outside of the family obtained by local optimization have low correlation with $\rho$.

Provided $\tau$ is sufficiently large, the family obtained by this local optimization is already complete and we can successfully terminate in Step 2.

Otherwise, if the family is incomplete, there is at least one qubit $k\in[n]$ for which $Q^k \neq P^k$. To implement Step 3 and sample a low-correlation projector, we eschew randomly guessing a qubit index in favor of Bell difference sampling. By anti-concentration properties of $B_\rho$ (see Theorem~\ref{thm:prod_progress}), with probability at least $\Omega(\tau^4)$ the resulting sample $R = \bigotimes^n_{j=1} R^j$ will be an operator from $\bigotimes_{1 \le j \le n, j\neq k} \{I, Q^j\}\otimes \{Q^k\}$. In this case, by picking out any index $i$ for which $R^{i} \neq I, P^{i}$, we obtain the correct stabilizer in that qubit up to sign. If we then guess the sign $\sgn\in\{\pm 1\}$ for which $R^{i}\ket{\phi^{i}} = \ket{\phi^{i}}$, we obtain a projector $\Pi^i_{\sgn R^i}$ which stabilizes $\ket{\phi}$. Moreover, it has low correlation according to the uncertainty principle, as mentioned above. 

We can then bootstrap by measuring with this projector in Step 4, thus amplifying the fidelity between $\ket{\phi}$ and $\rho$ by a constant factor.

As before, the fidelity can only increase a total of $O(\log (1/\tau))$ times, and as in our result for stabilizer states, each iteration of Step 3 has success probability $\Omega(\tau^4)$, resulting in an overall success probability of $\tau^{O(\log (1/\tau))}$ in finding a complete family of projectors and correctly outputting $\ket{\phi}$. As in all our applications, we can repeat this many times, accumulating a list of estimates that is guaranteed to contain $\ket{\phi}$ with high probability and pick out the one with the highest fidelity with $\rho$.

\subsection{States with high stabilizer dimension}
\label{sec:dim_overview}
We now turn to a more general class of states, namely those with high stabilizer dimension. Let $\sigma$ denote the (possibly mixed) state with stabilizer dimension at least $n-t$ which is closest to $\rho$, and suppose $F(\rho, \sigma)\ge \tau$. Similar to the case of stabilizer states (that is, $t=0$), the goal is to find the stabilizer group of $\sigma$. Although $\Weyl(\sigma)$ is not a complete stabilizer group anymore, we can still find the optimal $\sigma$ via state tomography once $\Weyl(\sigma)$ is pinned down (see \Cref{sec: reduce to find Clifford}). 

The algorithm for finding $\Weyl(\sigma)$ follows the stabilizer bootstrapping recipe, with some key differences. We first describe the aspects that are similar to our previous applications of the recipe. First, Step 1 works as in the stabilizer state setting: we run Bell difference sampling many times and select enough high-correlation Pauli operators that they generate a stabilizer group $H$ containing most of the mass of $B_\rho$ coming from high-correlation Pauli operators. We then proceed with a similar win-win argument as follows: (1) If the dimension of $H\cap \Weyl(\sigma)$ is large (say, at least $n-t'$), we directly obtain a desired output and terminate successfully in Step 2. (2) Otherwise if the dimension of $H\cap \Weyl(\sigma)$ is small (less than $n-t'$), we can sample a low-correlation projector that stabilizes $\sigma$ with a non-negligible probability in Step 3, and successfully amplify the fidelity in Step 4.

Here is the key challenge in the high stabilizer dimension setting. In the $t=0$ case, Step 3 relies on the anti-concentration property of $\cB_\rho$, i.e., $\cB_\rho$ is not too concentrated on \emph{any} proper subspace of $\Weyl(\sigma)$. But for $t > 0$, we do not have such a strong guarantee and instead have to make do with a weaker anti-concentration property, namely that $B_\rho$ places at most $1/2^n$ mass on any individual Pauli operator (\Cref{lem: Bell distribution small}). Because of this, we have to choose $t'=O(t+\log(1/\tau))$ slightly larger than $t$ to make Step 3 work (\Cref{lem:weaker_evenly_distribution}). The inadequacy of Step 3 poses a tough challenge for Step 2: When $\dim(H\cap \Weyl(\sigma))=n-t'$, $H$ does not even contain the full information about $\Weyl(\sigma)$, but we have to find $\Weyl(\sigma)$ despite the lack of information. This is the main difficulty beyond the $t=0$ case. From now on, we focus on how Step 2 works, as the other parts of the algorithm are similar to the $t=0$ case. The main guarantee is Lemma~\ref{lem: high dimension step 2}, the proof details of which we sketch below.

Since $\dim(H\cap \Weyl(\sigma))\ge n-t'$, by measuring $\rho$ in the joint eigenbasis of $H$, we can find an $(n-t')$-dimensional subspace of $\Weyl(\sigma)$ (see~\Cref{lem:find_heavy_subspace_S}). For the sake of demonstration, let's assume here that $H=\{I, Z\}^{\otimes n}$ and $H\cap \Weyl(\sigma)=\{I, Z\}^{\otimes n-t'}\otimes I^{\otimes t'}$. Then $\sigma$ has the form $\ketbra{z}\otimes \sigma_0$ for some $z\in \{0, 1\}^{n-t}$. Since $F(\rho, \sigma)$ is large, measuring $\rho$ in the computational basis yields many bit-strings starting with $z$. The affine span of these bit-strings is a subspace of $0^{n-t'}\otimes \{0, 1\}^{t'}$ (the affine span of a set $A$ is defined as the span of $A-A$). In addition, if the number of samples is large enough, the affine span should approximate the whole space $0^{n-t'}\otimes \{0, 1\}^{t'}$. Then the orthogonal space of the affine span is roughly $\{0, 1\}^{n-t'}\otimes 0^{t'}$, which is just $H\cap \Weyl(\sigma)$. In \Cref{lem: find the intersection of subspaces} we make these ideas rigorous. 

Having obtained an $(n-t')$-dimensional subspace of $\Weyl(\sigma)$ in this fashion, we now explain how to conclude the argument. Again assume for simplicity that this subspace is just $\{I, Z\}^{\otimes n-t'}\otimes I^{\otimes t'}$ and $\sigma=\ketbra{z}\otimes \sigma_0$. By measuring the first $n-t'$ qubits of $\rho$ in the computational basis, we learn $z$ with probability at least $\tau$. Now the problem of finding $\Weyl(\sigma)$ from $\rho$ reduces to finding $\Weyl(\sigma_0)$ from $\rho_0\triangleq \braket{z|\rho|z}/\tr(\braket{z|\rho|z})$. In other words, we have reduced a $n$-qubit problem to a $t'=O(t+\log(1/\tau))$-qubit problem. We then apply stabilizer bootstrapping to $\rho_0$, but with the Bell difference sampling replaced by uniformly random sampling of Pauli strings. The uniform distribution over Pauli strings is certainly evenly distributed in $\Weyl(\sigma_0)$ and thus has the necessary anti-concentration for Step 3. One catch is that it places exponentially (in $t'$) small mass on $\Weyl(\sigma_0)$, but exponential dependence in $t'$ is something we can afford because $t'$ is small. The details of this final step can be found in Lemma~\ref{lem: high stabilizer exponential time}.

\section{Related works}
\label{sec:related}

Most closely related to the present work are the aforementioned works of Grewal, Iyer, Kretschmer, and Liang~\cite{grewal2024agnostic,grewal2024improved}. Here we mention some other relevant works.

\vspace{0.5em}\noindent\textbf{Simulating and learning near-stabilizer states.} \enspace Simulating stabilizer states and, more generally, quantum states that are preparable by \emph{$t$-doped circuits} (quantum circuits with Clifford gates and at most $t$ of non-Clifford gates) has been widely explored~\cite{aaronson2004improved,bravyi2016improved,rall2019simulation,bravyi2019simulation,qassim2021improved}. The runtime and sample complexity of these algorithms scale polynomially in the system size $n$ and exponentially in the number of non-Clifford gates, so that $O(\log n)$-doped circuits can be classically efficiently simulated. 

On the learning front, Aaronson and Gottesman~\cite{aaronson2008identifying} proposed polynomial-time learning algorithms that either use a quadratic number of single-copy measurements or use a joint measurement on $O(n)$ copies. Montanaro~\cite{montanaro2017learning} subsequently gave a complete proof that $O(n)$ measurements are required if one can perform two-copy measurements, via a procedure called Bell difference sampling (see Section~\ref{sec:bell}). 

For $t$-doped circuits, Ref.~\cite{lai2022learning} gave an algorithm that, given an $n$-qubit quantum circuit consisting of \emph{one} layer of $t=O(\log n)$ T gates, outputs a circuit that is equivalent to the unknown circuit if the input state is $\ket{0}^{\otimes n}$. Given oracle access to the underlying circuit, an algorithm is also known for learning quantum circuits comprised of Clifford gates and a few T gates~\cite{leone2022learning}. For general $t$-doped states, Refs.~\cite{grewal2023efficient,leone2024learning,chia2024efficient,hangleiter2024bell} gave learning algorithms running in time $\mathrm{poly}(n,2^t)$, similar to the computational cost of simulation, in various closely related regimes.

\cite{arunachalam2024toleranttestingstabilizerstates} gave an algorithm for tolerant testing of stabilizer states that can decide whether a given state has stabilizer fidelity at least $\tau$ or at most $2^{-\mathrm{poly}(1/\tau)}$, improving on prior (non-tolerant) stabilizer testing algorithms~\cite{gross2021schur,grewal2024improved}. The sample complexity and runtime of the algorithm are $\poly(1/\tau)$ and $n\cdot \poly(1/\tau)$ respectively. Our algorithm for agnostic learning implies an algorithm for tolerant testing with an incomparable guarantee: whereas we can distinguish between stabilizer fidelity at least $\tau$ or at most $\tau - \epsilon$ for any $0 < \epsilon \le \tau \le 1$, our sample complexity depends on $n$, and our runtime scales quasipolynomially in $1/\tau$. After the original version of the present manuscript was made available online, three very recent works~\cite{arunachalam2024note,bao2024tolerant,mehraban2024improvedboundstestinglow} proposed a polynomial time algorithm that tests whether a pure state $\ket{\psi}$ has at most $\tau$ or at least $\poly(\tau)$ fidelity with a stabilizer state when $\tau=1/\poly(n)$.

\vspace{0.5em}\noindent\textbf{Shadow tomography.} \enspace As mentioned at the outset, quantum tomography suffers from an unavoidable exponential scaling in sample complexity, and realizable learning and agnostic tomography offer ways of circumventing this exponential scaling by considering more structured classes of quantum states. 

An alternative approach to circumventing this exponential scaling is through \emph{shadow tomography}, originally proposed by Aaronson~\cite{aaronson2018shadow}, where one only needs to approximate the expectation values of $m$ observables of the unknown state. Interestingly, it has been observed~\cite{buadescu2019quantum,anshu2023survey} that agnostic tomography (in fidelity) can be directly reduced to shadow tomography: take the observables to be the states in $\mathcal{C}$, run shadow tomography, and output the one with highest estimated correlation with $\rho$.

The literature on shadow tomography has almost exclusively focused on sample complexity. There has been a long line of works based on online learning which achieve $\poly(\log m, n,1/\epsilon)$ sample complexity~\cite{aaronson2019gentle,buadescu2021improved,brandao2019quantum,gong2023learning,watts2024quantum,chen2024optimalshadow} using highly entangled measurements. Because shadow tomography is known to be sample-efficient, with the best known sample complexity upper bound scaling as $O(n\log^2 m / \epsilon^4)$ to estimate $m$ observables~\cite{buadescu2021improved,watts2024quantum}, this immediately implies that agnostic tomography is sample-efficient.

In settings of shadow tomography where one can only make single-copy measurements, the classical shadows protocol of Huang, Kueng, and Preskill~\cite{huangPredictingManyProperties2020} uses $O(2^n\log m/\epsilon^2)$ single-copy measurements, and this bound is known to be tight for single-copy measurements~\cite{chen2022exponential}. This protocol is both statistically and computationally efficient in the special case where the observables have bounded Frobenius norm or are local. When the observables are Pauli observables, Bell sampling gives a statistically and computationally efficient protocol~\cite{huang2021information} which is essentially sample-optimal~\cite{chen2022exponential,chen2024optimal}. Beyond these simple settings, however, the computational complexity of shadow tomography remains a challenging open question. Additionally, the above reduction from agnostic tomography to shadow tomography incurs a polynomial overhead in the size of $|\mathcal{C}|$. These obstacles make it challenging to get efficient algorithms for agnostic tomography through shadow tomography.

\vspace{0.5em}\noindent\textbf{Magic estimation.} \enspace In quantum resource theory, a \emph{magic monotone} is a function of a quantum state which does not increase if the state is transformed via completely positive trace-preserving maps that preserve the convex hull of stabilizer states. These are used to quantify the ``non-stabilizerness'' of quantum states. A number of different magic monotones have been proposed like stabilizer rank, stabilizer nullity, stabilizer extent, (inverse) stabilizer fidelity, Pauli rank, mana, and robustness of magic~\cite{bravyi2019simulation,gu2024pseudomagic,veitch2014resource,howard2017application,beverland2020lower,bu2019efficient,liu2022many,haug2024probing,bansal2024pseudorandom} but have not been regarded as experimentally accessible given the dearth of efficient protocols for estimating them. 

\emph{Stabilizer R\'{e}nyi entropies} have recently been proposed by Leone, Oliviero, and Hamma as an experimentally friendly proxy for non-stabilizerness~\cite{leone2022stabilizer}. These are known to be genuine magic monotones for R\'{e}nyi index $\alpha \ge 2$ and non-monotone for $\alpha < 2$~\cite{haug2023stabilizer,leone2024stabilizer}. Unfortunately, the complexity of estimation can still scale exponentially in system size~\cite{oliviero2022measuring} or require auxiliary information like conjugate state access or assumptions like a tensor network ansatz or odd R\'{e}nyi index~\cite{haug2024efficient,tarabunga2024nonstabilizerness}. Even under these assumptions, these protocols only achieve an additive approximation to $2^{(1-\alpha)H_\alpha}$, where $H_\alpha$ is the stabilizer entropy, which can be of the same order as the system size.

\emph{Additive Bell magic} is another measure of magic, proposed by Haug and Kim~\cite{haug2023scalable}, that is computationally efficient to estimate via Bell difference sampling, even in practice as was done on IonQ's 11-qubit quantum computer~\cite{haug2023scalable} and in Rydberg atom arrays for system size $12$ recently by Bluvstein et al~\cite{bluvstein2024logical}. However, as with stabilizer entropy, the estimation procedure can scale super-polynomially in the system size because Bell difference sampling is meant to achieve an additive approximation to $2^{-B}$, where $B$ is the additive Bell magic, which can be of the same order as the system size. Furthermore, additive Bell magic is not known to be a genuine magic monotone.

By our statistical lower bound in Lemma~\ref{thm:info_lower_informal}, $\log(1/F_{\cS})$ runs into the same issue as $H_\alpha, B$: the former can potentially scale super-logarithmically or even linearly in the system size, and for such states, estimation necessarily requires super-polynomially many measurements. Our positive results instead give efficient algorithms in the regime where $\log(1/F_{\cS})$ scales slightly sub-logarithmically in the system size. 

Because $\log(1/F_{\cS})$ and $B$ do not generally upper or lower bound each other, the regime in which we are able to efficiently estimate $\log(1/F_{\cS})$ by Theorem~\ref{thm:magic_est_informal} does not necessarily coincide with the regime where we can efficiently estimate $B$. On the other hand, $\log(1/F_{\cS})$ and $H_\alpha$ can be bounded in terms of each other~\cite{haug2024efficient}, so the latter can be efficiently estimated when the former can be. That said, an approximation to the latter gives very little information about the former, as known bounds on $\log(1/F_{\cS})$ in terms of $H$ are quite loose. Additionally, stabilizer fidelity is a natural monotone to target in its own right given its appealing operational interpretation. Indeed, it gives the possibility of not just quantifying non-stabilizerness, but exhibiting a stabilizer state that \emph{witnesses} the level of non-stabilizerness.

Finally, we note that \emph{stabilizer nullity} (i.e. $n$ minus the stabilizer dimension) is a magic monotone which is efficiently computable in the sense that it is possible to test whether a state has a certain stabilizer nullity or is far from any such state~\cite{grewal2023efficient}. However, stabilizer nullity is rather brittle: even if a state has low stabilizer nullity, if it undergoes a small amount of noise, this need no longer be the case.

\section{Preliminaries}\label{sec:prelim}
In this section, we provide the basic concepts and results required throughout this paper.

\vspace{0.5em} \noindent \textbf{General notational conventions.} \enspace We use $[n]$ to denote the set $\{1,\ldots,n\}$. We use $\norm{\cdot}_\infty$ to
represent the operator norm for a matrix and the infinity norm for a vector, and use $\norm{\cdot}_1$ to denote the $L_1$ norm for a vector. We use $I_n$ to denote the $2^n\times 2^n$ identity matrix, omitting the subscript $n$ when it is clear from context. The trace operator is denoted by $\tr$. When we say “with high probability” without specification, we mean with probability at least $2/3$. Given a distribution $\cD$ over a domain $\Omega$ and a subset $S\subseteq \Omega$, we use $\cD(S)$ to denote $\Pr_{x\sim\cD}[x\in S]$. We use standard big-O notation ($O, \Omega, o, \omega$) throughout.

\subsection{Quantum states and measurements}
An \emph{$n$-qubit quantum state} is a positive semi-definite Hermitian operator of trace one on $(\mathbb{C}^2)^{\otimes n}\simeq\mathbb{C}^{2^n}$.
The set of $n$-qubit states is denoted by $D(\mathbb{C}^{2^n})$.
We use the standard bra-ket notation where $|\psi\rangle$ is the vector described by $\psi$ and $\langle\psi|=|\psi\rangle^\dagger$.
The standard basis (computational basis) of $\mathbb{C}^{2^n}$ is $\{|s\rangle:s\in\{0,1\}^n\}$.
A state $\rho\in D(\mathbb{C}^{2^n})$ is \emph{pure} if it is rank $1$.
That is, there exists a unit vector $|\psi\rangle\in\mathbb{C}^{2^n}$ so that $\rho=|\psi\rangle\!\langle\psi|$.
As a result, the set of pure states can be identified with the projective space on $\mathbb{C}^{2^n}$ (denoted by $\mathbb{CP}^{2^n-1}$) by $\rho\rightarrow|\psi\rangle$.
For the ease of notations, we take the convention that when we are considering $n$-qubit states but what $\psi$ actually describes is a vector in $\mathbb{C}^{2^{n-t}}$, $|\psi\rangle$ should be viewed as $|\psi\rangle\otimes I_t$, and similarly $\langle\psi|$ should be viewed as $\langle\psi|\otimes I_t$.

An important map on $n$-qubit quantum states is the \emph{partial trace}.
For $j\in[n]$, the operator $\tr_{-j}$ is a map $D((\mathbb{C}^2)^{\otimes n})\rightarrow D(\mathbb{C}^2)$ defined as
\[\tr_{-j}=\underbrace{\tr\otimes\cdots\otimes\tr}_{j-1~\tr}\otimes\operatorname{id}\otimes\underbrace{\tr\otimes\cdots\otimes\tr}_{n-j~\tr},\]
where $\operatorname{id}$ is the identity map.
Similarly, the operator $\tr_{>n-j}$ is a map $D((\mathbb{C}^2)^{\otimes n})\rightarrow D((\mathbb{C}^2)^{\otimes n-j})$ defined as
\[\tr_{>n-j}=\underbrace{\operatorname{id}\otimes\cdots\otimes\operatorname{id}}_{n-j~\operatorname{id}}\otimes\underbrace{\tr\otimes\cdots\otimes\tr}_{j~\tr}.\]

A \emph{projector-valued measure (PVM)} is a set of projection operators $\{\Pi_i\}_i$ on $\mathbb{C}^{2^n}$ that sums to $I_n$.
They define measurements on $n$-qubit quantum states.
By Born's rule, measuring PVM $\{\Pi_i\}$ on $\rho$ yields outcome $i$ with probability $\tr(\Pi_i\rho)$.
The post-measurement state after obtaining result $i$ is $\Pi_i\rho\Pi_i/\tr(\Pi_i\rho)$.

\subsection{Geometry of quantum states}
\label{sec:geometry} 

For two quantum states $\rho,\sigma\in D(\mathbb{C}^{2^n})$, we use $D_{\tr}(\rho,\sigma)$ to denote the \emph{trace distance} between $\rho$ and $\sigma$, i.e. a half of the trace norm of the operator $\rho - \sigma$. 

\begin{definition}[Fidelity of quantum states]
    We use $F(\rho,\sigma)=F(\sigma,\rho)=(\tr\sqrt{\sqrt{\rho}\sigma\sqrt{\rho}})^2$ to denote the \emph{fidelity} between $\rho$ and $\sigma$.
    We have $0\leq F(\rho,\sigma)\leq1$ and $F(\rho,\sigma)=1$ if and only if $\rho=\sigma$.
    When $\sigma=|\psi\rangle\!\langle\psi|$ is pure, the fidelity simplifies into $F(\rho,|\psi\rangle)=\langle\psi|\rho|\psi\rangle$.
    For $\mathcal{C}\subseteq D(\mathbb{C}^{2^n})$, denote $F_\mathcal{C}(\rho)\triangleq \max_{\sigma\in \mathcal{C}}F(\rho, \sigma)$ (throughout the paper, $\mathcal{C}$ is compact so the maximum exists). We also use interchangeably the notation $F(\rho,\mathcal{C}) \triangleq F_{\mathcal{C}}(\rho)$.
\end{definition}

The following property follows directly from the definition, and we defer the proof to~\Cref{app:defer_fidelity}.

\begin{lemma}\label{lem: fidelity high dimension}
    Let $t\in \mathbb{N}$, $s\in \{0, 1\}^{n-t}$ and $\rho, \sigma=\ketbra{s}\otimes \sigma_0$ be two $n$-qubit states. Denote $\rho_s=\braket{s|\rho|s}/\tr(\braket{s|\rho|s})$. We have
    \begin{equation}\label{eq: fidelity of rho and sigma}
    F(\rho,\sigma)=\tr(\braket{s|\rho|s})F(\rho_s, \sigma_0)\leq \tr(\braket{s|\rho|s}).
    \end{equation}
    The equality holds when $\sigma_0=\rho_s$.
\end{lemma}

\noindent We equip $D(\mathbb{C}^{2^n})$ with the \emph{Bures metric} $A(\rho,\sigma)=\arccos\sqrt{F(\rho,\sigma)}$.
It is a metric and in particular satisfies the triangle inequality $A(\rho,\nu)\leq A(\rho,\sigma)+A(\sigma,\nu)$~\cite{nielsen2010quantum}.
The induced metric (also called the \emph{Fubini-Study metric}) on the set of pure states simplifies into $A(|\psi\rangle,|\phi\rangle)=\arccos|\langle\psi|\phi\rangle|$.

\begin{definition}[$\mu$-packing set]
    For $0<\mu\leq1$, we call a set of single-qubit pure states $\mathcal{K}\subseteq\mathbb{CP}^1$ a $\mu$-packing set if for any two distinct states $\ket{\psi},\ket{\phi}\in\mathcal{K}$, $F(\ket{\psi},\ket{\phi})\leq1-\mu$ (i.e., $A(\ket{\psi},\ket{\phi})\geq\arcsin\sqrt{\mu}$).
    We assume that measuring the PVM $\{|\phi\rangle\langle\phi|,I-|\phi\rangle\langle\phi|\}$ takes $O(1)$ time for $\forall |\phi\rangle\in\mathcal{K}$.
\end{definition}

\begin{lemma}[Cardinality of $\mu$-packing set]\label{lem:packing}
For $0<\mu\leq1$, if $\mathcal{K}$ is a \emph{$\mu$-packing set}, then $|\mathcal{K}|=O(1/\mu)$.
\end{lemma}

\noindent The proof is standard and we defer it to Appendix~\ref{app:defer_packing}.

\subsection{Pauli operators, Weyl operators, and Clifford gates}
The \emph{Pauli matrices} are defined as \begin{equation*}
    I=\begin{pmatrix}1&0\\0&1\end{pmatrix},X=\begin{pmatrix}0&1\\1&0\end{pmatrix},Y=\begin{pmatrix}0&-i\\i&0\end{pmatrix},Z=\begin{pmatrix}1&0\\0&-1\end{pmatrix}\,.
\end{equation*}
Elements in $\{\pm1,\pm i\}\times\{I,X,Y,Z\}^{\otimes n}$ are called \emph{Pauli operators}.
They form the $n$-qubit Pauli group $\mathcal{P}^n$ via matrix multiplication.

\begin{definition}[Weyl operators and Pauli strings]
    For $x=(a,b)\in\mathbb{F}_2^{2n}$, define its corresponding \emph{Weyl operator} as $W_x\triangleq i^{a\cdot b}\bigotimes_{j=1}^n(X^{a_j}Z^{b_j})$, where the inner product $a\cdot b$ is performed in $\mathbb{Z}$ instead of $\mathbb{F}_2$. We will refer to $x\in\mathbb{F}_2^{2n}$ as \emph{Pauli strings}.
    We will sometimes use Pauli strings and their corresponding Weyl operators interchangeably.
    
    When combined with a sign $\sgn\in\{\pm1\}$, we call $(\sgn, x)$ a \emph{signed Pauli string}. To any such signed Pauli string, we associate the projector 
    \begin{equation}
        \Pi^\sgn_x \triangleq \frac{I + \sgn W_x}{2}\,. \label{eq:proj_def}
    \end{equation}
\end{definition} 
\noindent The Weyl operators form an orthogonal basis (with respect to the Hilbert-Schmidt inner product) for the space of $2^n\times2^n$ Hermitian matrices.
We refer to $\tr(W_x\rho)^2$ as the \emph{correlation} of $W_x$ (or $x$) with respect to $\rho$. 

The mapping from a Weyl operator $W_x$ to its associated Pauli string $x$ is an isomorphism between the quotient group $\mathcal{P}^n/\{\pm1,\pm i\}$ and the symplectic vector space $\mathbb{F}_2^{2n}$.
The isomorphism gives a natural meaning to the \emph{symplectic inner product}:
\begin{definition}[Symplectic inner product]
    Given $x,y\in\mathbb{F}^{2n}_2$, define the \emph{symplectic inner product} $\langle x,y\rangle$ as follows: $\langle x,y\rangle=0$ if $W_x$ and $W_y$ commute, and $\langle x,y\rangle=1$ if $W_x$ and $W_y$ anti-commute.
\end{definition}
\noindent Henceforth, given Pauli strings $x,y\in\mathbb{F}^{2n}_2$, the notation $\langle \cdot, \cdot\rangle$ will always denote the symplectic inner product instead of the ``standard'' inner product over $\mathbb{F}^{2n}_2$.
We say a subspace $V\subseteq\mathbb{F}_2^{2n}$ is isotropic if $\forall x,y\in V$, $\langle x,y\rangle=0$.
The maximum dimension of an isotropic subspace is $n$.
We call such subspaces \emph{stabilizer families} (more commonly called Lagrangian subspaces).

The following kind of uncertainty principle is well known, see for example Theorem 1 of Ref.~\cite{asadian2016heisenberg}.
\begin{lemma}\label{lem:uncertain}
For $x,y\in\mathbb{F}_2^{2n}$, if $\langle x,y\rangle=1$, then for any $n$-qubit state $\rho$, $\tr(W_x\rho)^2+\tr(W_y\rho)^2\leq1$.
\end{lemma}

\begin{proof}
Consider observable $O=\tr(W_x\rho)W_x+\tr(W_y\rho)W_y$.
We have $\tr(O\rho)=\tr(W_x\rho)^2+\tr(W_y\rho)^2$, and $\tr(O^2\rho)=\tr(W_x\rho)^2+\tr(W_y\rho)^2$.
Since the variance of $O$ is non-negative, we have $\tr(O^2\rho)-\tr(O\rho)^2\geq0$, so $\tr(W_x\rho)^2+\tr(W_y\rho)^2\leq1$.
\end{proof}

\begin{lemma}\label{lem:high-corr-commute}
Given an $n$-qubit state $\rho$, elements in the set $\{x\in\mathbb{F}_2^{2n}:\tr(W_x\rho)^2>\frac{1}{2}\}$ are pairwise commuting.
\end{lemma}

\begin{proof}
Suppose two elements in the set are anti-commuting, then this contradicts~\Cref{lem:uncertain}.
\end{proof}

\begin{definition}
    The \emph{Clifford group} is the normalizer of the Pauli group.
    Its elements, the \emph{Clifford gates}, are unitaries $G$ over $\mathbb{C}^{2^n}$ such that $\forall x\in\mathbb{F}_2^{2n}$, $GW_xG^\dagger=\sgn W_y$ for a unique $\sgn\in\{\pm1\}$ and $y\in\mathbb{F}_2^{2n}$.
    As a result, we define the action of $C$ on $x$ as $C(x)=y$.
    Similarly, for $S\subseteq\mathbb{F}_2^{2n}$, we define $C(S)=\{C(x):x\in S\}$.
\end{definition}

\noindent It is known that any Clifford gate can be decomposed into a sequence of single-qubit and two-qubit Clifford gates~\cite{calderbank1997quantum}.
We call such a sequence a \emph{Clifford circuit}.
A \emph{$t$-doped state} is a state obtained from $|0^n\rangle$ by applying Clifford gates and at most $t$ non-Clifford single-qubit gates (unitaries over $\mathbb{C}^2$).

\subsection{Stabilizer states and optimization landscape}
Fix an $n$-qubit state $\rho\in D(\mathbb{C}^{2^n})$.
We say a projector $\Pi$ over $\mathbb{C}^{2^n}$ \emph{stabilizes} $\rho$ if $\Pi\rho\Pi=\rho$.
We say a Weyl operator $W_x$ and also the Pauli string $x$ is a \emph{stabilizer of $\rho$} if $W_x\rho W_x=\rho$.
We say a signed Weyl operator $\sgn W_x$ and its associated signed Pauli string $(\sgn,x)$ stabilize $\rho$ if $\sgn W_x\rho=\rho$. Given a set of Weyl operators or Pauli strings $S$, let $\Stab(S)$ denote the set of mixed states stabilized by all elements in $S$.

\begin{definition}[Stabilizer states and stabilizer dimension]
    Let $\Weyl(\rho)=\{x\in\mathbb{F}_2^{2n}:W_x\rho W_x=\rho\}$ denote the \emph{stabilizer group of $\rho$}. This is an isotropic subspace of $\mathbb{F}_2^{2n}$.
    The \emph{stabilizer dimension} of $\rho$ is the dimension of $\Weyl(\rho)$.
    In particular, a \emph{stabilizer state} is a pure state with stabilizer dimension $n$.
    Alternatively, stabilizer states are the states that can be prepared by applying a Clifford gate to the state $|0^n\rangle$.
    We will use $\cS$ to denote the set of $n$-qubit stabilizer states, and use $\cS_a^b$ to denote the set of $a$-qubit states with stabilizer dimension at least $b$. We often omit the subscript $a$ when $a=n$. 
\end{definition}

\noindent The set of single-qubit stabilizer states is $\cS_1=\{\ket{0}, \ket{1}, \ket{+}, \ket{-}, \ket{+i}, \ket{-i}\}$.
They are the $\pm1$ eigenvectors of $X$, $Y$ and $Z$, respectively.
Define $\cS\cP=\cS_1^{\otimes n}$ to be the set of $n$-qubit \emph{stabilizer product states}. Recalling the notation from Section~\ref{sec:geometry}, $F_\mathcal{S}(\rho)$ is called the \emph{stabilizer fidelity} of $\rho$.

It is known that for $|\phi\rangle,|\phi'\rangle\in\mathcal{S}$, the highest $\neq1$ value of $|\langle\phi|\phi'\rangle|$ is $\frac{1}{\sqrt{2}}$ (see Corollary 3 of Ref.~\cite{garcia2014geometry}).
Moreover, stabilizer states that are nearest neighbors of each other have the following relation:

\begin{lemma}[Relation between nearest neighbor stabilizer states, Theorem 13 of Ref.~\cite{garcia2014geometry}]\label{lem:NN_stab}
For stabilizer state $|\phi\rangle$, the stabilizer states $|\phi'\rangle$ with $|\langle\phi|\phi'\rangle|=\frac{1}{\sqrt{2}}$ are of the form $|\phi'\rangle=\frac{I+i^\ell W_x}{\sqrt{2}}|\phi\rangle$ for some $x\in\mathbb{F}_2^{2n}$ and some integer $\ell$.
\end{lemma}

\noindent In this work, we develop tools for understanding the optimization landscape of fidelity with respect to stabilizer states. To this end, we consider the following notion:

\begin{definition}[$\gamma$-approximate local maximizer]\label{def:localmax}
Fix an $n$-qubit state $\rho$.
For $\gamma>0$, a stabilizer state $|\phi\rangle\in\mathcal{S}$ is a \emph{$\gamma$-approximate local maximizer of fidelity with $\rho$} if
\[\langle\phi|\rho|\phi\rangle\geq\gamma\max_{\substack{|\phi'\rangle\in\mathcal{S},\\|\langle\phi'|\phi\rangle|=\frac{1}{\sqrt{2}}}}\langle\phi'|\rho|\phi'\rangle.\]
That is, $|\phi\rangle$ approximately maximizes fidelity over its nearest neighbors $|\phi'\rangle$ in $\mathcal{S}$.
\end{definition}

\noindent In particular, observe that the global maximizer of stabilizer fidelity $\arg\max_{|\phi\rangle\in\mathcal{S}}\langle\phi|\rho|\phi\rangle$ is a $1$-approximate local maximizer of fidelity with $\rho$ over $\mathcal{S}$.

\subsection{Bell measurement and Bell difference sampling}
\label{sec:bell}

Given $n$, the \emph{Bell state} is defined as $|\Omega\rangle=\frac{1}{\sqrt{2^n}}\sum_{x\in\mathbb{F}_2^n}|x\rangle|x\rangle$.
\emph{Bell measurement} refers to the measurement in the Bell basis given by $\{|\Psi_x\rangle\triangleq (W_x\otimes I)|\Omega\rangle:x\in\mathbb{F}_2^{2n}\}$, an orthonormal basis for $\mathbb{C}^{2^n}\otimes\mathbb{C}^{2^n}$.
We can think of the output of Bell measurement as an element in $\mathbb{F}_2^{2n}$.

\begin{definition}[Bell difference sampling]
    Given an $n$-qubit state $\rho$, \emph{Bell difference sampling}~\cite{montanaro2017learning,gross2021schur} refers to the process of performing Bell measurement on $2$ copies of $\rho$, then performing another Bell measurement on $2$ new copies of $\rho$, and outputting the sum of the two outcomes.
    Let $\cB_\rho$ be the distribution of Bell difference sampling on $\rho$.
\end{definition}

\begin{lemma}[Bell difference sampling, Ref.~\cite{gross2021schur}]\label{lem:BDS}
Given $4$ copies of $\rho$, Bell difference sampling uses $2$-copy measurements to yield a sample from the distribution which places mass
\[\cB_\rho(x):=\frac{1}{4^n}\sum_{a\in\mathbb{F}_2^{2n}}(-1)^{\langle x,a\rangle}\tr(W_a\rho)^4\] on $x$, for every Pauli string $x$.
\end{lemma}

\begin{proof}
Eq.~(3.7) of Ref.~\cite{gross2021schur} proves that the POVM of Bell difference sampling is $\{\Pi_x\}_{x\in\mathbb{F}_2^{2n}}$, where
\begin{equation*}
\Pi_x=\frac{1}{4^n}\sum_{a\in\mathbb{F}_2^{2n}}(-1)^{\langle x,a\rangle}W_a^{\otimes4}.
\end{equation*}
Hence $\cB_\rho(x)=\tr(\Pi_x\rho^{\otimes4})=\frac{1}{4^n}\sum_{a\in\mathbb{F}_2^{2n}}(-1)^{\langle x,a\rangle}\tr(W_a\rho)^4.$
\end{proof}

\subsection{Subroutines}
\label{sec:subroutines}

In this section, we compile various known guarantees for basic subroutines that are used in our algorithms. Lemmas~\ref{lem: full tomography} and~\ref{lem:classical_shadow} come directly from prior work. The proofs of correctness for the other lemmas are also standard, and we include them for the sake of completeness in Appendix~\ref{app:subroutinesproof}.

For our algorithms for learning states with high stabilizer dimension, we will need to apply full tomography to suitably chosen subsystems.
For our task, standard full tomography with single-copy measurements suffices.
It has the following guarantees:

\begin{lemma}[Full tomography via single-copy measurements~\cite{gutaFastStateTomography2020}]\label{lem: full tomography}
    Given copies of an $n$-qubit quantum state $\rho$, there is an algorithm that outputs a density matrix $\widehat{\rho}$ such that $D_{\tr}(\rho,\widehat{\rho})\leq \epsilon$ with probability at least $1-\delta$. The algorithm performs $O(2^{4n}n\log(1/\delta)/\epsilon^2)$ single-copy measurements on $\rho$ and takes $O(2^{4n}n^2\log(1/\delta)/\epsilon^2)$ time.
\end{lemma}

\noindent We will use an assortment of algorithms for estimating fidelities between an unknown state and a collection of known pure states. For learning stabilizer states and stabilizer product states, we use the following simple consequence of the classical shadows protocol of~\cite{huangPredictingManyProperties2020}.

\begin{lemma}[Estimating fidelities via classical shadows~\cite{huangPredictingManyProperties2020}]\label{lem:classical_shadow}
Given an $n$-qubit quantum state $\rho$ and $M$ stabilizer states $|\phi_1\rangle,\ldots,|\phi_M\rangle$, there is an algorithm that, with probability at least $1-\delta$, estimates $\langle\phi_i|\rho|\phi_i\rangle$ to additive error at most $\epsilon$ for all $i$.
The algorithm only uses single-copy measurements.
The sample complexity is $O\left(\frac{1}{\epsilon^2}\log\frac{M}{\delta}\right)$ and the time complexity is $O\left(\frac{M}{\epsilon^2}n^2\log\frac{M}{\delta}\right)$.
\end{lemma}

For our application to learning states with high stabilizer dimension, we need to slightly extend Lemma~\ref{lem:classical_shadow}:

\begin{lemma}[Estimating fidelities for states with high stabilizer dimension via classical shadows]\label{subroutine: fidelity high dimension}
    Given $t\in\mathbb{N}$ an $n$-qubit quantum state $\rho$, and $M$ Clifford unitaries $C_1, \cdots, C_M$, there exists an algorithm that, with probability at least $1-\delta$, estimates $\mathrm{tr}(\langle0^{n-t}|C_i^\dagger\rho C_i|0^{n-t}\rangle)$ to additive error at most $\epsilon$ for all $i$. The algorithm only uses single-copy measurements. The sample complexity is $O\left(\frac{2^{2t}}{\epsilon^2}\log\frac{2^{t}M}{\delta}\right)$ and the time complexity is $O\left(\frac{2^{3t}M}{\epsilon^2}n^2\log\frac{2^{t}M}{\delta}\right)$.
\end{lemma}

\noindent Finally, for our applications to learning discrete product states, we can obtain improved polynomial dependence on $n$ using the following fidelity estimation result:

\begin{lemma}[Estimating fidelities via tomography]\label{lem:tomography_fidelity}
Given an $n$-qubit state $\rho$ and $M$ single-qubit pure states $|\phi_1\rangle,\ldots,|\phi_M\rangle$, there exists an algorithm that, with probability at least $1-\delta$, estimates $\langle\phi_i|\tr_{-j}(\rho)|\phi_i\rangle$ to additive error at most $\epsilon$ for all $i\in[M]$ and $j\in[n]$.
The algorithm only uses single-copy measurements.
The sample complexity is $O\left(\frac{1}{\epsilon^2}\log\frac{n}{\delta}\right)$ and the time complexity is $O\left(\frac{n}{\epsilon^2}\log\frac{n}{\delta}+Mn\right)$.
\end{lemma}

\noindent For learning stabilizer states, stabilizer product states, and states with high stabilizer dimension, we require the following well-known result on estimating Pauli observables with Bell measurements:

\begin{lemma}[Estimating correlations via Bell measurements]\label{subroutine: estimate Pauli correlation by Bell measurements}
    Let $\epsilon, \delta>0$, $S$ be a set of Pauli strings of size $M$. There exists an algorithm that, given copies of an unknown $n$-qubit state $\rho$, estimates $\tr(W_y\rho)^2$ for every $y\in S$ within error $\epsilon$ with probability at least $1-\delta$. The algorithm only performs Bell measurements on $\rho$. The sample complexity is $4\log(2M/\delta)/\epsilon^2$ and the time complexity is $O(Mn\log(M/\delta)/\epsilon^2)$.
\end{lemma}

\noindent When it comes to measuring in the basis defined by a stabilizer group, it is necessary to first synthesize the corresponding Clifford circuit.
Such manipulations are well-studied, see for example Ref.~\cite{aaronson2004improved}.
Our algorithm requires a subroutine to generate the Clifford circuit for a subspace of the stabilizer group, which we recall below:

\begin{lemma}[Clifford circuit synthesis, Lemma 3.2 of Ref.~\cite{grewal2023efficient}]\label{lem:Clifford_Synthesis}
Given a set of $m$ vectors spanning a $d$-dimensional subspace $A\subset\mathbb{F}_2^{2n}$.
There exists an algorithm that first decides whether the subspace is isotropic, and if so, outputs a circuit $C$ such that $C(A)=0^{2n-d}\otimes\mathbb{F}_2^d$.
The algorithm runs in $O(mn\min\{m,n\})$ time and $C$ if produced, contains $O(nd)$ number of elementary gates.
\end{lemma}

\noindent Finally, we will need the following entirely classical anti-concentration result showing that for any distribution over a subspace of $\mathbb{F}^d_2$, with enough i.i.d. samples from the distribution, the probability that the samples span the entire subspace is lower bounded:

\begin{lemma}[Finding a heavy-weight subspace]\label{lem: sample heavy-weight subspace}
    Let $\epsilon>0, m\in \mathbb{N}$, and $\cD$ be a distribution over $\mathbb{F}_2^d$. Suppose $x_1,\cdots, x_m$ are $m$ i.i.d. samples from $\cD$.
    \begin{enumerate}[label=(\alph*)]
        \item If $m=d$, with probability at least $\epsilon^d$, $\cD(\spn(x_1,\cdots, x_d))\ge 1-\epsilon$.
        \item Fix $0<\delta<1$. If $m\ge \frac{2\log(1/\delta)+2d}{\epsilon}$, with probability at least $1-\delta$, $\cD(\spn(x_1,\cdots, x_m))\ge 1-\epsilon$.\footnote{We note that a proof of (b) appears in Lemma 2.3 of Ref.~\cite{grewal2023efficient}, but the proof there is slightly flawed. Specifically, a Chernoff bound is applied to dependent random variables.}
    \end{enumerate}
\end{lemma}

\section{Properties of Bell difference sampling}
\label{sec:bell_properties}

In this section, we compile a collection of properties of Bell difference sampling that we will make use of in the sequel.

Firstly, for a mixed state $\rho$, if $|\phi\rangle$ is a maximizer of fidelity with $\rho$, it is known that $\cB_\rho$ concentrates on $\Weyl(|\phi\rangle)$ (See Lemma A.3 of Ref.~\cite{grewal2024agnostic}).
Here we present a more general version of this fact.

\begin{lemma}\label{lem: Bell difference sampling not small high dimension}
    Let $\sigma\in \cS^{n-t}$ be a state with $F(\rho, \sigma)\ge \tau$. Bell difference sampling on $\rho$ yields a Pauli string in $\Weyl(\sigma)$ with probability at least $\frac{\tau^4}{2^{2t}}$.
\end{lemma}
\begin{proof}
    There exists a Clifford $C$ such that $\sigma=C^\dagger (\ketbra{0^{n-t}}\otimes \sigma_0)C$. The fidelity $F(\rho, \sigma)$ and the probability $\Pr_{\cB_\rho}[\Weyl(\sigma)]$ are invariant when we apply $C$ to $\rho, \sigma$. Therefore, we can assume $\sigma = \ketbra{0^{n-t}}\otimes \sigma_0$ without loss of generality. By \Cref{lem: fidelity high dimension}, $\tr(\braket{0^{n-t}|\rho|0^{n-t}})\ge F(\rho, \sigma)\ge \tau$. Denote $\rho_0=\tr_{>n-t}(\rho)$, then $\braket{0^{n-t}|\rho_0|0^{n-t}}\ge \tau$. So
    \begin{align*}
        \Pr_{y\sim \cB_\rho}[y\in \Weyl(\sigma)]
        &\ge \sum_{y\in \{0, 1\}^{n-t}}\cB_{\rho}(0^{n}y0^{t})\\
        &=\frac{1}{4^n}\sum_{a\in \{0, 1\}^{2n}}\sum_{y\in \{0, 1\}^{n-t}}(-1)^{\braket{0^{n}y0^{t}, a}}\tr(W_a\rho)^4\\
        &=\frac{1}{2^{n+t}}\sum_{a\in \{0, 1\}^{n+t}}\tr(W_{0^{n-t}a}\rho)^4\\
        &\ge \frac{1}{2^{n+t}}\sum_{y\in \{0, 1\}^{n-t}}\tr(W_{0^ny0^t}\rho)^4\\
        &=\frac{1}{2^{n+t}}\sum_{s_1,s_2,s_3,s_4\in\mathbb{F}_2^{n-t}}\sum_{y\in \mathbb{F}_2^{n-t}}(-1)^{y\cdot(s_1+s_2+s_3+s_4)}\prod_{j=1}^4\braket{s_j|\rho_0|s_j}\\
        &=\frac{1}{2^{2t}}\sum_{s_1+s_2+s_3+s_4=0}\prod_{j=1}^4\braket{s_j|\rho_0|s_j}\\
        &\ge \frac{1}{2^{2t}}\braket{0^{n-t}|\rho_0|0^{n-t}}^4\ge \frac{\tau^4}{2^{2t}}\,.
    \end{align*}
    In the fourth line, we use the fact that $\sum_{y\in \{0, 1\}^{n-t}}(-1)^{\braket{0^ny0^t,a}}$ is $0$ if $a$ does not start with $0^t$ and is $2^{n-t}$ otherwise. In the sixth line, we expand $W_{0^ny0^t}=\sum_{s\in \mathbb{F}_2^{n-t}}(-1)^{y\cdot s}\ketbra{s}$. In the seventh line, we use the fact that $\sum_{y\in \{0, 1\}^{n-t}}(-1)^{y\cdot x}$ is 0 if $x=0$ and is $2^{n-t}$ otherwise.
\end{proof}

\noindent Our algorithm heavily relies on the fact that not only does $\cB_\rho$ concentrate on $\Weyl(|\phi\rangle)$, it is relatively evenly spread within $\Weyl(|\phi\rangle)$ and in particular is not too concentrated on any proper subspace of $\Weyl(\ket{\phi})$ (see \Cref{thm:stab_progress} and \Cref{thm:prod_progress}). As a warm-up, we first prove that $\cB_\rho$ is not too concentrated on any single element.

\begin{lemma}\label{lem: Bell distribution small}
    $\cB_{\rho}(x)\leq \frac{1}{2^n}$ for any $x\in\mathbb{F}_2^{2n}$.
\end{lemma}
\begin{proof}
    It is well known that $\sum_{x}W_x\otimes W_x = 2^n \SWAP$ (see, e.g., \cite[Lemma 4.10]{chen2022complexity}), so
    \begin{align*}
        \cB_\rho(x) &=\frac{1}{4^n}\sum_{a\in\mathbb{F}_2^{2n}}(-1)^{\braket{x,a}}\tr(W_a\rho)^4\leq \frac{1}{4^n}\sum_{a\in\mathbb{F}_2^{2n}}\tr(W_a\rho)^2=\frac{1}{2^n}\tr((\rho\otimes \rho) \SWAP)=\frac{\tr(\rho^2)}{2^n}\leq \frac{1}{2^n}\,. \qedhere
    \end{align*}
\end{proof}

\noindent Before proving~\Cref{thm:stab_progress} and \Cref{thm:prod_progress}, we need the following key lemma.

\begin{lemma}\label{lem:progress}
Let $\rho$ be an $n$-qubit mixed state, and let $\tau\triangleq\langle0^n|\rho|0^n\rangle$. Suppose that
\begin{equation*}
    \tau\geq\gamma\max_{z\in\{\ket{+},\ket{-},\ket{+i},\ket{-i}\}}\langle z0^{n-1}|\rho|z0^{n-1}\rangle\,,    
\end{equation*}
for some $\frac{1}{2}<\gamma\leq1$.
Let $S\triangleq 0^n\times\mathbb{F}_2^n=\mathrm{Weyl}(|0^n\rangle)$, and let $T\triangleq 0^{n+1}\times\mathbb{F}_2^{n-1}$.
Then
\[\sum_{x\in S\backslash T}\cB_\rho(x)\geq(\gamma - 1/2)^2\tau^4.\]
\end{lemma}

\noindent A special case of this was proven as Lemma 5.5 in~\cite{grewal2024improved}, but we extend it in two ways: 1) we consider general mixed states instead of just pure states $\rho$, and 2) our guarantee applies to $\gamma$-approximate local maximizers of fidelity, instead of just global maximizers.

\begin{proof}
Suppose $\rho=\frac{1}{2}\left(I\otimes\rho_I+X\otimes\rho_X+Y\otimes\rho_Y+Z\otimes\rho_Z\right)$.
Then $\rho_I,\rho_X,\rho_Y,\rho_Z$ are Hermitian by Hermicity of $\rho$.
Further denote $p_{\alpha s}=\langle s|\rho_\alpha|s\rangle\in\mathbb{R}$ for $\alpha\in\{I,X,Y,Z\}$, $s\in\mathbb{F}_2^{n-1}$. We have
\begin{align}
\MoveEqLeft \sum_{x\in S\backslash T}\cB_\rho(x)\nonumber\\
&=\frac{1}{4^n}\sum_{x\in S\backslash T}\sum_{a\in\mathbb{F}_2^{2n}}(-1)^{\langle x,a\rangle}\tr(W_a\rho)^4 \nonumber\\
&=\frac{1}{2^{n+1}}\left(\sum_{x\in0^n\times\mathbb{F}_2^n}\tr(W_x\rho)^4-\sum_{x\in10^{n-1}\times\mathbb{F}_2^n}\tr(W_x\rho)^4\right) \nonumber \\
&=\frac{1}{2^{n+1}}\sum_{x\in0^{n-1}\times\mathbb{F}_2^{n-1}}\left(\tr(W_x\rho_I)^4+\tr(W_x\rho_Z)^4-\tr(W_x\rho_X)^4-\tr(W_x\rho_Y)^4\right)\nonumber\\
&=\frac{1}{2^{n+1}}\sum_{x\in\mathbb{F}_2^{n-1}}\sum_{s_1,s_2,s_3,s_4\in\mathbb{F}_2^{n-1}}(-1)^{x\cdot(s_1+s_2+s_3+s_4)}\nonumber\\
&\qquad\qquad\qquad \times \left(p_{Is_1}p_{Is_2}p_{Is_3}p_{Is_4}+p_{Zs_1}p_{Zs_2}p_{Zs_3}p_{Zs_4}-p_{Xs_1}p_{Xs_2}p_{Xs_3}p_{Xs_4}-p_{Ys_1}p_{Ys_2}p_{Ys_3}p_{Ys_4}\right)\nonumber\\
&=\frac{1}{4}\sum_{s_1+s_2+s_3+s_4=0}\left(p_{Is_1}p_{Is_2}p_{Is_3}p_{Is_4}+p_{Zs_1}p_{Zs_2}p_{Zs_3}p_{Zs_4}-p_{Xs_1}p_{Xs_2}p_{Xs_3}p_{Xs_4}-p_{Ys_1}p_{Ys_2}p_{Ys_3}p_{Ys_4}\right)\,. \label{eq:individual}
\end{align}

The positivity of $\rho$ requires that
\begin{multline*}
\Bigl(\cos\frac{\theta}{2}\langle0|+e^{-i\varphi}\sin\frac{\theta}{2}\langle1|\Bigr)\otimes\langle s|\rho|\Bigl(\cos\frac{\theta}{2}|0\rangle+e^{i\varphi}\sin\frac{\theta}{2}|1\rangle\Bigr)\otimes|s\rangle\\
=\frac{1}{2}(p_{Is}+\sin\theta\cos\varphi\cdot p_{Xs}+\sin\theta\sin\varphi\cdot p_{Ys}+\cos\theta\cdot p_{Zs})\geq0
\end{multline*}
for $\forall\theta,\varphi$,
so $p_{Is}\geq\sqrt{p_{Xs}^2+p_{Ys}^2+p_{Zs}^2}$ for all $s\in\mathbb{F}^{n-1}_2$.

Now consider each individual term in Eq.~\eqref{eq:individual}.
We have
\begin{align*}
\MoveEqLeft p_{Is_1}p_{Is_2}p_{Is_3}p_{Is_4}+p_{Zs_1}p_{Zs_2}p_{Zs_3}p_{Zs_4}-p_{Xs_1}p_{Xs_2}p_{Xs_3}p_{Xs_4}-p_{Ys_1}p_{Ys_2}p_{Ys_3}p_{Ys_4}\\
&\ge \sqrt{(p_{Xs_1}^2+p_{Ys_1}^2+p_{Zs_1}^2)(p_{Xs_2}^2+p_{Ys_2}^2+p_{Zs_2}^2)(p_{Xs_3}^2+p_{Ys_3}^2+p_{Zs_3}^2)(p_{Xs_4}^2+p_{Ys_4}^2+p_{Zs_4}^2)}\\
&\qquad\qquad 
-|p_{Xs_1}||p_{Xs_2}||p_{Xs_3}||p_{Xs_4}|
-|p_{Ys_1}||p_{Ys_2}||p_{Ys_3}||p_{Ys_4}|
-|p_{Zs_1}||p_{Zs_2}||p_{Zs_3}||p_{Zs_4}|\\
&\geq0\,,
\end{align*}
where the second inequality is because
\begin{align*}
\MoveEqLeft \left(p_{Xs_1}^2+p_{Ys_1}^2+p_{Zs_1}^2\right)\left(p_{Xs_2}^2+p_{Ys_2}^2+p_{Zs_2}^2\right)\left(p_{Xs_3}^2+p_{Ys_3}^2+p_{Zs_3}^2\right)\left(p_{Xs_4}^2+p_{Ys_4}^2+p_{Zs_4}^2\right)\\
&\qquad \qquad -\left(|p_{Zs_1}||p_{Zs_2}||p_{Zs_3}||p_{Zs_4}|+|p_{Xs_1}||p_{Xs_2}||p_{Xs_3}||p_{Xs_4}|+|p_{Ys_1}||p_{Ys_2}||p_{Ys_3}||p_{Ys_4}|\right)^2\\
&\ge \left(p_{Xs_1}p_{Xs_2}p_{Ys_3}p_{Ys_4}-p_{Ys_1}p_{Ys_2}p_{Xs_3}p_{Xs_4}\right)^2+\left(p_{Ys_1}p_{Ys_2}p_{Zs_3}p_{Zs_4}-p_{Zs_1}p_{Zs_2}p_{Ys_3}p_{Ys_4}\right)^2\\
&\qquad \qquad \qquad +\left(p_{Zs_1}p_{Zs_2}p_{Xs_3}p_{Xs_4}-p_{Xs_1}p_{Xs_2}p_{Zs_3}p_{Zs_4}\right)^2 \geq 0\,.
\end{align*}
Consider the term $s_1=s_2=s_3=s_4=0$ in Eq.~\eqref{eq:individual}.
Denote $w=\frac{p_{I0^{n-1}}}{\tau}$, $x=\left|\frac{p_{X0^{n-1}}}{\tau}\right|$, $y=\left|\frac{p_{Y0^{n-1}}}{\tau}\right|$ and $z=\frac{p_{Z0^{n-1}}}{\tau}$, then this term equals to $\tau^4(w^4+z^4-x^4-y^4)$.
We want to lower bound this quality.
Note that $\tau=\langle00^{n-1}|\rho|00^{n-1}\rangle=\frac{p_{I0^{n-1}}+p_{Z0^{n-1}}}{2}$, $\langle\pm0^{n-1}|\rho|\pm0^{n-1}\rangle=\frac{p_{I0^{n-1}}\pm p_{X0^{n-1}}}{2}$ and $\langle\pm i0^{n-1}|\rho|\pm i0^{n-1}\rangle=\frac{p_{I0^{n-1}}\pm p_{Y0^{n-1}}}{2}$.
Since $|00^{n-1}\rangle$ is the closest stabilizer product state, we have the following constrained optimization problem:
\begin{align*}
\min&\qquad w^4+z^4-x^4-y^4\\
\text{s.t. }&\qquad
\begin{aligned}
w+z=2\\
w+x\leq\frac{2}{\gamma}\\
w+y\leq\frac{2}{\gamma}\\
x\geq0\\
y\geq0\\
x^2+y^2+z^2\leq w^2
\end{aligned}
\end{align*}
Without loss of generality assume $x\leq y$, the problem simplifies to
\begin{align*}
\min&\qquad w^4+(2-w)^4-x^4-y^4\\
\text{s.t. }&\qquad
\begin{aligned}
0\leq x\leq y\\
w+y\leq\frac{2}{\gamma}\\
x^2+y^2\leq4w-4
\end{aligned}
\end{align*}
We proceed by casework:

\vspace{0.5em}\noindent\textbf{Case 1:}
    $\sqrt{4w-4}\leq\frac{2}{\gamma}-w$.
    
    \noindent In this case, $1\leq w\leq w_1$ for $w_1 \triangleq \frac{2}{\gamma}-2\sqrt{\frac{2}{\gamma}}+2$. The minimum is reached at $y=\sqrt{4w-4}$, $x=0$, with minimum value $2(-w^2+2w)^2$, which is monotonically decreasing as $w$ grows from $1$ to $w_1$.
    Hence the minimum value is reached at $w=w_1$ in this region.

\vspace{0.5em}\noindent\textbf{Case 2:}
    $\frac{2}{\gamma}-w\leq\sqrt{4w-4}\leq\sqrt{2}\left(\frac{2}{\gamma}-w\right)$.
    
    \noindent In this case, $w_1\leq w\leq w_2$ for $w_2\triangleq \frac{2}{\gamma}-\sqrt{\frac{4}{\gamma}-1}$. The minimum is reached at $y=\frac{2}{\gamma}-w$, $x=\sqrt{4w-4-y^2}$, with minimum value $w^4+(2-w)^4-2\left(\frac{x^2+y^2}{2}\right)^2-2\left(y^2-\frac{x^2+y^2}{2}\right)^2=2(w^2-2w+2)^2-2\left(\left(\frac{2}{\gamma}-w\right)^2-(2w-2)\right)^2$.
    Note that when $w$ grows from $w_1$ to $w_2$, $\left(\frac{2}{\gamma}-w\right)^2-(2w-2)$ monotonically decreases to $0$.
    Hence the minimum value is reached at $w=w_1$ in this region.

\vspace{0.5em}\noindent\textbf{Case 3:}         
    $0\leq\sqrt{2}\left(\frac{2}{\gamma}-w\right)\leq\sqrt{4w-4}$.

    \noindent In this case, $w_2\leq w\leq\frac{2}{\gamma}$. The minimum value is reached at $x=y=\frac{2}{\gamma}-w$, with minimum value $w^4+(2-w)^4-2\left(\frac{2}{\gamma}-w\right)^2$, which is monotonically increasing, so the minimum value is reached at $w=w_2$ in this region.
    Thus the minimum value is $2(-w_1^2+2w_1)^2\geq256(17-12\sqrt{2})(\gamma - 1/2)^2$.

We conclude that $\sum_{x\in S\backslash T}\cB_\rho(x)\geq64(17-12\sqrt{2})(\gamma - 1/2)^2\tau^4\geq(\gamma - 1/2)^2\tau^4$.
\end{proof}

\noindent The particular choices of stabilizer states in~\Cref{lem:progress} are actually without loss of generality.
This is because we can always use a Clifford gate to rotate the relevant stabilizer states into the choice in~\Cref{lem:progress}:
\begin{lemma}[Lemma 5.1 of Ref.~\cite{grewal2024improved}]\label{lem:stab_rotate}
For a $n$ qubit stabilizer state $|\phi\rangle$, suppose $S=\mathrm{Weyl}(|\phi\rangle)$ is its stabilizer group and $T\subset S$ is a subspace of $S$ of dimension $n-t$.
There exists a Clifford gate $C$ such that $C|\phi\rangle=|0^n\rangle$, $C(S)=0^n\times\mathbb{F}_2^n$ and $C(T)=0^{n+t}\times\mathbb{F}_2^{n-t}$.
\end{lemma}

\noindent With Lemma~\ref{lem:progress} in hand, we can readily prove our main estimates of this section which show that $B_\rho$ is not too concentrated on any proper subspace of $\mathrm{Weyl}(\ket{\phi})$. First, we show this in the stabilizer state setting:

\begin{theorem}\label{thm:stab_progress}
Let $\rho$ be an $n$-qubit state.
Let $|\phi\rangle$ be a $\gamma$-approximate local maximizer of fidelity with $\rho$ over stabilizer states ($\frac{1}{2}<\gamma\leq1$).
Suppose $\langle\phi|\rho|\phi\rangle=\tau$, and let $S=\mathrm{Weyl}(|\phi\rangle)$.
Let $T\subset S$ be a proper subspace of $S$, then
\[\sum_{x\in S\backslash T}\cB_\rho(x)\geq(\gamma - 1/2)^2\tau^4.\]
\end{theorem}

\begin{proof}
Let $T'$ be a subspace of $S$ of dimension $n-1$ that contains $T$, i.e. $T\subseteq T'\subset S$.
By~\Cref{lem:stab_rotate}, there exists a Clifford gate $C$ such that $C|\phi\rangle=|0^n\rangle$, $C(S)=0^n\times\mathbb{F}_2^n$ and $C(T')=0^{n+1}\times\mathbb{F}_2^{n-1}$.
We have $\langle0^n|C\rho C^\dagger|0^n\rangle=\langle\phi|\rho|\phi\rangle=\tau$.
Since the neighborhood is invariant under Clifford rotations, $|0^n\rangle$ is a $\gamma$-approximate local maximizer for $C\rho C^\dagger$.
In particular, $\tau\triangleq \langle0^n|C\rho C^\dagger|0^n\rangle$ satisfies $\tau\geq\gamma\max_{z\in\{\ket{+},\ket{-},\ket{+i},\ket{-i}\}}\langle z0^{n-1}|C\rho C^\dagger|z0^{n-1}\rangle$.
From~\Cref{lem:BDS} we can see that $\cB_{C\rho C^\dagger}(CxC^\dagger)=\cB_\rho(x)$.
Thus by~\Cref{lem:progress} we have
\begin{align*}
    \sum_{x\in S\backslash T}\cB_\rho(x)\geq\sum_{x\in S\backslash T'}\cB_\rho(x)&=\sum_{x\in(0^n\times\mathbb{F}_2^n)\backslash(0^{n+1}\times\mathbb{F}_2^{n-1})}\cB_{C\rho C^\dagger}(x)\geq(\gamma - 1/2)^2\tau^4\,. \qedhere
\end{align*}
\end{proof}

\noindent We have an analogous result for the stabilizer product state setting:

\begin{theorem}\label{thm:prod_progress}
Let $\rho$ be an $n$-qubit state.
Let $|\phi\rangle=\argmax_{|\varphi\rangle\in\mathcal{SP}}F(|\varphi\rangle,\rho)$.
That is, $F(|\phi\rangle,\rho)=F_{\mathcal{SP}}(\rho)$.
Suppose $\langle\phi|\rho|\phi\rangle=\tau$, and let $S=\mathrm{Weyl}(|\phi\rangle)$.
Let $i\in[n]$. Then
\[\sum_{x\in S,x_i\vee x_{i+n}\neq0}\cB_\rho(x)\geq\frac{1}{4}\tau^4.\]
\end{theorem}

\begin{proof}
Without loss of generality assume $i=1$.
There exists a tensor product of single-qubit Clifford gates $C$ such that $C|\phi\rangle=|0^n\rangle$.
This $C$ maps stabilizer product states into stabilizer product states and thus $|0^n\rangle=\argmax_{|\varphi\rangle\in\mathcal{SP}}F(|\varphi\rangle,C\rho C^\dagger)$.
In particular, $\tau\triangleq\langle0^n|C\rho C^\dagger|0^n\rangle$ satisfies the inequality $\tau\geq\max_{z\in\{\ket{+},\ket{-},\ket{+i},\ket{-i}\}} \langle z0^{n-1}|C\rho C^\dagger|z0^{n-1}\rangle$.
Thus by~\Cref{lem:progress} we have
\begin{align*}
    \sum_{x\in S,x_ix_{i+n}\neq0}\cB_\rho(x) &=\sum_{x\in(0^n\times\mathbb{F}_2^n)\backslash(0^{n+1}\times\mathbb{F}_2^{n-1})}\cB_{C\rho C^\dagger}(x)\geq\frac{1}{4}\tau^4\,.\qedhere
\end{align*}
\end{proof}

\section{Agnostic tomography of stabilizer states}\label{sec:stabilizer}

In this section, we construct and analyze our algorithm for agnostic tomography of stabilizer states. In fact, we give a more general guarantee, namely an algorithm such that every $\gamma$-approximate local maximizer of fidelity has a non-negligible chance of being the final output:

\begin{theorem}\label{thm:all_gamma_approximate_local_maximizer}
    Fix $\tau>0$ and $1/2<\gamma\leq1$, and let $\rho$ be an unknown $n$-qubit state. There is an algorithm with the following guarantee.
    
    Let $\ket{\phi}$ be any $\gamma$-approximate local maximizer of fidelity with $\rho$, and suppose its fidelity with $\rho$ is at least $\tau$. Given copies of $\rho$, the algorithm outputs $\ket{\phi}$ with probability at least $((\gamma-1/2)\tau)^{O(\log\frac{1}{\tau})}$. 
    
    The algorithm only performs single-copy and two-copy measurements on $\rho$. The sample complexity is $O(\frac{n}{(\gamma-1/2)^2\tau^5})$ and the runtime is $O(\frac{n^2}{(\gamma-1/2)^2\tau^4}(n+\log\frac{1}{\gamma-1/2}+\frac{1}{\tau}))$.
\end{theorem}

\noindent We give a proof of~\Cref{thm:all_gamma_approximate_local_maximizer} in Sections~\ref{sec:alg_stab} and~\ref{sec:analysis_stab}. In the rest of this subsection, we record some consequences of this general result.

Firstly, by repeating the algorithm in~\Cref{thm:all_gamma_approximate_local_maximizer} sufficiently many times, we immediately obtain the list-decoding guarantee mentioned in Section~\ref{sec:stabilizer_our_results}:

\begin{corollary}[List-decoding stabilizer states]\label{cor:stab_list}
    Fix $\tau, \delta>0$ and $1/2<\gamma\leq1$, and let $\rho$ be an unknown $n$-qubit state. There is an algorithm with the following guarantee.

    There is an algorithm that, given copies of $\rho$, returns a list of stabilizer states of length $O(\log(1/\delta))\cdot ((\gamma-1/2)\tau)^{-O(\log\frac{1}{\tau})}$ so that with probability at least $1 - \delta$, all stabilizers which are $\gamma$-approximate local maximizers and have fidelity at least $\tau$ with $\rho$ appear in the list.

    The algorithm only performs single-copy and two-copy measurements on $\rho$. The sample complexity is $O(n\log(1/\delta)(\gamma - 1/2)^{-2-o(1)})\cdot (1/\tau)^{O(\log 1/\tau)}$ and the runtime is $O(n^3\log(1/\delta)(\gamma - 1/2)^{-2-o(1)})\cdot (1/\tau)^{O(\log 1/\tau)}$.
\end{corollary}

\noindent Since the closest stabilizer state is a $1$-approximate local maximizer, this also readily implies an algorithm for proper agnostic tomography of stabilizer states. Namely, we can run classical shadows to estimate the fidelity of every state in the list output by~\Cref{cor:stab_list} and select the one with the highest fidelity. The proof details are straightforward and deferred to Section~\ref{sec:defer_agnostic_stabilizer_cor}.

\begin{corollary}[Proper agnostic tomography of stabilizer states]\label{cor:agnostic_learning_stabilizer}
    Fix $\tau\ge \epsilon>0$ and $\delta>0$. There is an algorithm that, given copies of an $n$-qubit state $\rho$ with $F_\cS(\rho)\ge \tau$, returns a stabilizer state $\ket{\phi}$ that satisfies $F(\rho, \ket{\phi})\ge F_S(\rho)-\epsilon$ with probability at least $1-\delta$.
    The algorithm only performs single-copy and two-copy measurements on $\rho$. The sample complexity is $n\log(1/\delta)(1/\tau)^{O(\log 1/\tau)}+O((\log^2(1/\tau) + \log(1/\delta))/\epsilon^2)$ and the runtime is $O(n^2\log(1/\delta)\cdot (n + \log(1/\delta)/\epsilon^2))\cdot(1/\tau)^{O(\log 1/\tau)}$.
\end{corollary}

\noindent This also implies the first efficient algorithm for estimating the magic of a quantum state as quantified by its stabilizer fidelity. This is essentially immediate from the proof of Corollary~\ref{cor:agnostic_learning_stabilizer}, so we defer the details to Section~\ref{sec:defer_fid_est}.

\begin{corollary}[Efficient estimation of stabilizer fidelity]\label{cor:fid_est}
    Fix $\epsilon, \delta>0$. There is an algorithm that, given copies of an $n$-qubit state $\rho$, estimates $F_\cS(\rho)$ to within additive error $\epsilon$ with probability at least $1-\delta$. The sample complexity is $O(n\log(1/\delta))\cdot (1/\epsilon)^{O(\log 1/\epsilon)}$ and the runtime is $O(n^2\log(1/\delta)(n + \log(1/\delta)))\cdot (1/\epsilon)^{O(\log 1/\epsilon)}$.
\end{corollary}

\noindent One last implication of \Cref{thm:all_gamma_approximate_local_maximizer} is an upper bound of the number of $\gamma$-approximate local maximizers of fidelity, simply coming from the fact that the sum of the probabilities that each of them is output in~\Cref{thm:all_gamma_approximate_local_maximizer} is at most $1$:

\begin{corollary}\label{cor:structural}
Given an $n$-qubit state $\rho$, for $\frac{1}{2}<\gamma\leq1$,  the number of $\gamma$-approximate local maximizers of fidelity with $\rho$, with fidelity at least $\tau$ is at most 
\begin{equation}
    ((\gamma-1/2)\tau)^{-O(\log 1/\tau)}\,.
\end{equation}
In particular, when $\tau=\Theta(1)$ and $\gamma=\frac{1}{2}+\frac{1}{\poly(n)}$, there are polynomially many $\gamma$-approximate local maximizers.
\end{corollary}

\noindent Note that the threshold of $\gamma=1/2$ is tight in the sense that even for fidelity $\tau = 1/2$,
there can be exponentially many $1/2$-approximate local maximizers with fidelity at least $\tau = 1/2$.
To see this, let $\rho$ to be a stabilizer state.
By Theorem 15 of Ref.~\cite{garcia2014geometry}, there are $4(2^n-1)$ many nearest stabilizers, each of them have fidelity $\frac{1}{2}$ with $\rho$ and thus are $\frac{1}{2}$-approximate local maximizers. An upshot of Corollary~\ref{cor:structural} is that as soon as one moves away from this threshold of $\gamma = 1/2$ by a small margin, the number of $\gamma$-approximate local maximizers decreases dramatically.

\subsection{Construction of the algorithm}
\label{sec:alg_stab}
Let $\ket{\phi}$ be a $\gamma$-approximate local maximizer of fidelity with $\rho$ with fidelity at least $\tau$.
Recall that the notion of a complete family of projectors in this setting is a set of $2^n$ projectors corresponding to the signed stabilizers of $|\phi\rangle$ (see Eq.~\eqref{eq:proj_def}). If one could find these, one could measure in the corresponding stabilizer basis and obtain $\ket{\phi}$ with probability $\tau$ as desired.

Recall from the discussion in Section~\ref{sec:overview} that the main challenge is to implement Steps 1 and 3 of stabilizer bootstrapping, namely accumulating a high-correlation family and, if the collection is incomplete, sampling a low-correlation projector. Our main tool for these steps will be Bell difference sampling.

\vspace{0.5em}\noindent \textbf{Step 1: Find a high-correlation family.}\vspace{0.5em}

\noindent We begin by formally defining the notion of a high-correlation family in the context of learning stabilizer states:

\begin{definition}\label{def:high_cor_basis}
    Let $\rho$ be the unknown $n$-qubit state. We say $y\in\mathbb{F}_2^{2n}$ is a \emph{high-correlation Pauli string} if $\tr(W_y\rho)^2>0.7$, and a \emph{low-correlation Pauli string} if $\tr(W_y\rho)^2\leq 0.7$. We say a collection of Pauli strings $F\subset \mathbb{F}^{2n}_2$ is an \emph{$\epsilon$-high-correlation family} if 
    \begin{equation*}
        \Pr_{y\sim \cB_\rho}[\tr(W_y\rho)^2> 0.7\wedge y\not\in F]\leq \epsilon\,.    
    \end{equation*}
    A basis of an $\epsilon$-high-correlation family is called an \emph{$\epsilon$-high-correlation basis}.
\end{definition}

\noindent To produce a high-correlation family, we will use Bell difference sampling to produce a collection of Pauli strings, and then keep those strings $y$ for which our estimate of the correlation is sufficiently large. For convenience, we complete the resulting list $H$ of strings to a basis of a stabilizer family if $\dim(\spn(H))<n$. See Algorithm~\ref{alg:select_high_correlation_Pauli} below for a formal specification. 

\begin{algorithm}[htbp]
    \DontPrintSemicolon
    \caption{Select high-correlation Pauli strings}\label{alg:select_high_correlation_Pauli}
    \KwInput{$\delta, \epsilon>0$, copies of an $n$-qubit state $\rho$}
    \KwOutput{A basis $H$ of a stabilizer family}
    \Goal{With probability at least $1-\delta$, $H$ is a $\epsilon$-high-correlation basis}
    Bell difference sampling on $\rho$ for $m = 8(\log(3/\delta)+4n)/\epsilon$ times. Denote the outcomes by $x_1, \ldots, x_{m}$.\\
    Using Bell measurements (\Cref{subroutine: estimate Pauli correlation by Bell measurements}), obtain estimators $\widehat{E}_i$ of $\tr(W_{x_i}\rho)^2$ such that with probability at least $1-\delta/3$, $|\widehat{E}_i-\tr(W_{x_i}\rho)^2|\leq 0.1$ for all $i$.\\
    $H'= \{x_i: \widehat{E}_i> 0.6\}$.\label{line: select high-correlation Pauli strings, anti-commuting}\\
    Let $H$ be a basis for $\spn(H')$. Abort if $H$ contains anti-commuting Pauli strings.\label{line:corner_anti}\\
    If $|H|<n$, add some commuting Pauli string to $H$ to make it the basis of a stabilizer family.\label{line:corner_miss}\\
    \Return $H$.
\end{algorithm}

\begin{lemma}[Finding a high-correlation basis]\label{lem:high_correlation_space}
    With probability at least $1-\delta$, the output $H$ of \Cref{alg:select_high_correlation_Pauli} is an $\epsilon$-high-correlation basis. The algorithm uses $O\left(\frac{1}{\epsilon}(n+\log\frac{1}{\delta})\right)$ copies of $\rho$ and $O\left(\frac{n}{\epsilon}(n+\log\frac{1}{\delta})(n+\log\frac{1}{\delta}+\log\frac{1}{\epsilon})\right)$ time.
\end{lemma}
\begin{proof}
    Let $S=\{x_1, \cdots, x_m\}$, $T=\{y\in\mathbb{F}_2^{2n}:\tr(W_y\rho)^2>0.7\}$, and $S_h=S\cap T$, and let $p=\Pr_{y\sim \cB_\rho}[y\in T]$. Let $\cD_h$ be the distribution of Bell difference sampling conditioned on landing in the set $T$ of high-correlation Pauli strings, i.e., $\cD_h(y) = 1[y\in T]\cdot \cB_{\rho}(y)/p$. We first show that, with probability at least $1-2\delta/3$ over $S$, 
    \begin{equation}\label{eq: find a high-correlation basis 1}
        \Pr_{y\sim \cB_\rho}[y\in T\wedge y\not\in \spn(S_h)]\leq \epsilon.
    \end{equation}
    If $p\leq \epsilon$, \eqref{eq: find a high-correlation basis 1} is trivial. If $p>\epsilon$, By Chernoff bound, $\Pr_{S\sim\cB_\rho^{\otimes m}}[|S_h|\leq pm/2]\leq e^{-pm/8}\leq e^{-\epsilon m/8}\leq \delta/3$.
    Elements of $S_h$ can be regarded as independent samples from $\cD_h$. By \Cref{lem: sample heavy-weight subspace}, if $|S_h|\ge pm/2\ge (2\log(3/\delta)+4n)/(\epsilon/p)$, with probability at least $1-\delta/3$, $\Pr_{y\sim \cD_h}[y\in\spn(S_h)]\ge 1-\epsilon/p$. So 
    \begin{equation*}
        \Pr_{y\sim \cB_\rho}[y\in T\wedge y\not\in \spn(S_h)]=\Pr_{y\sim \cB_\rho}[y\in T]\Pr_{y\sim \cB_\rho}[y\not\in \spn(S_h)|y\in T]=p\Pr_{y\sim \cD_h}[y\not\in\spn(S_h)]\leq \epsilon.
    \end{equation*}

    Therefore, \eqref{eq: find a high-correlation basis 1} holds with probability at least $1-2\delta/3$ over $S$. So with probability at least $1-\delta$, \eqref{eq: find a high-correlation basis 1} holds and $|\widehat{E}_i-\tr(W_{x_i}\rho)^2|\leq 0.1$ for all $i$. We now show that in this case, $H$ must be an $\epsilon$-high-correlation basis. Indeed, since $\widehat{E_i}\ge \tr(W_{x_i}\rho)^2-0.1>0.6$ for $i\in S_h$, $H'$ in Line \ref{line: select high-correlation Pauli strings, anti-commuting} contains $S_h$. Furthermore, since $\tr(W_{x_i}\rho)^2\ge \widehat{E}_i-0.1>0.5$ for $i\in H'$, $H'$ and thus $H$ does not contain anti-commuting Pauli strings by \Cref{lem:high-corr-commute}. Hence, the algorithm will not abort and output a basis $H$ of a stabilizer family that contains $S_h$. So
    \begin{equation*}
        \Pr_{y\sim \cB_\rho}[\tr(W_y\rho)^2> 0.7\wedge y\not\in \spn(H)]\leq \Pr_{y\sim \cB_\rho}[\tr(W_y\rho)^2> 0.7\wedge y\not\in \spn(S_h)]\leq \epsilon.
    \end{equation*}
    The sample complexity is $4m+400\log(6m/\delta)=O\left(\frac{1}{\epsilon}(n+\log\frac{1}{\delta})\right)$ (the second term comes from Bell measurements, see \Cref{subroutine: estimate Pauli correlation by Bell measurements}).
    Bell difference sampling takes $O(mn)$ time.
    The Bell measurements takes $O(mn\log\frac{m}{\delta})$ time.
    Finding $H'$ takes $O(m)$ time.
    We can use~\Cref{lem:Clifford_Synthesis} to perform~\Cref{line:corner_anti,line:corner_miss}, and interpret $C^\dagger(0^n\otimes\mathbb{F}_2^n)$ as $H$ where $C$ is the output Clifford circuit.
    The running time of this part is thus $O(mn^2)$.
    Hence the total running time is $O\left(\frac{n}{\epsilon}(n+\log\frac{1}{\delta})(n+\log\frac{1}{\delta}+\log\frac{1}{\epsilon})\right)$.
\end{proof}

\vspace{0.5em}\noindent \textbf{Step 2: If the family is complete, i.e. if $\spn(H)=\Weyl(\ket{\phi})$, then directly obtain the answer.}\vspace{0.5em}

\noindent Let $H$ be the basis of a stabilizer family obtained in Step 1. The second step is to measure $\rho$ on the joint eigenspace of $\spn(H)$. If $\spn(H)=\Weyl(\ket{\phi})$, the measurement outputs $\ket{\phi}$ with probability $\braket{\phi|\rho|\phi}\ge \tau$.

\vspace{0.5em}\noindent \textbf{Step 3: If the family is incomplete, sample a low-correlation projector.}\vspace{0.5em}

\noindent On the other hand, if $\spn(H)\neq \Weyl(\ket{\phi})$, $\spn(H)\cap \Weyl(\ket{\phi})$ is a proper subspace of $\Weyl(\ket{\phi})$. Since Bell difference distribution is evenly distributed over $\Weyl(\ket{\phi})$, we can sample a Pauli string $y$ from $\Weyl(\ket{\phi})\backslash \spn(H)$ with a not-too-small probability (\Cref{thm:stab_progress}). Combined with the fact that $\spn(H)$ contains most of the high-correlation Pauli strings (\Cref{lem:high_correlation_space}), the sampled Pauli string is likely to be low-correlation. With an additional probability of $1/2$, we get the correct sign $\sgn\in\{\pm 1\}$ with $W_y\ket{\phi}=\sgn\ket{\phi}$. The projector $\frac{I + \sgn W_y}{2}$ will be a low-correlation projector in this case. We summarize this reasoning in the following lemma.

\begin{lemma}[Sampling a low-correlation projector]\label{lem:sample_low_correlation_Pauli}
Fix $\frac{1}{2}<\gamma\leq1$.
Let $\ket{\phi}$ be a $\gamma$-approximate local maximizer of fidelity with $\rho$ with fidelity at least $\tau$ and $H$ be a $\frac{1}{4}(\gamma-1/2)^2\tau^4$-high-correlation basis. If $\spn(H)\neq \Weyl(\ket{\phi})$, we have 
\begin{equation}
    \Pr_{y\sim \cB_\rho, \sgn\sim \{\pm 1\}}[\tr(W_y\rho)^2\leq 0.7\wedge W_y\ket{\phi}=\sgn\ket{\phi}]\ge\frac{3}{8}(\gamma - 1/2)^2\tau^4.
\end{equation}
\end{lemma}
\begin{proof}
Let $S=\Weyl(|\phi\rangle)$.
Since $\spn(H)\neq S$, $S\backslash\spn(H)\subset S$.
Thus
\begin{align*}
\MoveEqLeft\Pr_{y\sim\cB_\rho}[\tr(W_y\rho)^2\leq0.7\wedge y\in S]\\
&\geq\Pr_{y\sim\cB_\rho}[\tr(W_y\rho)^2\leq0.7\wedge y\in S\wedge y\notin\spn(H)]\\
&=\Pr_{y\sim\cB_\rho}[y\in S\wedge y\notin\spn(H)]-\Pr_{y\sim\cB_\rho}[y\in S\wedge y\notin\spn(H)\wedge\tr(W_y\rho)^2>0.7]\\
&\geq\Pr_{y\sim\cB_\rho}[y\in S\wedge y\notin\spn(H)]-\Pr_{y\sim\cB_\rho}[y\notin\spn(H)\wedge\tr(W_y\rho)^2>0.7]\\
&\geq(\gamma - 1/2)^2\tau^4-\frac{1}{4}(\gamma - 1/2)^2\tau^4 =\frac{3}{4}(\gamma - 1/2)^2\tau^4,
\end{align*}
where the fourth line follows from~\Cref{thm:stab_progress} and the assumption that $H$ is a high-correlation basis.
The claim follows immediately.
\end{proof}

\vspace{0.5em}\noindent \textbf{Step 4: Bootstrap by measuring.}\vspace{0.5em}

\noindent Let $(\sgn, y)$ be the signed Pauli string sampled in Step 3. Consider the post-measurement state $\rho' = \Pi_y^{\sgn}\rho \Pi_{y}^{\sgn}/\tr(\Pi_y^{\sgn}\rho \Pi_{y}^{\sgn})$, recalling from Eq.~\eqref{eq:proj_def} that $\Pi_y^\sgn$ denotes the projector $\frac{I + \sgn W_y}{2}$. We prove that if $(\sgn, y)$ a low-correlation Pauli string that stabilizes $\ket{\phi}$, then the measurement amplifies the state's fidelity with $\ket{\phi}$, and $\ket{\phi}$ is still a $\gamma$-approximate local maximizer of fidelity with $\rho'$.

\begin{lemma}\label{lem:fidelity_amplification}
    Let $\ket{\phi}$ be a state and $(\sgn, y)$ be a signed Pauli string such that $W_y\ket{\phi}=\sgn\ket{\phi}$ and $\tr(W_y\rho)^2\leq 0.7$. Denote $\rho' = \Pi_y^{\sgn}\rho \Pi_{y}^{\sgn}/\tr(\Pi_y^{\sgn}\rho \Pi_{y}^{\sgn})$. Then $\braket{\phi|\rho'|\phi}\ge 1.08\braket{\phi|\rho|\phi}$. Furthermore, if $\ket{\phi}$ is a $\gamma$-approximate local maximizer of fidelity with $\rho$, it is also a $\gamma$-approximate local maximizer of fidelity with $\rho'$.
\end{lemma}

\begin{proof}
We have
\[\langle\phi|\rho'|\phi\rangle=\langle\phi|\frac{\Pi_y^{\sgn}\rho\Pi_y^{\sgn}}{\tr(\Pi_y^{\sgn}\rho)}|\phi\rangle=\frac{\langle\phi|\rho|\phi\rangle}{\tr(\Pi_y^{\sgn}\rho)}\geq\frac{\tau}{\tr(\frac{I+\sgn W_y}{2}\rho)}\geq\frac{\tau}{\frac{1+\sqrt{0.7}}{2}}\geq1.08\tau.\]
By~\Cref{lem:NN_stab}, for any stabilizer $|\phi'\rangle$ with $|\langle\phi|\phi'\rangle|=\frac{1}{\sqrt{2}}$, there exists a Pauli string $x\in\mathbb{F}_2^{2n}$ and an integer $\ell $ such that $|\phi'\rangle=\frac{I+i^\ell W_x}{\sqrt{2}}|\phi\rangle$.
Note that
\begin{align*}
\Pi_y^{\sgn}|\phi'\rangle&=\frac{I+\sgn W_y}{2}\frac{I+i^\ell W_x}{\sqrt{2}}|\phi\rangle\\
&=\frac{I+i^\ell W_x+\sgn W_y+\sgn (-1)^{\langle x,y\rangle}i^\ell W_xW_y}{2\sqrt{2}}|\phi\rangle\\
&=
\begin{cases}
|\phi'\rangle&\text{~if~}\langle x,y\rangle=0\\
\frac{1}{\sqrt{2}}|\phi\rangle&\text{~if~}\langle x,y\rangle=1\,.
\end{cases}
\end{align*}
Thus
\begin{equation*}
\langle\phi'|\rho'|\phi'\rangle=\langle\phi'|\frac{\Pi_y^{\sgn}\rho\Pi_y^{\sgn}}{\tr(\Pi_y^{\sgn}\rho)}|\phi'\rangle=
\begin{cases}
\frac{\langle\phi'|\rho|\phi'\rangle}{\tr(\Pi_y^{\sgn}\rho)}\leq\frac{1}{\gamma}\frac{\langle\phi|\rho|\phi\rangle}{\tr(\Pi_y^{\sgn}\rho)}&\text{~if~}\langle x,y\rangle=0\\
\frac{1}{2}\frac{\langle\phi|\rho|\phi\rangle}{\tr(\Pi_y^{\sgn}\rho)}&\text{~if~}\langle x,y\rangle=1
\end{cases}
\leq\frac{1}{\gamma}\langle\phi|\rho'|\phi\rangle\,.
\end{equation*}
Thus $|\phi\rangle$ is a $\gamma$-approximate local maximizer of fidelity with $\rho'$.
\end{proof}

\noindent Therefore, we can repeat Steps 1-4 on $\rho'$ recursively. The last thing we need to check is that we can prepare sufficient copies of $\rho'$ from $\rho$ with high probability.

\begin{lemma}\label{lem:prepare_rhop}
    Fix $\tau, \delta>0$, and $N\in \mathbb{N}$. Let $\sigma, \rho$ be two states such that $F(\rho, \sigma)\ge \tau$. Let $\mathfrak{P}=\{\Pi_1,\ldots,\Pi_t\}$ be a set of projectors that stabilize $\sigma$. By measurements and post-selections, we can prepare $N$ copies of state $\rho' = \Pi_t\cdots \Pi_1\rho \Pi_1\cdots \Pi_t/\tr(\Pi_t\cdots \Pi_1\rho)$ using $m_{\mathsf{prepare}}=\frac{2}{\tau}\left(N+\log\frac{1}{\delta}\right)$ copies of $\rho$ with probability at least $1-\delta$. The running time is $O(m_{\mathsf{preprare}}Tt)$, where $T$ is the time for measuring one projector.
\end{lemma}

\begin{proof}
If $t=0$, directly returning $N$ copies of $\rho$ suffices, so the lemma is trivial in this case.
If $t>0$, since $\Pi_i\sigma\Pi_i=\sigma$, the support of $\sigma$ lies entirely in the $+1$ eigenspace of $\Pi_i$.
Thus $\Pi_i\sqrt{\sigma}=\sqrt{\sigma}\Pi_i=\sqrt{\sigma}$.
Hence we have
\[1\geq F(\sigma,\rho')=\left(\tr\sqrt{\sqrt{\sigma}\rho'\sqrt{\sigma}}\right)^2=\frac{\left(\tr\sqrt{\sqrt{\sigma}\rho\sqrt{\sigma}}\right)^2}{\tr(\Pi_t\cdots \Pi_1\rho)}=\frac{F(\rho,\sigma)}{\tr(\Pi_t\cdots \Pi_1\rho)}\geq\frac{\tau}{\tr(\Pi_t\cdots \Pi_1\rho)}.\]
That is, for each copy of $\rho$, the probability of getting the desired sequence of measurement outcomes is at least $\tau$. Let $X_i$ be the indicator that the measurement sequence on the $i$th copy of $\rho_0$ yields the desired result. If we define $\Upsilon=1-\frac{N}{pm_{\mathsf{prepare}}}$, then by Chernoff bound,
\[\Pr[\sum_iX_i\leq N]=\Pr[\sum_iX_i\leq pm_{\mathsf{prepare}}(1-\Upsilon)]\leq e^{-\frac{pm_{\mathsf{prepare}}}{2}\Upsilon^2}\leq e^{-\frac{\tau}{2}m_{\mathsf{prepare}}+N}=\delta.\]
Measuring one projector takes $T$ time, so the total time is $O(m_{\mathsf{prepare}}Tt)$.
\end{proof}

\vspace{0.5em}\noindent \textbf{The full algorithm.}\vspace{0.5em}

\noindent Now we are ready to present the full algorithm, see \Cref{alg:agnostic_learning_stabilizer}.

\begin{algorithm}[htbp]
    \DontPrintSemicolon
    \caption{Agnostic tomography of stabilizer states}\label{alg:agnostic_learning_stabilizer}
    \KwInput{$\tau>0$, $\frac{1}{2}<\gamma\leq1$, copies of an $n$-qubit state $\rho$}
    \KwOutput{A stabilizer state $\ket{\phi}$}
    \Goal{For every $\gamma$-approximate local maximizer $\ket{\phi'}$ with $F(\rho, \ket{\phi})\ge \tau$, $\ket{\phi}=\ket{\phi'}$ with probability at least $\left((\gamma-1/2)\tau\right)^{O\left(\log\frac{1}{\tau}\right)}$}
    Set $\mathfrak{P}_0=\varnothing$, $\mathfrak{R}=\varnothing$, $t_{\max}=\left\lfloor\log_{1.08}\frac{1}{\tau}\right\rfloor$, $\rho_0=\rho$, $\tau_0=\tau$.\\
    \For{$t=0$ \KwTo $t_{\max}$}{
    Prepare $O(\frac{n}{(\gamma-1/2)^2\tau_t^4})$ copies of $\rho_t$ from $\rho$ by \Cref{lem:prepare_rhop} (with $\delta=\frac{1}{6}$, $\mathfrak{P}=\mathfrak{P}_t$). Break the loop if not enough copies are produced.\label{line: state preparation}\\
    Run \Cref{alg:select_high_correlation_Pauli} on $\rho_t$ with $\delta=\frac{1}{5}$ and $\epsilon=\frac{1}{4}(\gamma-1/2)^2\tau_t^4$. Denote the output by $H_t$, which is the basis of a stabilizer family.\label{line: select high-correlation}\\
    Measure $\rho$ on the joint eigenbasis of the stabilizer family $\spn(H_t)$ once. Denote the result $|\phi_t\rangle$. Set $\mathfrak{R}\leftarrow\mathfrak{R}\cup\{|\phi_t\rangle\}$.\label{line:correct_family}\\
    Run Bell difference sampling on $\rho_t$ once. Denote the sample by $y_t$.\label{line:wrong_family}\\
    Randomly pick a sign $\sgn_t\in \{\pm 1\}$\label{line:guess}.

    Define $\mathfrak{P}_{t+1}=\mathfrak{P}_t\cup\{\Pi^{\sgn_t}_{y_t}\}$, $\rho_{t+1}=\Pi_{y_t}^{\sgn_t}\rho_t\Pi_{y_t}^{\sgn_t}/\tr(\Pi_{y_t}^{\sgn_t}\rho_t\Pi_{y_t}^{\sgn_t})$ and $\tau_{t+1}=1.08\tau_t$.\label{line:update_param}
    }
    \Return a uniformly random element $|\phi_r\rangle$ from $\mathfrak{R}$. If $\mathfrak{R}=\varnothing$, return failure.
\end{algorithm}

\subsection{Analysis of the algorithm}
\label{sec:analysis_stab}
We now prove that \Cref{alg:agnostic_learning_stabilizer} satisfies the requirements of \Cref{thm:all_gamma_approximate_local_maximizer}. Fix a $\gamma$-approximate local maximizer $\ket{\phi}$ of fidelity with $\rho$ with fidelity at least $\tau$. The goal is to show that \Cref{alg:agnostic_learning_stabilizer} outputs $\ket{\phi}$ with probability at least $\left((\gamma-1/2)\tau\right)^{O\left(\log\frac{1}{\tau}\right)}$. 

We say the algorithm succeeds up to iteration $t$ if $|\phi_i\rangle=|\phi\rangle$ for some $0\leq i<t$ (i.e., Step 2 succeeds at some iteration), or $W_{y_i}\ket{\phi}=\sgn_i\ket{\phi}$, and/or $\tr(W_{y_i}\rho_i)^2\leq 0.7$ for all $0\leq i<t$ (i.e., the algorithm reaches iteration $t$ without aborting and Step 3 succeeds at every iteration). Denote the event that the algorithm succeeds up to iteration by $B_t$. Then we have the following lemma.

\begin{lemma}\label{lem:correct in the middle}
    If the algorithm succeeds up to iteration $t$, it will succeed up to iteration $t+1$ with probability at least $\frac{1}{4}(\gamma-1/2)^2\tau^4$. Formally speaking, we have
    \begin{equation*}
        \Pr[B_{t+1}|B_t]\ge\frac{1}{4}(\gamma-1/2)^2\tau^4\,.
    \end{equation*}
\end{lemma}
\begin{proof}
When $B_t$ happens, either $|\phi_i\rangle=|\phi\rangle$ for some $0\leq i<t$, in which case $B_{t+1}$ always holds, or $W_{y_i}\ket{\phi}=\sgn_i\ket{\phi}$ and $\tr(W_{y_i}\rho_i)^2\leq 0.7$ for all $0\leq i<t$. In this case, by~\Cref{lem:fidelity_amplification}, $|\phi\rangle$ is a $\gamma$-approximate local maximizer of fidelity with $\rho_t$, with $\langle\phi|\rho_t|\phi\rangle\geq1.08^t\tau=\tau_t$. Thus by~\Cref{lem:prepare_rhop}, at iteration $t$, with probability at least $\frac{5}{6}$~\Cref{line: state preparation} does not abort and we get the desired number of $\rho_t$. By~\Cref{lem:high_correlation_space}, with probability at least $\frac{4}{5}$,~\Cref{line: select high-correlation} returns a $\frac{1}{4}(\gamma-1/2)^2\tau_t^4$-high-correlation basis $H_t$. If $H_t=\Weyl(|\phi\rangle)$, $|\phi_t\rangle=|\phi\rangle$ with probability at least $\tau$. If $H_t\neq\Weyl(|\phi\rangle)$, then by~\Cref{lem:sample_low_correlation_Pauli} the probability of $\tr(W_y\rho)^2\leq 0.7\wedge W_y\ket{\phi}=\sgn\ket{\phi}$ is at least $\frac{3}{8}(\gamma-1/2)^2\tau_t^4$. To sum up, we have
\begin{align*}
    \Pr[B_t|B_{t-1}] &\geq\min\left\{1,\frac{5}{6}\times\frac{4}{5}\min\left\{\tau,\frac{3}{8}(\gamma-1/2)^2\tau_t^4\right\}\right\}\geq\frac{1}{4}(\gamma-1/2)^2\tau^4\,. \qedhere
\end{align*}
\end{proof}

\noindent Since we cannot find stabilizers with low correlations indefinitely, at some point we will get $|\phi\rangle$, allowing us to prove our main result:

\begin{proof}[Proof of~\Cref{thm:all_gamma_approximate_local_maximizer}]
Since $B_0$ holds trivially, $\Pr[B_0]=1$. By~\Cref{lem:correct in the middle}, we have
\[\Pr[B_{t_{\max}+1}]\geq\Bigl(\frac{1}{4}(\gamma-1/2)^2\tau^4\Bigr)^{t_{\max}+1}\,,\]
where $t_{\max}$ is defined in Line 1 of~\Cref{alg:agnostic_learning_stabilizer}. Note that $B_{t_{\max}+1}$ means the event that $|\phi\rangle\in\mathfrak{R}$ or $W_{y_i}\ket{\phi}=\sgn_i\ket{\phi}$ and $\tr(W_{y_i}\rho_i)^2\leq 0.7$ for all $0\leq i\leq t_{\max}$. But if the later case happens, by~\Cref{lem:fidelity_amplification} we have $\langle\phi|\rho_{t_{\max}+1}|\phi\rangle\geq1.08^{t_{\max}+1}\tau>1$, which is impossible. Thus when $B_{t_{\max}+1}$ happens, $|\phi\rangle\in\mathfrak{R}$. Since $|\mathfrak{R}|\leq t_{\max}+1$, we have
\[\Pr[|\phi_r\rangle=|\phi\rangle]\geq\frac{\Pr[B_{t_{\max}+1}]}{t_{\max}+1}\geq\frac{1}{1+\log_{1.08}\frac{1}{\tau}}\Bigl(\frac{1}{4}(\gamma-1/2)^2\tau^4\Bigr)^{1+\log_{1.08}\frac{1}{\tau}}=\left((\gamma-1/2)\tau\right)^{O\left(\log\frac{1}{\tau}\right)}.\]

During each iteration,~\Cref{line: select high-correlation} uses at most $\frac{cn}{(\gamma-1/2)^2\tau_t^4}$ copies of $\rho_t$ for some constant $c$, while~\Cref{line:correct_family} uses $1$ copy of $\rho$ and~\Cref{line:wrong_family} uses $4$ copies of $\rho_t$. Hence it suffices to prepare $\frac{2cn}{(\gamma-1/2)^2\tau_t^4}$ copies of $\rho_t$ from $\frac{2}{\tau}(\frac{2cn}{(\gamma-1/2)^2\tau_t^4}+\log6)$ copies of $\rho$ at~\Cref{line: state preparation}. The total number of copies required is thus
\[\sum_{t=0}^{t_{\max}}1+\frac{2}{\tau}\Bigl(\frac{2cn}{(\gamma-1/2)^2\tau_t^4}+\log6\Bigr)=O\Bigl(\frac{n}{(\gamma-1/2)^2\tau^5}\Bigr)\,.\]

As for running time, during each iteration,~\Cref{line: state preparation} runs in $\frac{c_1n^2}{(\gamma-1/2)^2\tau\tau_t^4}t$ time for some constant $c_1$ since measuring projectors that correspond to Pauli eigenspaces takes $O(n)$ time. \Cref{line: select high-correlation} runs in time $\frac{c_2n^2}{(\gamma-1/2)^2\tau_t^4}(n+\log\frac{1}{(\gamma-1/2)^2\tau_t^4})$ for some constant $c_2$ by~\Cref{lem:high_correlation_space}. By~\Cref{lem:Clifford_Synthesis}, the Clifford circuit for measuring in $H_t$ basis is of size $O(n^2)$, so~\Cref{line:correct_family} takes $c_3n^2$ time for some constant $c_3$. \Cref{line:wrong_family} takes $c_4n$ time for some constant $c_4$ and~\Cref{line:guess,line:update_param} takes $c_5$ time for some constant $c_5$. Thus the total running time is
\begin{align*}
\MoveEqLeft \sum_{t=0}^{t_{\max}}\frac{c_1n^2}{(\gamma-1/2)^2\tau\tau_t^4}t+\frac{c_2n^2}{(\gamma-1/2)^2\tau_t^4}\Bigl(n+\log\frac{1}{(\gamma-1/2)^2\tau_t^4}\Bigr)+c_3n^2+c_4+c_5\\
&=O\Bigl(\frac{n^2}{(\gamma-1/2)^2\tau^4}\Bigl(n+\log\frac{1}{\gamma-1/2}+\frac{1}{\tau}\Bigr)\Bigr)\,. \qedhere
\end{align*}
\end{proof}

\section{Agnostic tomography of quantum states with high stabilizer dimension.}
\label{sec:dimension}
In this section, we apply our technique to a more general setting: agnostic tomography of quantum states with stabilizer dimension $n-t$ for $t=O(\log n)$. Let $\cS^{n-t}$ be the set of states with stabilizer dimension at least $n-t$. 

A state $\sigma\in \cS^{n-t}$ can be described by $C^\dagger (\ketbra{0^{n-t}}\otimes \sigma_0)C$ for some Clifford gate $C$ and some density matrix of a $t$-qubit state $\sigma_0$. Indeed, suppose $P$ is a $(n-t)$-dimensional stabilizer group that stabilizes $\sigma$. According to \Cref{lem:Clifford_Synthesis}, there exists a Clifford gate $C$ that maps $P$ to $\{I, Z\}^{\otimes n-t}\otimes I^{\otimes t}$. Then $C\sigma C^\dagger$ is stabilized by $\{I, Z\}^{\otimes n-t}\otimes I^{\otimes t}$, thus it has the form $\ketbra{s}\otimes \sigma_0$ for some $s\in \{0, 1\}^{n-t}$ and a $t$-qubit state $\sigma_0$. Furthermore, we can let $s=0^{n-t}$ by absorbing some $X$-gates to $C$. This classical description is efficient when $t=O(\log n)$.

\begin{theorem}\label{thm:agnostic_learning_high_stabilizer_dimension_states}
    Fix $t\in \mathbb{N}$, $\tau\ge\epsilon>0, \delta>0$, and let $\rho$ be an unknown $n$-qubit state. There is an algorithm with the following guarantee.
    
    Given copies of $\rho$ with $F(\rho, \cS^{n-t})\ge \tau$, returns a state $\sigma\in \cS^{n-t}$ (described by $C^\dagger(\ketbra{0^{n-t}}\otimes \sigma_0)C$) such that $F(\rho, \sigma)\ge F(\rho, \cS^{n-t})-\epsilon$ with probability at least $1-\delta$.
    
    The algorithm only performs single-copy measurements and two-copy measurements on $\rho$. The sample complexity is $n(2^t/\tau)^{O(\log 1/\epsilon)}\log(1/\delta)$ and the time complexity is $n^3(2^t/\tau)^{O(\log 1/\epsilon)}\log^2(1/\delta)$.
\end{theorem}

\noindent Ref.~\cite{grewal2023efficient} provides an algorithm with $\poly(n, 2^t, 1/\epsilon)$ runtime and sample complexity for learning states with stabilizer dimension $n - t$ in the realizable setting, i.e. when $\tau = 1$. Here we extend their result to the agnostic setting. Our algorithm remains efficient when $\tau, \epsilon=\Omega(1)$, and is quasipolynomial when $\tau, \epsilon=1/\poly(n)$. When $\tau=1$ and $\epsilon=1/\poly(n)$, our algorithm is inefficient, which is worse than \cite{grewal2023efficient}. However, a slight modification recovers the $\poly(n, 2^{t}, 1/\epsilon)$ complexity when $1/\tau\leq 1+c\epsilon$ for some constant $c$, see \Cref{remark: high dim 1}. 

An important application is to (improper) agnostic tomography of $t$-doped quantum states, as a $t$-doped quantum state has stabilizer dimension at least $n-2t$ \cite[Lemma 4.2]{grewal2024improved}:

\begin{corollary}\label{cor:agnostic_learning_high_stabilizer_dimension_states}
    Fixed $t\in\mathbb{N}$, $\tau\ge \epsilon>0$, $\delta>0$. There is an algorithm that, given copies of an $n$-qubit state $\rho$ such that $F(\rho, \ket{\phi})\ge \tau$ for some $t$-doped state $\ket{\phi}$, returns a state $\sigma\in\mathcal{S}^{n-2t}$ such that $F(\rho, \sigma)\ge \tau-\epsilon$ with probability at least $1-\delta$. The algorithm uses $n(2^t/\tau)^{O(\log\frac{1}{\epsilon})}\log(1/\delta)$ copies of $\rho$ and $n^3(2^t/\tau)^{O(\log\frac{1}{\epsilon})}\log(1/\delta)$ time.
\end{corollary}

\noindent It is worth emphasizing that we do not have to know $t$ beforehand as we can enumerate $t$ from $0$ to $O(\log(n))$. 

\subsection{Reduce to finding the correct Clifford unitary}\label{sec: reduce to find Clifford}

Define the $t$-qubit state $\rho^{C}_{n-t}\triangleq \braket{0^{n-t}|C\rho C^\dagger|0^{n-t}}/\tr(\braket{0^{n-t}|C\rho C^\dagger|0^{n-t}})$. According to \Cref{lem: fidelity high dimension} and the fact that fidelity is invariant under unitary,
\begin{equation}\label{eq: closest high dimension state fidelity}
    F(\rho, \sigma)=F(C\rho C^\dagger, \ketbra{0^{n-t}}\otimes \sigma_0)=\tr(\braket{0^{n-t}|C\rho C^\dagger|0^{n-t}})F(\rho^C_{n-t}, \sigma_0).
\end{equation}
Therefore, given $C$, the optimal choice of $\sigma_0$ is $\rho^{C}_{n-t}$ and the optimal fidelity is $\tr(\braket{0^{n-t}|C\rho C^\dagger|0^{n-t}})$. We can calculate the description of $\rho^C_{n-t}$ via full state tomography, by the following lemma whose proof is deferred to Appendix~\ref{sec:defer_tomo_C}.

\begin{lemma}\label{lem: full tomography given C}
    Let $\tau, \epsilon, \delta>0$, $t\in \mathbb{N}$, $\rho$ be an $n$-qubit state, and $C$ be a Clifford gate such that $\tr(\braket{0^{n-t}|C\rho C^\dagger|0^{n-t}})\ge \tau$. Given access to copies of $\rho$, there exists an algorithm that outputs the density matrix of a state $\sigma_0$ such that $F(\rho_{n-t}^C, \sigma_0)\ge 1-\epsilon$ with probability at least $1-\delta$. The algorithm performs $2^{O(t)}\log(1/\delta)/\epsilon^2\tau$ single-copy measurements on $\rho$ and takes $2^{O(t)}n^2\log(1/\delta)/\epsilon^2\tau$ time.
\end{lemma}

\noindent As a result, to find a state $\sigma=C^\dagger(\ketbra{0^{n-t}}\otimes \sigma_0)C\in \cS^{n-t}$ that is closest to $\rho$, it suffices to find the correct Clifford gate $C$. To do this, we will exhibit an algorithm which finds $C$ with non-negligible probability, at which point we can repeat multiple times and invoke Lemma~\ref{lem: full tomography given C} to obtain an algorithm for agnostic tomography. We encapsulate this reduction in the following, whose proof is deferred to Appendix~\ref{sec:defer_reduce_cliff}.

\begin{lemma}\label{lem: reduce to find Clifford}
    Fix $t\in\mathbb{N}, \tau\ge \epsilon>0, \delta>0$. Given copies of an $n$-qubit state $\rho$ such that $F(\rho, \cS^{n-t})\ge \tau$, if there exists an algorithm $\mathcal{A}$ that outputs a Clifford $C$ such that $\tr(\braket{0^{n-t}|C\rho C^\dagger|0^{n-t}})\ge F(\rho, \cS^{n-t})-\epsilon/3$ with probability at least $p$ using $S$ copies and $T$ time, then
    \begin{enumerate}[label=(\alph*)]
        \item there exists an algorithm that outputs a Clifford $C$ such that $\tr(\braket{0^{n-t}|C\rho C^\dagger|0^{n-t}})\ge F(\rho, \cS^{n-t})-2\epsilon/3$ with probability at least $1-\delta$ using $O(\frac{S}{p}\log\frac{1}{\delta}+\frac{2^{O(t)}}{\epsilon^2}\log\frac{1}{p\delta})$ copies and $O(\frac{T}{p}\log\frac{1}{\delta}+\frac{2^{O(t)}n^2\log(1/\delta)}{\epsilon^2p}\log\frac{1}{p\delta})$ time.
        \item there exists an algorithm that outputs a state $\sigma=C^\dagger(\ketbra{0^{n-t}}\otimes \sigma_0)C\in \cS^{n-t}$ such that $F(\rho, \sigma)\ge F(\rho, \cS^{n-t})-\epsilon$ with probability at least $1-\delta$ using $O(\frac{S}{p}\log\frac{1}{\delta}+\frac{2^{O(t)}}{\epsilon^2}\log\frac{1}{p\delta}+\frac{2^{O(t)}}{\epsilon^2\tau}\log\frac{1}{\delta})$ copies and $O(\frac{T}{p}\log\frac{1}{\delta}+\frac{2^{O(t)}n^2\log(1/\delta)}{\epsilon^2p}\log\frac{1}{p\delta}+\frac{2^{O(t)}n^2}{\epsilon^2\tau}\log\frac{1}{\delta})$ time.
    \end{enumerate}
\end{lemma}

\subsection{Construction of the algorithm}

Throughout, let $\sigma^*\in \cS^{n-t}$ be the closest state to $\rho$ in $\cS^{n-t}$.
The workflow of the algorithm will be similar to \Cref{alg:agnostic_learning_stabilizer}, with the crucial technical complication that we cannot actually hope to obtain a full set of generators for $\Weyl(\sigma^*)$ due to difficulties in implementing Step 3. Instead, our notion of complete projectors needs to be modified: we enter Step 2 if the span of the Weyl operators we have found \emph{has large intersection} with $\Weyl(\sigma^*)$. Without a full set of generators, we need a more involved procedure than simply measuring in their joint eigenbasis. We elaborate upon this in the sequel.

\vspace{0.5em}\noindent \textbf{Step 1: Find a high-correlation family.}\vspace{0.5em}

\noindent The first step is the same as in \Cref{alg:agnostic_learning_stabilizer}. We run \Cref{alg:select_high_correlation_Pauli} to find a high-correlation basis $H$.

\vspace{0.5em}\noindent \textbf{Step 2: If the family is complete, i.e. if $\dim(\spn(H)\cap \Weyl(\sigma^*))$ is large, then obtain the answer.}\vspace{0.5em}

\noindent When $\dim(\spn(H)\cap \Weyl(\sigma^*))$ is at least $n-t'$, we give an algorithm to obtain the correct Clifford unitary $C$ in \Cref{lem: high dimension step 2}. Here $t'=4\log(1/\tau)+(2t+2)$ is larger than $t$ due to the inadequacy of Step 3 (see \Cref{lem:weaker_evenly_distribution}). As mentioned above, this renders our analysis of Step 2, specifically the proof of \Cref{lem: high dimension step 2}, much more complicated because $\spn(H)$ does not contain full information about $\Weyl(\sigma^*)$. This part of the proof is involved and we defer the details to \Cref{sec: high dimension step 2}.

\begin{definition}
    Let $t'\ge t\in\mathbb{N}$, $H$ be a basis of a stabilizer family. Define $H^{n-t'}_{n-t}=\{\sigma\in \cS^{n-t}:\dim(\spn(H)\cap \Weyl(\sigma))\ge n-t'\}$, i.e., the set of states in $\cS^{n-t}$ that are stabilized by an $(n-t')$-dimensional subspace of $\spn(H)$.
\end{definition}

\begin{lemma}\label{lem: high dimension step 2}
    Fix $t'\ge t\in \mathbb{N}$, $\tau\ge\epsilon>0$, and $\delta>0$. There is an algorithm that, given copies of an $n$-qubit state $\rho$ and a basis $H$ of a stabilizer family such that $F(\rho, H^{n-t'}_{n-t})\ge \tau$, output a Clifford circuit $C$ such that $\tr(\braket{0^{n-t}|C\rho C^\dagger|0^{n-t}})\ge F(\rho, H^{n-t'}_{n-t})-\epsilon$ with probability at least $1-\delta$. The algorithm uses $(1/\epsilon)^{O(t')}\log(1/\delta)$ copies and $n^3(1/\epsilon)^{O(t')}\log^2(1/\delta)$ time. 
\end{lemma}

\noindent In particular, if $\dim(\spn(H)\cap \Weyl(\sigma^*))\ge n-t'$, then $\sigma^*\in H^{n-t'}_{n-t}\subseteq \cS^{n-t}$ by definition. Since $\sigma^*$ is the closest state to $\rho$ in $\cS^{n-t}$, we have $F(\rho, H^{n-t'}_{n-t})=F(\rho, \sigma^*)=F(\rho, \cS^{n-t})\ge \tau$. Therefore, the output $C$ satisfies $\tr(\braket{0^{n-t}|C\rho C^\dagger|0^{n-t}})\ge F(\rho, H^{n-t'}_{n-t})-\epsilon=F(\rho, \cS^{n-t})-\epsilon$.

\vspace{0.5em}\noindent \textbf{Step 3: If the family is incomplete, sample a low-correlation projector.}\vspace{0.5em}

\noindent The goal of Step 3 is to sample a low correlation Pauli string in $\Weyl(\sigma^*)$ when $\dim(\spn(H)\cap \Weyl(\sigma))$ is small. In our analysis for stabilizer states, we used~\Cref{lem:sample_low_correlation_Pauli} to ensure that a single sample from $\cB_\rho$ is likely to be a low correlation Pauli string in $\Weyl(\ket{\phi})$ when $\ket{\phi}$ is a $\gamma$-approximate local maximizer. At the core of this proof was the anti-concentration property of Bell difference sampling. However, this is not the case for agnostic tomography of states with high stabilizer dimension. We can only prove the following weaker version of \Cref{lem:sample_low_correlation_Pauli}:

\begin{lemma}\label{lem:weaker_evenly_distribution}
    Let $\sigma\in \cS^{n-t}$ be a state with $F(\rho, \sigma)\ge \tau$ and $H$ be a $\frac{\tau^4}{2^{2t+2}}$-high-correlation basis in the sense of Definition~\ref{def:high_cor_basis}. If $\dim(\spn(H)\cap \Weyl(\sigma))\leq n-4\log(1/\tau)-(2t+2)$, we have
    \begin{equation*}
        \Pr_{y\sim \cB_\rho, \sgn\sim \{\pm 1\}}[\tr(W_y\rho)^2\leq 0.7\wedge W_y\sigma=\sgn\sigma] \ge \frac{\tau^4}{2^{2t+2}}.
    \end{equation*}
\end{lemma}
\begin{proof}
    Denote $P=\spn(H)\cap \Weyl(\sigma)$. By definition of $H$, $\Pr_{y\sim \cB_\rho}[\tr(W_y\rho^2)> 0.7\wedge y\not\in \spn(H)]\leq \frac{\tau^4}{2^{2t+2}}$. By \Cref{lem: Bell distribution small}, $\Pr_{y\sim \cB_\rho}[y\in P]\leq \frac{\abs{P}}{2^{n}}\leq \frac{\tau^4}{2^{2t+2}}$. By \Cref{{lem: Bell difference sampling not small high dimension}}, $\Pr_{\cB_\rho}[\Weyl(\sigma)]\ge \frac{\tau^4}{2^{2t}}$. Therefore
    \begin{align*}
        \MoveEqLeft\Pr_{y\sim \cB_\rho, \sgn\sim \{\pm 1\}}[\tr(W_y\rho)^2\leq 0.7\wedge W_y\sigma=\sgn\sigma]\\
        &=\frac12 \Pr_{y\sim \cB_\rho}[\tr(W_y\rho)^2\leq 0.7\wedge y\in \Weyl(\sigma)]\\
        &=\frac12\left(\Pr_{y\sim\cB_\rho}[y\in \Weyl(\sigma)]-\Pr_{y\sim\cB_\rho}[\tr(W_y\rho)^2>0.7\wedge y\in \Weyl(\sigma)]\right)\\
        &\ge \frac12 \left(\frac{\tau^4}{2^{2t}}-\Pr_{y\sim \cB_\rho}[\tr(W_y\rho)^2> 0.7\wedge y\in P]-\Pr_{y\sim \cB_\rho}[\tr(W_y\rho)^2> 0.7\wedge y\in \Weyl(\sigma)\wedge y\not\in \spn(H)]\right)\\
        &\ge \frac12\left(\frac{\tau^4}{2^{2t}}-\Pr_{y\sim \cB_\rho}[y\in P]-\Pr_{y\sim \cB_\rho}[\tr(W_y\rho)^2> 0.7\wedge y\not\in \spn(H)]\right)\\
        &\ge \frac12\left(\frac{\tau^4}{2^{2t}}-\frac{\tau^4}{2^{2t+2}}-\frac{\tau^4}{2^{2t+2}}\right)=\frac{\tau^4}{2^{2t+2}}\,.\qedhere
    \end{align*}
\end{proof}

Define $t'=4\log(1/\tau)+(2t+2)$. According to \Cref{lem:weaker_evenly_distribution}, if $\dim(H\cap \Weyl(\sigma^*))\leq n-t'$, Bell difference sampling is likely to produce a low correlation Pauli string in $\Weyl(\sigma^*)$, which allows us to apply our bootstrapping-by-measurement technique in Step 4.

\vspace{0.5em}\noindent \textbf{Step 4: Bootstrap by measuring.}\vspace{0.5em}

\noindent Assume Step 3 gives a low-correlation signed Pauli string that stabilizes $\sigma^*$. The last step is to measure, post-select, and recurse. \Cref{lem: high dim Step 4} and \Cref{lem: prepare state high dimension} exhibit all properties required for this step. 

\begin{lemma}\label{lem: high dim Step 4}
    Fix $t\in\mathbb{N}, \tau>0$ and a $n$-qubit state $\rho$ such that $F(\rho, \cS^{n-t})\ge \tau$. Let $\sigma^*$ be a state in $\cS^{n-t}$ that is closest to $\rho$, i.e., $F(\rho, \sigma^*)=F(\rho, \cS^{n-t})$. Suppose $C$ is a Clifford gate such that $C^\dagger Z_1 C \sigma^*=\sigma^*$ and $\tr(C^\dagger Z_1 C \rho)^2\leq 0.7$. Define $\rho_0=\braket{0|C\rho C^\dagger|0}/\tr(\braket{0|C\rho C^\dagger|0})$.
    \begin{enumerate}[label=(\alph*)]
        \item $C\sigma^* C^\dagger$ has the form $\ketbra{0}\otimes \sigma^*_0$ for some $\sigma^*_0\in \cS^{n-t-1}_{n-1}$. Furthermore, $\sigma^*_0$ is one of the closest state to $\rho_0$ in $\cS^{n-1-t}_{n-1}$, i.e., $F(\rho_0, \sigma^*_0)=F(\rho_0, \cS^{n-t-1}_{n-1})$.
        \item $F(\rho_0, \cS^{n-t-1}_{n-1})\ge 1.08F(\rho, \cS^{n-t}_{n})$
    \end{enumerate}
\end{lemma}

\begin{lemma}\label{lem: prepare state high dimension}
    For $k, N\in \mathbb{N}, \tau, \delta>0$, set $N'=\frac{2}{\tau}(N+\log(\frac{1}{\delta}))$. Given $N'$ copies of an $n$-qubit state $\rho$ and the classical description of a Clifford gate $C$ such that $\tr(\braket{0^k|C\rho C^\dagger|0^k})\ge \tau$,
    we can prepare $N$ copies of $\rho_0=\braket{0^k|C\rho C^\dagger|0^k}/\tr(\braket{0^k|C\rho C^\dagger|0^k})$ with probability at least $1-\delta$ using $O(n^2 N')$ time.
\end{lemma}

\noindent To interpret the two lemmas, let's assume Step 3 returns a low correlation signed Pauli string $(\sgn, y)$ that stabilizes $\sigma^*$. We find a Clifford gate $C$ such that $C(\sgn W_y)C^\dagger = Z_1$. Then we measure the first qubit of $C\rho C^\dagger$ on the computational basis and post-select on the outcome $0$. \Cref{lem: high dim Step 4} tells us that the closest state to $\rho_0$ in $\cS_{n-1}^{n-1-t}$ is still $\sigma^*_0$, and the fidelity $F(\rho_0, \sigma^*_0)$ is amplified by a constant factor. So we can recurse on $\rho_0$ with a higher fidelity. \Cref{lem: prepare state high dimension} ensures that we can prepare a sufficient number of $\rho_0$ for recursion.
\begin{proof}[Proof of \Cref{lem: high dim Step 4}]
    (a) $C\sigma^* C^\dagger=C(C^\dagger Z_1C\sigma^*)C^\dagger=Z_1C\sigma^* C^\dagger$. So $C\sigma^* C^\dagger$ has the from $\ketbra{0}\otimes \sigma_0^*$ for some $(n-1)$-qubit state. Since $\sigma^*$ is stabilized by a $(n-t)$-dimensional stabilizer group, $\sigma^*_0$ is stabilized by a $(n-1-t)$-dimensional stabilizer group, i.e., $\sigma^*_0\in \cS^{n-t-1}_{n-1}$. By \Cref{lem: fidelity high dimension}, $F(\rho, \sigma^*)=F(C\rho C^\dagger, C\sigma^* C^\dagger)=\tr(\braket{0|C\rho C^\dagger|0})F(\rho_0, \sigma^*_0)$. If there is a state $\sigma_0\in \cS^{n-t-1}_{n-1}$ such that $F(\rho_0, \sigma_0)>F(\rho_0, \sigma^*_0)$, then $$F(\rho , C^\dagger (\ketbra{0}\otimes \sigma_0)C)=\tr(\braket{0|C\rho C^\dagger|0})F(\rho_0, \sigma_0)>F(\rho, \sigma^*)=F(\rho, \cS^{n-t}),$$ 
    a contradiction. So $\sigma^*_0$ is one of the closest states to $\rho_0$ in $\cS^{n-t-1}_{n-1}$.\\
    \indent (b) Since $F(\rho_0, \cS_{n-1}^{n-1-t})=F(\rho_0, \sigma^*_0)=F(\rho, \sigma^*)/\tr(\braket{0|C\rho C^\dagger|0})=F(\rho, \cS_n^{n-t})/\tr(\braket{0|C\rho C^\dagger|0})$, we only need to prove that $\tr(\braket{0|C\rho C^\dagger|0})\leq 1/1.08$. Indeed, 
    \begin{equation*}
        0.7\ge \tr(C^\dagger Z_1 C\rho)^2=\abs{\tr(\braket{0|C\rho C^\dagger|0})-\tr(\braket{1|C\rho C^\dagger|1})}^2=\abs{2\tr(\braket{0|C\rho C^\dagger|0})-1}^2.
    \end{equation*}
    So $\tr(\braket{0|C\rho C^\dagger|0})\leq (\sqrt{0.7}+1)/2\leq 1/1.08$.
\end{proof}

\begin{proof}[Proof of \Cref{lem: prepare state high dimension}]
    From \Cref{lem:prepare_rhop} (where $\mathfrak{P}=\{C^\dagger \ketbra{0^k} C\}$, $\sigma=C^\dagger (\ketbra 0^{k}\otimes \rho_0)C$), using $N'$ copies of $\rho$ and $O(n^2 N')$ time ($n^2$ comes from the cost of applying Clifford gate), with probability at least $1-\delta$, we can prepare $N$ copies of
    \begin{equation*}
        \rho'\triangleq \frac{C^\dagger \ketbra{0^k}C \rho C^\dagger \ketbra{0^k}C}{\tr(C^\dagger \ketbra{0^k}C \rho C^\dagger \ketbra{0^k}C)}=\frac{C^\dagger (\ketbra{0^k}\otimes \braket{0^k|C\rho C^\dagger|0^k})C}{\tr(\braket{0^k|C\rho C^\dagger|0^k})}=C^\dagger (\ketbra{0^k}\otimes \rho_0) C.
    \end{equation*}
    Tracing out the first $k$ qubit of $C\rho'C^\dagger$, we obtain $N$ copies of $\rho_0$.
\end{proof}

\vspace{0.5em}\noindent \textbf{The full algorithm}\vspace{0.5em}\\ We present the full algorithm in \Cref{alg:agnostic_learning_states_high_stabilizer_dimension}.

\begin{algorithm}[htbp]
    \DontPrintSemicolon
    \caption{Agnostic tomography of states with high stabilizer dimension}\label{alg:agnostic_learning_states_high_stabilizer_dimension}
    \KwInput{$t\in \mathbb{N}, \tau\ge \epsilon>0$, copies of an $n$-qubit state $\rho$}
    \Promise{$F(\rho, \cS^{n-t})\ge \tau$}
    \KwOutput{A Clifford gate $C$}
    \Goal{With probability at least $2(\tau^4/2^{2t+4})^{k_{\max}}/(3(k_{\max}+1))$ ($k_{\max}$ defined below), $\tr(\braket{0^{n-t}|C\rho C^\dagger|0^{n-t}})\ge F(\rho, \cS^{n-t})-\epsilon$.}
    Set $\mathfrak{R}=\emptyset$, $k_{\max}=\lfloor\log_{1.08}(1/\tau)\rfloor+1$, $C_0=I^{\otimes n}$, $t'=2t+2+4\log(1/\tau)$.\\
    \For{$k=0$ \KwTo $k_{\max}$}{
    Define $\tau_k=1.08^{k}\tau, \epsilon_k=1.08^k \epsilon, \rho_k=\braket{0^k|C_k \rho C_k^\dagger|0^k}/\tr(\braket{0^k|C_k \rho C_k^\dagger|0^k})$.\\
    Use $\frac{2}{\tau}(m_1+m_2+m_3+\log(\frac{3}{2}))$ copies of $\rho$ to prepare $\rho_k$ by Lemma~\ref{lem: prepare state high dimension}. Break the loop if the number of $\rho_k$ is less than $m_1+m_2+m_3$. Here $m_1$, $m_2$, and $m_3$ are the number of samples of the next three lines, respectively.\label{line: high line 4}\\
    Run \Cref{alg:select_high_correlation_Pauli} on $\rho_k$ (with $(\epsilon, \delta)$ set to $(\tau^4/2^{2t+2}, 1/3)$). Denote the output by $H_k$. Break the loop if \Cref{alg:select_high_correlation_Pauli} fails.\label{line: high line 5}\\
    Run \Cref{lem: high dimension step 2} on $\rho_k$ and $H_k$ (with $(n, t', t, \tau, \epsilon, \delta)$ set to $(n-k, t', t, \tau_k, \epsilon_k, 1/3)$). The output is an $(n-k)$-qubit Clifford gate $U_k$. Define $R_k=(I^{\otimes k}\otimes U_k)C_k$. Add $R_k$ to $\mathfrak{R}$.\label{line: high line 6}\\
    Bell difference sampling on $\rho_t$ once. Denote the sample by $Q_k$. Randomly select a sign $\sgn_k\in \{-1, 1\}$.\label{line: high line 7}\\
    Find a $(n-k)$-qubit Clifford gate $V_k$ such that $V_k \sgn_k Q_k V_k^\dagger = Z_1$.\\
    Define $C_{k+1}=(I^{\otimes k}\otimes V_k)C_k$.\\
    }
    \Return a uniformly random element from $\mathfrak{R}$. If $\mathfrak{R}=\emptyset$, return $I^{\otimes n}$.
\end{algorithm}

\subsection{Analysis of the algorithm}\label{sec: high dimension analysis}
Fix a state $\sigma^*\in \cS^{n-t}$ such that $F(\rho, \sigma^*)=F(\rho, \cS^{n-t})$. Define $\sigma^*_k=\braket{0^k|C_k\sigma^* C_k^\dagger|0^k}/\tr(\braket{0^k|C_k\sigma^* C_k^\dagger|0^k})$.
To better demonstrate how the algorithm works, we define the following events for every $0\leq k\leq k_{\max}$:
\begin{enumerate}
    \item We say the algorithm correctly proceeds to iteration $k$ if it does not break the loop before iteration $k$ and $\tr(\rho_r Q_r)^2\leq 0.7$ and $Q_r\sigma_r^*=\sgn_r \sigma_r^*$ for every $0\leq r\leq k$. Denote the event by $A_k$.
    \item We say the algorithm succeeds at iteration $k$ if it correctly proceeds to iteration $k-1$, does not break the loop at iteration $k$, and $\dim(\spn(H_k)\cap \Weyl(\sigma^*_k))\ge n-k-t'$. Denote the event by $B_k$.
    \item Define $E_k=B_0\vee B_1\cdots \vee B_k$.
\end{enumerate}

For convenience, we define $A_{-1}$, $B_{-1}$, and $E_{-1}$ to be the whole probability space.

\begin{lemma}\label{lem: high stabilizer  events}
    For the events defined above, we have:
    \begin{enumerate}[label=(\alph*)]
        \item If $A_{k-1}$ happens, then $C_k\sigma^* C_k^\dagger=\ketbra{0^{k}}\otimes \sigma_k^*$ and 
        \begin{equation*}
            F(\rho_k, \sigma^*_k)= F(\rho_k, \cS_{n-k}^{n-k-t})\ge 1.08^k F(\rho, \cS_{n}^{n-t})\ge 1.08^k\tau\,.
        \end{equation*}
        \item $\Pr[A_k\vee E_k|A_{k-1}\vee E_{k-1}]\ge \tau^4/2^{2t+4},~\forall 0\leq k\leq k_{\max}$.
        \item $\Pr[A_{k_{\max}-1}]=0$.
        \item $\Pr[E_{k_{\max}-1}]\ge (\tau^4/2^{2t+4})^{k_{\max}}$.
        \item If $E_{k_{\max}-1}$ happens, with probability at least $2/(3(k_{\max}+1))$, the output of the algorithm is a Clifford gate $C$ such that $\tr(\braket{0^{n-t}|C\rho C^\dagger|0^{n-t}})\ge F(\rho, \sigma^*)-\epsilon$.
        \item The output of the algorithm is a Clifford gate $C$ such that $\tr(\braket{0^{n-t}|C\rho C^\dagger|0^{n-t}})\ge F(\rho, \sigma^*)-\epsilon$ with probability at least $2(\tau^4/2^{2t+4})^{k_{\max}}/(3(k_{\max}+1))$.
    \end{enumerate}
\end{lemma}
\begin{proof}
    (a) If $A_{k-1}$ happens, we prove by induction that for every $0\leq r\leq k$,
    \begin{equation}\label{eq: high stabilizer events 1}
        C_r\sigma^* C_r^\dagger=\ketbra{0^r}\otimes \sigma_r^*,~F(\rho_r, \sigma^*_r)=F(\rho_r, \cS_{n-r}^{n-r-t})\ge 1.08^r F(\rho, \cS_{n}^{n-t})\ge 1.08^r\tau=\tau_{r}.
    \end{equation}
    When $r=0$, \eqref{eq: high stabilizer events 1} is trivial. Assume \eqref{eq: high stabilizer  events 1} is true for $r-1$. By definition of $\rho_{r-1}, \rho_r$, we have
    \begin{align*}
        \rho_r&\propto \braket{0^r|C_r\rho C_r^\dagger|0^r}=\braket{0^r|(I^{\otimes r-1} \otimes V_{r-1})C_{r-1}\rho C_{r-1}^\dagger (I^{\otimes r-1} \otimes V_{r-1}^\dagger)|0^r}\\
        &\propto \braket{0|V_{k-1}\rho_{r-1}V_{k-1}^\dagger|0}.
    \end{align*}
    Similarly, $\sigma^*_r\propto \braket{0|V_{k-1}\sigma^*_{r-1}V_{k-1}^\dagger|0}$.
    So $\rho_r=\braket{0|V_{k-1}\rho_{r-1}V_{k-1}^\dagger|0}/\tr(\braket{0|V_{k-1}\rho_{r-1}V_{k-1}^\dagger|0})$. By the induction hypothesis, $F(\rho_{r-1}, \sigma^*_{r-1})\ge \tau_{r-1}$. Since $A_{r-1}$ happens, $V_{r-1}$ satisfies $V_{r-1}^\dagger Z_1 V_{r-1}\sigma^*_{r-1}=\sgn_{r-1}Q_{r-1}\sigma^*_{r-1}=\sigma^*_{r-1}$ and $\tr(V_{r-1}^\dagger Z_1 V_{r-1}\rho_{r-1})^2=\tr(Q_{r-1}\rho_{r-1})^2\leq 0.7$. \eqref{eq: high stabilizer events 1} follows from \Cref{lem: high dim Step 4} (where $(\rho, \sigma^*, C)\leftarrow (\rho_{r-1}, \sigma^*_{r-1}, V_{r-1})$), $F(\rho_r, \sigma^*_r)=F(\rho_r, \cS_{n-r}^{n-r-t})\ge 1.08 F(\rho_{r-1}, \cS_{n-(r-1)}^{n-(r-1)-t})$.

    (b) Conditioned on $A_{k-1}\vee E_{k-1}$, the goal is to lower bound the probability of $A_k\vee E_k$. If $E_{k-1}$ happens, then $A_{k}\vee E_{k}$ happens for sure. Now assume $E_{k-1}$ does not happen and $A_{k-1}$ happens. We go through the iteration $k$. By \Cref{lem: fidelity high dimension} and (a), \begin{equation*}
        \tr(\braket{0^k|C_k\rho C_k^\dagger|0^k})\ge F(C_k\rho C_k^\dagger, \ketbra{0^k}\otimes \sigma_k^*)=F(\rho, \sigma^*)\ge \tau.
    \end{equation*}
    So the state preparation (line \ref{line: high line 4}) succeeds with probability at least $2/3$ according to \Cref{lem: prepare state high dimension}. By \Cref{lem:high_correlation_space}, line \ref{line: high line 5} succeeds with probability at least $2/3$. From now on we suppose the success of both lines, which happens with probability at least $1/3$. 

    $H_k$ is a $\frac{\tau^4}{2^{2t+2}}$-high-correlation basis. If $\dim(\spn(H_k)\cap \Weyl(\sigma_k^*))\ge n-k-t'$, then $E_k$ happens by definition. Otherwise if $\dim(\spn(H_k)\cap \Weyl(\sigma_k^*))< n-k-t'$, by \Cref{lem:weaker_evenly_distribution}, $(\sgn_{k}, Q_k)$ is a low-correlation signed Pauli string that stabilizes $\sigma_k^*$ with probability at least $\tau^4/2^{2t+2}$. In this case, $A_k$ happens. Therefore, $\Pr[A_k\vee E_k|A_{k-1}\vee E_{k-1}]\ge (1/3)\times (\tau^4/2^{2t+2})\ge \tau^4/2^{2t+4}$.

    (c) If $A_{k_{\max}-1}$ happens, by (a), $F(\rho_{k_{\max}}, \sigma^*_{k_{\max}})\ge 1.08^{k_{\max}}\tau>1$, a contradiction.

    (d) By (b), (c),
    \begin{equation*}
        \frac{\tau^4}{2^{2t+4}}\leq \Pr[A_k\vee E_{k}|A_{k-1}\vee E_{k-1}]=\frac{\Pr[A_k\vee E_k, A_{k-1}\vee E_{k-1}]}{\Pr[A_{k-1}\vee E_{k-1}]}\leq \frac{\Pr[A_k\vee E_k]}{\Pr[A_{k-1}\vee E_{k-1}]}.
    \end{equation*}
    So $Pr[A_k\vee E_k]\ge (\tau^4/2^{2t+4})^{k+1}$. In particular, by (c), $A_{k_{\max}-1}$ never happens, so $\Pr[E_{k_{\max}-1}]\ge (\tau^4/2^{2t+4})^{k_{\max}}$.

    (e) Suppose $B_k$ happens. Write $H=H_k$ for simplicity. By definition of $B_k$, $\dim(\spn(H)\cap \Weyl(\sigma^*_k))\ge n-k-t'$, i.e., $\sigma^*_k\in H_{n-k-t}^{n-k-t'}$. By (a), $\sigma^*_k$ is the closest state to $\rho_k$ in $\cS_{n-k}^{n-k-t}$. Since $H_{n-k-t}^{n-k-t'}\subseteq \cS_{n-k}^{n-k-t}$, we have $F(\rho_k, H_{n-k-t}^{n-k-t'})=F(\rho_k, \sigma_k^*)\ge \tau_t$. According to \Cref{lem: high dimension step 2}, with probability at least $2/3$, the output $U_k$ of line \ref{line: high line 6} satisfies $\tr(\braket{0^{n-k-t}|U_k\rho_k U_k^\dagger|0^{n-k-t}})\ge F(\rho_k, H_{n-k-t}^{n-k-t'})-\epsilon_k=F(\rho_k, \sigma^*_k)-\epsilon_k$. Therefore,
    \begin{align*}
        \tr(\braket{0^{n-t}|R_k\rho R_k^\dagger|0^{n-t}})&=\tr(\braket{0^{n-t}|(I^{\otimes k}\otimes U_k)C_k\rho C_{k}^\dagger (I^{\otimes k}\otimes U_k^\dagger)|0^{n-t}})\\
        &=\tr(\braket{0^{n-k-t}|U_k\rho_k U_k^\dagger|0^{n-k-t}})\tr(\braket{0^k|C_k\rho C_k^\dagger|0^k})\\
        &\ge (F(\rho_k, \sigma_k^*)-\epsilon_k)\tr(\braket{0^k|C_k\rho C_k^\dagger|0^k})\\
        &= (F(\rho_k, \sigma_k^*)-\epsilon_k)\frac{F(\rho, \sigma^*)}{F(\rho_k, \sigma_k^*)}\\
        &\ge F(\rho, \sigma^*)-\epsilon.
    \end{align*}
    In the fourth line, we use $F(\rho, \sigma^*)=F(C_k\rho C_k^\dagger, \ketbra{0^k}\otimes \sigma_k^*)=F(\rho_k, \sigma^*_k)\tr(\braket{0^k|C_k\rho C_k^\dagger|0^k})$. In the last line, we use $F(\rho, \sigma^*)/F(\rho_k, \sigma_k^*)\leq 1/1.08^{k}$ from (a) (recall that $B_k$ implies $A_{k-1}$ by definition). 

    We have proved that if $B_k$ happens, $R_k$ is a desired output with probability at least $2/3$. Since $\abs{\mathfrak{P}}\leq k_{\max}+1$, the algorithm returns a desired output with probability at least $2/(3(k_{\max}+1))$.

    (f) This is straightforward from (d), (e).
\end{proof}

\noindent Now we are ready to prove the main theorem.

\begin{proof}[Proof of \Cref{thm:agnostic_learning_high_stabilizer_dimension_states}]
    Replace the $\epsilon$ in \Cref{alg:agnostic_learning_states_high_stabilizer_dimension} by $\epsilon/3$. 
    \Cref{lem: high stabilizer events}(f) establishes that \Cref{alg:agnostic_learning_states_high_stabilizer_dimension} outputs a Clifford gate $C$ such that $\tr(\braket{0^{n-t}|C\rho C^\dagger|0^{n-t}})\ge F(\rho, \cS^{n-t})-\epsilon/3$ with probability at least $p=(\tau/2^t)^{O(\log(2/\tau))}$. We now count the sample complexity $S$ and time complexity $T$.

    At iteration $k$, line \ref{line: high line 5} consumes $m_1=O(\frac{2^{2t}n}{\tau^4})$ copies of $\rho_k$. Line \ref{line: high line 6} consumes $m_2=(1/\epsilon)^{O(t')}=(1/\epsilon)^{O(t+\log(1/\tau))}=(2^t/\tau)^{O(\log(1/\epsilon))}$ copies of $\rho_k$. Line \ref{line: high line 7} consumes $m_3=4$ copies of $\rho_k$. These copies are prepared in Line \ref{line: high line 4}, which consumes $\frac{2}{\tau}(m_1+m_2+m_3+\log(\frac{3}{2}))=n(2^t/\tau)^{O(\log(1/\epsilon))}$ copies of $\rho$. There are $k_{\max}+1=O(\log(2/\tau))$ iterations. The overall sample complexity is $S=n(2^t/\tau)^{O(\log(1/\epsilon))}$.

    At iteration $k$, line \ref{line: high line 4} takes $O(n^2\frac{1}{\tau}(m_1+m_2+\log(\frac{3}{2})))=n^3(2^t/\tau)^{O(\log(1/\epsilon))}$ time, line \ref{line: high line 5} takes $O(\frac{2^{4t}n^2}{\tau^8}(n+\log\frac{2^{2t}}{\tau^4}))=n^3(2^t/\tau)^{O(\log(1/\epsilon))}$ time (\Cref{lem:high_correlation_space}), line \ref{line: high line 6} takes $n^3(1/\epsilon_t)^{O(t')}=n^3(2^t/\tau)^{O(\log(1/\epsilon))}$ time (\Cref{lem: high dimension step 2}), and other lines takes at most $O(n^3)$ time. There are $k_{\max}+1=O(\log(2/\tau))$ iterations. The overall time complexity is $T=n^3(2^t/\tau)^{O(\log(1/\epsilon))}$.

    To sum up, \Cref{alg:agnostic_learning_states_high_stabilizer_dimension} is a algorithm that outputs a Clifford gate $C$ such that $\tr(\braket{0^{n-t}|C\rho C^\dagger|0^{n-t}})\ge F(\rho, \cS^{n-t})-\epsilon/3$ with probability at least $p=(\tau/2^t)^{O(\log(2/\tau))}$ using $S=n(2^t/\tau)^{O(\log(1/\epsilon))}$ copies of $\rho$ and $T=n^3(2^t/\tau)^{O(\log(1/\epsilon))}$ time. By \Cref{lem: reduce to find Clifford}, there exists an algorithm that outputs a state $\sigma=C^\dagger (\ketbra{0^{n-t}}\otimes \sigma_0)\in \cS^{n-t}$ such that $F(\rho, \sigma)\ge F(\rho, \cS^{n-t})-\epsilon$ with probability at least $1-\delta$. The sample complexity is 
    $$O\left(\frac{S}{p}\log\frac{1}{\delta}+\frac{2^{O(t)}}{\epsilon^2}\log\frac{1}{p\delta}+\frac{2^{O(t)}}{\epsilon^2\tau}\log\frac{1}{\delta}\right)=n\left(\frac{2^t}{\tau}\right)^{O(\log(1/\epsilon))}\log\frac{1}{\delta}\,,$$
    and the time complexity is
    \begin{align*}O\left(\frac{T}{p}\log\frac{1}{\delta}+\frac{2^{O(t)}n^2\log(1/\delta)}{\epsilon^2p}\log\frac{1}{p\delta}+\frac{2^{O(t)}n^2}{\epsilon^2\tau}\log\frac{1}{\delta}\right)&=n^3\left(\frac{2^t}{\tau}\right)^{O(\log(1/\epsilon))}\log^2\frac{1}{\delta}\,.\qedhere
    \end{align*}
\end{proof}

\subsection{Proof of Lemma \ref{lem: high dimension step 2}}\label{sec: high dimension step 2}
In this section, we prove \Cref{lem: high dimension step 2}. Let $\sigma^*\in H_{n-t}^{n-t'}$ be the closest state to $\rho$ in $H_{n-t}^{n-t'}$, i.e., $F(\rho, \sigma^*)=F(\rho, H_{n-t}^{n-t'})$. The goal is to find a Clifford unitary $C$ such that $\tr(\braket{0^{n-t}|C\rho C^\dagger|0^{n-t}})\ge F(\rho, \sigma^*)-\epsilon$. By definition, $\dim(\spn(H)\cap \Weyl(\sigma^*))\ge n-t'$. Therefore, measuring $\rho$ in the basis $H$ reveals some information about an $(n-t')$-dimensional subspace of $\Weyl(\sigma^*)$, as shown in the following lemma.
\begin{lemma}\label{lem:find_heavy_subspace_S}
    Fix $t'\ge t\in \mathbb{N}$, $\tau\ge\epsilon>0$. There is an algorithm that, given copies of an $n$-qubit state $\rho$ and a basis $H$ of a stabilizer family such that $F(\rho, H^{n-t'}_{n-t})\ge \tau$, outputs a Clifford circuit $C$ such that $\tr(\braket{0^{n-t'}|C\rho C^\dagger|0^{n-t'}}) F(\rho_{n-t'}^C, \cS^{t'-t}_{t'})\ge F(\rho, H^{n-t'}_{n-t})-\epsilon$ with probability at least $\epsilon^{t'+1}(\tau-\epsilon)$. The algorithm uses $t'+2$ copies and $O(n^3)$ time. Here $\cS_{t'}^{t'-t}$ is the set of $t'$-qubit states with stabilizer dimension at least $t'-t$.
\end{lemma}

\noindent We leave the proof to the end of this section. Let $C_1$ be the output of \Cref{lem:find_heavy_subspace_S}. Suppose we can find $C_2$ be the $t'$-qubit Clifford unitary such that $\tr(\braket{0^{t'-t}|C_2\rho_{n-t'}^{C_1} C_2^\dagger|0^{t'-t}})\ge F(\rho_{n-t'}^{C_1}, \cS^{t'-t}_{t'})-\epsilon$.
Then, letting $C=(I^{\otimes n-t'}\otimes C_2)C_1$, we have
\begin{align}
    \tr(\braket{0^{n-t}|C\rho C^\dagger|0^{n-t}})&=\tr(\braket{0^{n-t}|(I^{\otimes n-t'}\otimes C_2)C_1\rho C_1^\dagger (I^{\otimes n-t'}\otimes C_2^\dagger)|0^{n-t}})\nonumber\\
    &=\tr(\braket{0^{n-t'}|C_1\rho C_1^\dagger|0^{n-t'}})\tr(\braket{0^{t'-t}|C_2\rho_{n-t'}^{C_1} C_2^\dagger|0^{t'-t}})\nonumber\\
    &\ge \tr(\braket{0^{n-t'}|C_1\rho C_1^\dagger|0^{n-t'}})(F(\rho_{n-t'}^{C_1}, \cS^{t'-t}_{t'})-\epsilon)\nonumber\\
    &\ge \tr(\braket{0^{n-t'}|C_1\rho C_1^\dagger|0^{n-t'}})F(\rho_{n-t'}^{C_1}, \cS^{t'-t}_{t'})-\epsilon\nonumber\\
    &\ge F(\rho, H^{n-t'}_{n-t})-2\epsilon\,, \label{eq: high dimension step 2 equation 1}
\end{align}
and thus $C$ is a desired output of \Cref{lem: high dimension step 2} (with $\epsilon$ rescaled to $\epsilon/2$). 
So the problem reduces to finding $C_2$, which is equivalent to finding a $t'$-qubit state in $\cS_{t'}^{t'-t}$ that is approximately closest to $\rho_{n-t'}^C$. This new problem has the same form as the original agnostic tomography question (\Cref{thm:agnostic_learning_high_stabilizer_dimension_states}), except the number of qubits is reduced from $n$ to $t'$. In other words, if we can solve agnostic tomography of states with high stabilizer dimension with exponential sample and time complexity (with respect to the system size), we can find $C_2$ with $2^{O(t')}=\poly(n, 1/\tau)$ sample and time complexity. We give this weaker form of \Cref{thm:agnostic_learning_high_stabilizer_dimension_states} in the following lemma:

\begin{lemma}[Agnostic tomography of states with high stabilizer dimension in exponential time]\label{lem: high stabilizer exponential time} 
    Fix $t\in \mathbb{N}$, $\tau\ge\epsilon>0, \delta>0$. There is an algorithm that, given copies of an $n$-qubit state $\rho$ with $F(\rho, \cS^{n-t})\ge \tau$, returns a Clifford gate $C$ such that $\tr(\braket{0^{n-t}|C\rho C^\dagger|0^{n-t}})\ge F(\rho, \cS^{n-t})-\epsilon$ with probability at least $1-\delta$. The algorithm uses $(2/\tau)^{O(n)}(1/\epsilon)^{O(t)}\log(1/\delta)$ and $(2/\tau)^{O(n)}(1/\epsilon)^{O(t)}\log^2(1/\delta)$ time.
\end{lemma}

\noindent We remark that the lemma has the same form as \Cref{thm:agnostic_learning_high_stabilizer_dimension_states} except that the complexity is worse (exponential in $n$). However, this is already non-trivial, since enumerating all $(n-t)$-dimensional stabilizer groups takes $2^{\Omega(n(n-t))}$ time. The proof again relies on stabilizer bootstrapping, with the key difference that in place of Bell difference sampling, we simply sample Pauli strings uniformly at random. This ensures that the sample is evenly distributed in $\Weyl(\sigma^*)$. The catch with uniform sampling is that the sampled Pauli string lies in $\Weyl(\sigma^*)$ with only an exponentially small probability, but because the system size in question is small, this is something we can now afford. The analysis is almost a tautology given the steps in~\Cref{sec: high dimension analysis} so we defer its proof to \Cref{sec: few qubits case high dimension}.

Equipped with \Cref{lem:find_heavy_subspace_S} and \Cref{lem: high stabilizer exponential time}, we are able to prove \Cref{lem: high dimension step 2}.

\begin{proof}[Proof of \Cref{lem: high dimension step 2}]
    Running the algorithm in \Cref{lem:find_heavy_subspace_S} (with $\epsilon$ set to $\epsilon/4$), with probability at least $(\epsilon/4)^{t'+1}(\tau-\epsilon/4)\ge (\epsilon/4)^{t'+1}\tau/2$, we obtain a Clifford gate $C_1$ such that
    $$\tr(\braket{0^{n-t'}|C_1\rho C_1^\dagger|0^{n-t'}})\ge \tr(\braket{0^{n-t'}|C_1\rho C_1^\dagger|0^{n-t'}})F(\rho_{n-t'}^{C_1}, \cS^{t'-t}_{t'})\ge F(\rho, H_{n-t}^{n-t'})-\frac{\epsilon}{4}\ge \frac{\tau}{2},$$ 
    According to \Cref{lem: prepare state high dimension}, we can prepare $N$ copies of $\rho_{n-t'}^{C_1}$ using $\frac{4}{\tau}(N+\log(2))$ copies of $\rho$ with probability at least $1/2$. 
    Given $N=(2/\tau)^{O(t')}(1/\epsilon)^{O(t)}$ copies of $\rho_{n-t'}^C$, \Cref{lem: high stabilizer exponential time} (with $(\tau, \epsilon, \delta)$ set to $(\tau/2, \epsilon/4, 1/2)$) returns a $t'$-qubit Clifford gate $C_2$ such that $\tr(\braket{0^{t'-t}|C_2\rho_{n-t'}^{C_1} C_2^\dagger|0^{t'-t}})\ge F(\rho_{n-t}^{C_1}, \cS_{t'}^{t'-t})-\epsilon/4$ with probability at least $1/2$.
    Let $C=(I^{\otimes (n-t')}\otimes C_2)C_1$. With the same calculation as \eqref{eq: high dimension step 2 equation 1}, we have $\tr(\braket{0^{n-t}|C\rho C^\dagger|0^{n-t}})\ge F(\rho, H^{n-t'}_{n-t})-\epsilon/2$.

    In summary, with probability at least $p=(\epsilon/4)^{t'+1}\tau/8$, the procedure above outputs a Clifford gate $C$ such that $\tr(\braket{0^{n-t}|C\rho C^\dagger|0^{n-t}})\ge F(\rho, H^{n-t'}_{n-t})-\epsilon/2$. It costs $S=(2/\tau)^{O(t')}(1/\epsilon)^{O(t)}$ samples and $T=O(n^3)+(2/\tau)^{O(t')}(1/\epsilon)^{O(t)}$ time. With the same repeating argument as \Cref{lem: reduce to find Clifford}(a), we can amplify the success probability to $1-\delta$. This argument introduces another error of $\epsilon/2$, so the overall error is $\epsilon$. The sample complexity is $$O\left(\frac{S}{p}\log\frac{1}{\delta}+\frac{2^{O(t)}}{\epsilon^2}\log\frac{1}{p\delta}\right)=\left(\frac1\epsilon\right)^{O(t')}\log\frac{1}{\delta}$$ and the time complexity is \begin{align*}
        O\left(\frac{T}{p}\log\frac{1}{\delta}+\frac{2^{O(t)}n^2\log(1/\delta)}{\epsilon^2p}\log\frac{1}{p\delta}\right)&=n^3\left(\frac1\epsilon\right)^{O(t')}\log^2\frac{1}{\delta}\,.\qedhere
    \end{align*}
\end{proof}

To close the section, we prove \Cref{lem:find_heavy_subspace_S}. Basically the goal is to find $\spn(H)\cap \Weyl(\sigma^*)$ given $\dim(\spn(H)\cap \Weyl(\sigma^*))\ge n-t'$. For simplicity, we first assume $\spn(H)=\{I, Z\}^{\otimes n}$ and work on the computational basis. We focus on the field $\mathbb{F}_2^n$ (instead of $\mathbb{F}_2^{2n}$). For $y\in\mathbb{F}_2^n$, define $Z^y\triangleq \otimes_{i=1}^n Z_i^{y_i}$. For any set $A\subseteq \mathbb{F}_2^n$, define $A_\mathcal{Z}=\{Z^y:y\in A\}$. We also define the \emph{affine span} of a subset $S\subseteq \mathbb{F}_2^n$ as $\spn_{\text{aff}}\triangleq \spn(S-S)$.

We briefly overview the algorithm in \Cref{lem:find_heavy_subspace_S}. Recall from the discussion in Section~\ref{sec:dim_overview} that the idea is to compute the affine span of sufficiently many samples from measuring $\rho$ in the joint eigenbasis of $H$, and identifying $H\cap \Weyl(\sigma^*)$ as the orthogonal complement of this affine span. In \Cref{lem: find the intersection of subspaces}, we make this idea rigorous. A crucial tool is the principle of inclusion-exclusion (\Cref{lem: principle of inclusion-exclusion}), whose proof is deferred to \Cref{sec: proof of inclusion-exclusion}.

\begin{lemma}[Principle of inclusion-exclusion]\label{lem: principle of inclusion-exclusion}
    Let $n\ge t\in \mathbb{N}$, $\rho$ be a $n$-qubit state, and $P_1$ be a stabilizer group. For any $\sigma_1\in \Stab(P_1)$ and $\sigma_2\in\cS^{n-t}$, there exists a state $\sigma\in \Stab(P_1)\cap \cS^{n-t}$ such that $F(\rho, \sigma)\ge F(\rho, \sigma_1)+F(\rho, \sigma_2)-1$.
\end{lemma}

\begin{lemma}\label{lem: find the intersection of subspaces}
    Fix $\tau\ge \epsilon>0$ and $t'\ge t\in \mathbb{N}$. Let $\rho$ be an $n$-qubit quantum state. Suppose there exists a $\sigma\in \cS^{n-t}$ such that $\dim(\Weyl(\sigma)\cap \{I, Z\}^{\otimes n})\ge n-t'$ and $F(\rho, \sigma)\ge \tau$. If we measure $\rho$ in the computational basis for $t'+1$ times and obtain $z_0, z_1,\cdots, z_{t'}$, with probability at least $\epsilon^{t'+1}$, $F(\rho, \Stab(\spn_{\text{aff}}(\{z_0, z_1,\cdots, z_{t'}\})^\perp_\mathcal{Z})\cap \cS^{n-t}) \ge F(\rho, \sigma)-\epsilon$.
\end{lemma}

\begin{proof}
    For simplicity, assume $\Weyl(\sigma)\cap \{I, Z\}^{\otimes n}\supseteq\{I, Z\}^{\otimes n-t'}\otimes I^{\otimes t'}$.
    This is without loss of generality because there exists a Clifford gate (indeed a CNOT circuit) that maps a $(n-t')$-dimensional subspace of $\Weyl(\sigma)\cap \{I, Z\}^{\otimes n}$ to $\{I, Z\}^{\otimes n-t'}$ without changing $\{I, Z\}^{\otimes n}$. Note that the algorithm will not rely on this Clifford gate (as we do not know $\Weyl(\sigma)$ ahead).

    Since $\sigma$ is stabilized by $\{I, Z\}^{\otimes n-t'}\otimes I^{\otimes t'}$, it has the form $\ketbra{s}\otimes \sigma_0$ for some $s\in \{0, 1\}^{n-t'}$ and $\dim(\Weyl(\sigma_0))\ge t'-t$. Define $p=\tr(\braket{s|\rho|s})$ and $\rho_0=\braket{s|\rho|s}/p$. By \Cref{lem: fidelity high dimension}, $\epsilon\leq \tau\leq F(\rho, \sigma) = p F(\rho_0, \sigma_0)\leq p$. 

    The probability that the outcome of a $Z$-basis measurement starts with $s$ is $\sum_{y\in \{0, 1\}^w}\braket{sy|\rho|sy} = p$. Therefore, with probability $p^{t'+1}$, all $z_0, \cdots, z_{t'}$ start with $s$. From now on we condition on this event. Write $z_i=sy_i$ for some $y_i\in \{0, 1\}^{t'}$. The conditional probability of $y$ is $\cD_0(y)\triangleq \braket{sy|\rho|sy}/p=\braket{y|\rho_0|y}$, exactly the probability of $y$ when we measure $\rho_0$ in the computational basis. Fixing $y_0$, $y_i-y_0$ ($1\leq i\leq t'$) are i.i.d. samples from $\cD_0(y-y_0)$. By \Cref{lem: sample heavy-weight subspace}(a), with conditional probability at least $(\epsilon/p)^{t'}$, 
    \begin{equation}\label{eq: span_aff high weight}
        \Pr_{y\sim\cD_0}[y-y_0\in \spn_{\text{aff}}(\{y_0, \cdots, y_{t'}\})]=\Pr_{y\sim \cD_0}[y-y_0\in \spn(\{y_1-y_0,\cdots, y_{t'}-y_0\})]\ge 1-\frac{\epsilon}{p}.
    \end{equation}
    Therefore, with probability at least $p^{t'+1}(\epsilon/p)^{t'}\ge \epsilon^{t'+1}$, all $z_0, \cdots, z_{t'}$ start with $s$ and \eqref{eq: span_aff high weight} holds. We now prove that under these two events, $F(\rho, \Stab(\spn_{\text{aff}}(\{z_0, z_1,\cdots, z_m\})^\perp)\cap \cS^{n-t}) \ge F(\rho, \sigma)-\epsilon$.

    Again, without loss of generality, assume $\spn_{\text{aff}}(\{y_0, \cdots, y_{t'}\})=0^{t'-r}\times \{0, 1\}^r$ for some $0\leq r\leq t'$. Then 
    \begin{equation*}
        \spn_{\text{aff}}(\{z_0, \cdots, z_m\})^\perp_\mathcal{Z}=(0^{n-r}\times \{0, 1\}^r)^\perp_\mathcal{Z}=(\{0, 1\}^{n-r}\times 0^{n-r})_\mathcal{Z}=\{I, Z\}^{\otimes n-r}\otimes I^{\otimes r}.
    \end{equation*}
    Write $y_0=uv$ for $u\in \{0, 1\}^{t'-r}, v\in \{0, 1\}^r$. \eqref{eq: span_aff high weight} implies that with probability at least $1-\epsilon/p$, a sample $y$ from $\cD_0$ starts with $u$. In other words, 
    \begin{equation}\label{eq: first lemma high dimension 1}
        \tr(\braket{u|\rho_0|u})\ge 1-\frac{\epsilon}{p}.
    \end{equation}
    By \Cref{lem: fidelity high dimension}, \eqref{eq: first lemma high dimension 1} implies there exists an $t'$-qubit state $\phi_0\in \Stab(\{I, Z\}^{\otimes t'-r}\otimes I^{\otimes r})$ such that $F(\rho_0, \phi_0)\ge 1-\epsilon/p$. Meanwhile, the stabilizer dimension of $\sigma_0$ is at least $t'-t$ and $F(\rho_0, \sigma_0)=F(\rho, \sigma)/p$. Therefore, according to \Cref{lem: principle of inclusion-exclusion}, there exists a $t'$-qubit state $\psi_0\in\Stab(\{I, Z\}^{\otimes t'-r}\otimes I^{\otimes r})\cap \cS^{t'-t}_{t'}$ such that $F(\rho_0, \psi_0)\ge (F(\rho, \sigma)-\epsilon)/p$. Define $\psi=\ketbra{s}\otimes \psi_0$. Then $\psi$ is stabilized by $\{I, Z\}^{\otimes n-r}\otimes I^{\otimes r}$ and has stabilizer dimension at least $n-t$. Hence, $\psi\in \Stab(\{I, Z\}^{\otimes n-r}\otimes I^{\otimes r})\cap \cS^{n-t}$ and
    \begin{align*}
        F(\rho, \Stab(\spn_{\text{aff}}(\{z_0, z_1,\cdots, z_m\})^\perp_\mathcal{Z})\cap \cS^{n-t}) &\ge F(\rho, \psi)=pF(\rho_0, \psi_0) \ge F(\rho, \sigma)-\epsilon\,.\qedhere
    \end{align*}
\end{proof}

\noindent With \Cref{lem: find the intersection of subspaces}, our algorithm is clear: Measure $\rho$ in the computational basis for $t'+1$ times and calculate the orthogonal space of the affine span. We formally specify this in \Cref{alg: find heavy subspace} and analyze the algorithm in the proof of \Cref{lem:find_heavy_subspace_S}. 

\begin{algorithm}[htbp]
    \DontPrintSemicolon
    \caption{Algorithm of \Cref{lem:find_heavy_subspace_S}}\label{alg: find heavy subspace}
    \KwInput{$t'\ge t\in \mathbb{N}, \tau\ge \epsilon>0$, copies of an $n$-qubit state $\rho$, a basis $H$ of a stabilizer family}
    \Promise{$F(\rho, H^{n-t'}_{n-t})\ge \tau$}
    \KwOutput{A Clifford circuit $C$}
    \Goal{With probability at least $\epsilon^{t'+1}(\tau-\epsilon)$, $\tr(\braket{0^{n-t'}|C\rho C^\dagger|0^{n-t'}})F(\rho_{n-t'}^C, \cS^{t'-t}_{t'})\ge F(\rho, H^{n-t'}_{n-t})-\epsilon$.}
    Find a Clifford gate $C_1$ such that $C_1\spn(H)C_1^\dagger=\{I, Z\}^{\otimes n}$ by \Cref{lem:Clifford_Synthesis}.\label{line: alg4line1}\\
    Measure $C_1\rho C_1^\dagger$ on computational basis for $t'+1$ times, obtaining $z_0,\cdots, z_{t'}$.\label{line: alg4line2}\\
    Calculate a basis $T$ of a $(n-t')$-dimensional subspace of $\spn_\text{aff}(z_0, \cdots, z_{t'})^\perp_\mathcal{Z}$.\label{line: alg4line3}\\
    Find a Clifford gate $C_2$ that maps $\spn(T)$ to $\{I, Z\}^{\otimes n-t'}\otimes I^{\otimes t'}$ by \Cref{lem:Clifford_Synthesis}.\label{line: alg4line4}\\
    Measure the first $n-t'$ qubits of $C_2C_1\rho C_1^\dagger C_2^\dagger$ on the computational basis. Denote the outcome by $s\in \{0, 1\}^{n-t'}$.\label{line: alg4line5}\\ 
    \Return $X^s C_2C_1$, where $X^s=\otimes_{i=1}^{n-t'}X_i^{s_i}$. 
\end{algorithm}

\begin{proof}[Proof of \Cref{lem:find_heavy_subspace_S}]
    We now prove that \Cref{alg: find heavy subspace} achieves the stated goal, thus proving \Cref{lem:find_heavy_subspace_S}. 
    Let $\sigma^*\in H_{n-t}^{n-t'}$ be the state such that $F(\rho, \sigma^*)=F(\rho, H_{n-t}^{n-t'})\ge \tau\ge \epsilon$. 
    By definition of $H_{n-t}^{n-t'}$, $\sigma^*\in \cS^{n-t}$ and $\dim(\Weyl(\sigma^*)\cap \spn(H))\ge n-t'$, so $C_1\sigma^*C_1^\dagger \in \cS^{n-t}$ and $\dim(\Weyl(C_1\sigma^*C_1^\dagger)\cap \{I, Z\}^{\otimes n})\ge n-t'$. By \Cref{lem: find the intersection of subspaces}, with probability at least $\epsilon^{t'+1}$, $F(C_1\rho C_1^\dagger, \Stab(\spn_{\text{aff}}(z_0, \cdots, z_{t'})^\perp_\mathcal{Z})\cap \cS^{n-t})\ge F(C_1\rho C_1^\dagger, C_1\sigma^*C_1^\dagger)-\epsilon=F(\rho, \sigma^*)-\epsilon$. Since $\spn(T)\subseteq \spn_{\text{aff}}(z_0, \cdots, z_{t'})^\perp_\mathcal{Z}$, we have $$F(C_1\rho C_1^\dagger, \Stab(\spn(T))\cap \cS^{n-t})\ge F(\rho, \sigma^*)-\epsilon.$$
    In other words, there exists a state $\sigma\in \Stab(\spn(T))\cap \cS^{n-t}$ such that $F(C_1\rho C_1^\dagger, \sigma)\ge F(\rho, \sigma^*)-\epsilon$.
    $C_2\sigma C_2^\dagger$ is a state in $\cS^{n-t}$ and stabilized by $\{I, Z\}^{\otimes n-t'}\otimes I^{\otimes t'}$. Therefore, it has the from $\ketbra{r}\otimes \sigma_0$ for some $r\in \{0, 1\}^{n-t'}$ and $\sigma_0\in \cS^{t'-t}_{t'}$. By \Cref{lem: fidelity high dimension}, $$\tr(\braket{r|C_2C_1\rho C_1^\dagger C_2^\dagger|r})\ge F(C_2C_1\rho C_1^\dagger C_2^\dagger, C_2\sigma C_2^\dagger)=F(C_1\rho C_1^\dagger, \sigma)\ge \tau-\epsilon.$$
    Therefore, the outcome $s$ of the measurement in Line \ref{line: alg4line5} is $r$ with probability at least $\tau-\epsilon$. If this happens, the output of the algorithm is $X^r C_2C_1$. By \Cref{lem: fidelity high dimension},
    \begin{align*}
        &\tr(\braket{0^{n-t'}|X^rC_2C_1\rho C_1^\dagger C_2^\dagger X^r|0^{n-t'}})F(\rho_{n-t'}^{X^rC_2C_1}, \cS^{t'-t}_{t'})\\
        \ge&\tr(\braket{r|C_2C_1\rho C_1^\dagger C_2^\dagger|r})F(\rho_{n-t'}^{X^rC_2C_1}, \sigma_0)\\
        =&\tr(\braket{r|C_2C_1\rho C_1^\dagger C_2^\dagger|r})F\left(\frac{\braket{r|C_2C_1\rho C_1^\dagger C_2^\dagger|r}}{\tr(\braket{r|C_2C_1\rho C_1^\dagger C_2^\dagger|r})}, \sigma_0\right)\\
        =&\tr(C_2C_1\rho C_1^\dagger C_2^\dagger, \ketbra{r}\otimes \sigma_0)=\tr(C_2C_1\rho C_1^\dagger C_2^\dagger, C_2\sigma C_2^\dagger)\\
        =&\tr(C_1\rho C_1^\dagger, \sigma)\ge F(\rho, \sigma^*)-\epsilon=F(\rho, H_{n-t}^{n-t'})-\epsilon.
    \end{align*}
    So \Cref{alg: find heavy subspace} succeeds with probability at least $\epsilon^{t'+1}(\tau-\epsilon)$. The sample complexity is $t'+2$. Line \ref{line: alg4line1}, Line \ref{line: alg4line3} and Line \ref{line: alg4line4} use $O(n^3)$ time. Line \ref{line: alg4line2} uses $O(n^2t)=O(n^3)$ time ($n^2$ is the cost of applying $C_1$ to $\rho$). Other lines use less time. So the overall time complexity is $O(n^3)$. 
\end{proof}

\begin{remark}\label{remark: high dim 1}
    There is a slight modification of \Cref{lem: find the intersection of subspaces}: Measure $\rho$ in the computational basis $(2\log(2)+2t')/\epsilon + 1$ times instead of $t'+1$ times. The same analysis (with \Cref{lem: sample heavy-weight subspace}(a) replaced by \Cref{lem: sample heavy-weight subspace}(b)) shows that the probability of success is at least $\tau^{(2\log(2)+2t')/\epsilon+1}/2$. Denote $\lambda=\tau^{1/\epsilon}$. The success probability is at least $\lambda^{O(t')}/2$. 

    Based on this modification, \Cref{lem:find_heavy_subspace_S} succeeds with probability at least $\lambda^{O(t')}(\tau-\epsilon)/2$ using $O(t'/\epsilon)$ samples and $O(n^3)$ time. The sample/time complexity of \Cref{lem: high stabilizer exponential time} becomes $\poly(n, \frac{1}{\epsilon},\log\frac{1}{\delta})(\frac{2}{\tau})^{O(n)}(\frac{1}{\lambda})^{O(t)}$.
    The sample/time complexity of \Cref{lem: high dimension step 2} becomes $\poly(n, \frac{1}{\epsilon},\log\frac{1}{\delta})(\frac{2}{\tau\lambda})^{O(t')}$. The sample/time complexity of \Cref{thm:agnostic_learning_high_stabilizer_dimension_states} becomes $\poly(n, \frac{1}{\epsilon},\log\frac{1}{\delta})(\frac{2}{\tau\lambda})^{O(t')}=\poly(n, \frac{1}{\epsilon},\log\frac{1}{\delta})(2^t/\tau)^{O(\log\frac{2}{\tau\lambda})}$. As a comparison, the current version of \Cref{thm:agnostic_learning_high_stabilizer_dimension_states} has sample/time complexity $\poly(n, \frac{1}{\epsilon},\log\frac{1}{\delta})(2^t/\tau)^{O(\log\frac{1}{\epsilon})}$.

    When $\tau$ is small, say, $\tau\leq 1/2$, the modified version is much worse because $1/\lambda = (1/\tau)^{1/\epsilon}\gg 1/\epsilon$. However, when $\tau$ is close to 1, the modified version could be better. Specifically, the current version runs in time super-polynomial in $n$ when $\tau=1$ and $\epsilon=1/\poly(n)$. On the other hand, if $1/\tau\leq 1+c\epsilon$ for some constant $c$, $1/\lambda\leq (1+c\epsilon)^{1/\epsilon}\leq e^c$ is a constant. Thus, $(2^t/\tau)^{O(\log\frac{2}{\tau\lambda})}=\poly(2^t)$ and the modified version recovers the $\poly(n, 1/\epsilon, 2^t)$ time complexity achieved in the realizable setting by~\cite{grewal2023efficient}.
\end{remark}

\section{Agnostic tomography of discrete product states}\label{sec:product}

In this section, we give algorithms for agnostic tomography of states from $\mathcal{K}^{\otimes n}$, where $\mathcal{K}$ is any $\mu$-packing set. As in our preceding results, we give a more general guarantee, namely an algorithm for which every element of $\mathcal{K}^{\otimes n}$ with fidelity at least $\tau$ with $\rho$ has a non-negligible chance of being the final output:

\begin{theorem}\label{thm:product_base}
Fix $\tau>0$ and a $\mu$-packing set $\mathcal{K}$, and let $\rho$ be an unknown $n$-qubit state.
There is an algorithm with the following guarantee.

Let $\ket{\phi}\in\mathcal{K}^{\otimes n}$, and suppose its fidelity with $\rho$ is at least $\tau$. Given copies of $\rho$, the algorithm outputs $\ket{\phi}$ with probability at least $(n|\mathcal{K}|)^{-O(\log(1/\tau)/\mu)}$.

The algorithm only performs single-copy measurements on $\rho$.
The sample complexity is $O(\frac{1}{\mu^3\tau}\log(1/\tau)\log n)$ and the time complexity is $O(\frac{1}{\mu^4\tau}\log^2(1/\tau)\log n+\frac{n}{\mu^3}\log(1/\tau)\log n)$.
\end{theorem}

\noindent We give a proof of Theorem~\ref{thm:product_base} in Sections~\ref{sec:product_alg} and~\ref{sec:product_analysis}. In the rest of this subsection, we record some consequences of this general result.

Firstly, as in Section~\ref{sec:stabilizer}, by repeating the algorithm in Theorem~\ref{thm:product_base} sufficiently many times, we obtain a list-decoding guarantee:

\begin{corollary}\label{cor:prod_list}
Fix $\tau, \delta > 0$ and a $\mu$-packing set $\mathcal{K}$, and let $\rho$ be an unknown $n$-qubit state.

There is an algorithm that, given copies of $\rho$, returns a list of states in $\mathcal{K}^{\otimes n}$ of length $\log(1/\delta)\cdot (n|\mathcal{K}|)^{O(\log(1/\tau)/\mu)}$ so that with probability at least $1-\delta$, all states in $\mathcal{K}^{\otimes n}$ with fidelity at least $\tau$ with $\rho$ appear in the list.

The algorithm only performs single-copy measurements on $\rho$.
Both the sample complexity and the time complexity are $\log(1/\delta)\cdot (n|\mathcal{K}|)^{O(\log(1/\tau)/\mu)}$.
\end{corollary}

\noindent Similarly, this also readily implies an algorithm for proper agnostic tomography of discrete product states. The proof details are straightforward and deferred to Appendix~\ref{sec:defer_product_cor}.

\begin{corollary}
\label{cor:product_base}
    Fix $\tau\ge \epsilon>0$, $\mu>0$, and $\delta>0$. There is an algorithm that, given copies of an $n$-qubit state $\rho$ and a $\mu$-packing set $\mathcal{K}$ such that $F_{\mathcal{K}^{\otimes n}}(\rho)\ge \tau$, returns a state $\ket{\phi}\in \mathcal{K}^{\otimes n}$ that satisfies $F(\rho, \ket{\phi})\ge F_{\mathcal{K}^{\otimes n}}(\rho)-\epsilon$ with probability at least $1-\delta$. The algorithm only performs single-copy measurements on $\rho$. Both the sample complexity and the time complexity are $\log^2(1/\delta)\cdot (n|\mathcal{K}|)^{O(\log(1/\tau)/\mu)} /\epsilon^2$.
\end{corollary}

\noindent Note that if we specialize $\mathcal{K}$ to the set of single-qubit stabilizer states (a $\mu=1/2$-packing set), this theorem implies a $n^{O(\log 1/\tau)}/\epsilon^2$-time algorithm for agnostic tomography of stabilizer product states, recovering the result of Ref.~\cite{grewal2024agnostic}.
This is rather surprising because the algorithm does not utilize the special structure of stabilizer states, as opposed to Ref.~\cite{grewal2024agnostic}, illustrating the utility of our stabilizer bootstrapping technique.
In Section~\ref{sec:stabilizer_product}, we will show that with the help of Bell difference sampling, we can improve the complexity even further in this special case.

\subsection{Construction of the algorithm}
\label{sec:product_alg}

Let $|\phi\rangle = \bigotimes_{j=1}^n|\phi^j\rangle \in\mathcal{K}^{\otimes n}$ satisfy $F(|\phi\rangle,\rho)\geq\tau$. The notion of a complete family of projectors in this setting is a set of projectors $\{\Pi^i\}$ where $\Pi^i$ projects the $i$-th qubit in the direction of $\ket{\phi^i}$ and maps the other qubits via the identity. If one could find these, one could measure in the basis which is the tensor product of the bases $\{\ket{\phi^i}, \ket{\psi^i}\}$, where $\psi^i$ is a direction orthogonal to $\ket{\phi^i}$, and obtain $\ket{\phi}$ with probability at least $\tau$. But of course, in this simple setting, the measurement is not even necessary as we can read off $\ket{\phi}$ from the projectors themselves.

\vspace{0.5em}\noindent \textbf{Step 1: Find a high-correlation family.}\vspace{0.5em}

\noindent Since now the problem contains a product structure instead of a stabilizer group structure, the notion of a high-correlation family needs to be modified. First, as noted above, the projectors we work with here will not be given by signed Pauli strings, but instead by projectors of the form $\Pi^i_{\ket{\psi}}$ which are given by projection to some direction $\ket{\psi}\in \mathcal{K}$ in one of the qubits, and by identity in the other qubits. Instead of finding projectors whose estimated correlation with $\rho$ is above some threshold, for each qubit $i\in[n]$ we instead keep the projector among $\{\Pi^i_{\ket{\psi}}\}_{\ket{\psi}\in\mathcal{K}}$ which has the \emph{highest} correlation with $\rho$.

The following result, which can be thought of as the analogue of Lemma~\ref{lem:high-corr-commute}, ensures that for the high-correlation family obtained in this way, any projector $\Pi^i_{\ket{\psi}}$ outside the family has low correlation with $\rho$:

\begin{lemma}\label{lem:one-in-packing}
Fix an $n$-qubit state $\rho$ and a $\mu$-packing set $\mathcal{K}$.
Given a qubit index $i\in[n]$, at most one element $\ket{\psi}\in\mathcal{K}$ satisfies
\[\tr(\tr_{-i}(\rho)|\psi\rangle\!\langle\psi|)>\frac{1+\sqrt{1-\mu}}{2}.\]
\end{lemma}

\begin{proof}
Suppose for the sake of contradiction that there are distinct $|\psi\rangle, \ket{\varphi}\in\mathcal{K}$ satisfying the inequality.
Then we have 
\[\arccos\sqrt{\langle\psi|\tr_{-i}(\rho_i)|\psi\rangle}<\arccos\sqrt{\frac{1+\sqrt{1-\mu}}{2}}=\frac{1}{2}\arccos\sqrt{1-\mu},\]
and similar for $|\varphi\rangle$.
But by the triangle inequality of the Bures metric, we have
\[\arccos\sqrt{\langle\psi|\tr_{-i}(\rho)|\psi\rangle}+\arccos\sqrt{\langle\varphi|\tr_{-i}(\rho)|\varphi\rangle}\geq\arccos\sqrt{F(|\psi\rangle,|\varphi\rangle)}\geq\arccos\sqrt{1-\mu},\]
a contradiction.
\end{proof}

\begin{lemma}\label{lem:unique-high-corr}
Fix an $n$-qubit state $\rho$ and a $\mu$-packing set $\mathcal{K}$. Define
\begin{equation*}
    \theta(\mu) \triangleq \frac{1+\sqrt{1-\mu}}{2}+\frac{\mu}{8}\,.
\end{equation*}
Run the algorithm in~\Cref{lem:tomography_fidelity} to estimate $\langle\psi|\tr_{-j}(\rho)|\psi\rangle$ for all $|\psi\rangle\in\mathcal{K}$ and $j\in[n]$ so that with probability at least $1-\delta$, all fidelities are estimated to error within $\frac{\mu}{16}$.
Let $|\psi^j\rangle$ be the state in $\mathcal{K}$ that has the highest estimated fidelity with $\tr_{-j}(\rho)$.
Then with probability at least $1-\delta$, if $|\psi\rangle\in\mathcal{K}$ and $|\psi\rangle\neq|\psi^j\rangle$, then $\langle\psi|\tr_{-j}(\rho)|\psi\rangle\leq\theta(\mu)$
The sample complexity of this procedure is $O(\frac{1}{\mu^2}\log\frac{n}{\delta})$ and the time complexity is $O(\frac{n}{\mu^2}\log\frac{n}{\delta})$.
\end{lemma}

\begin{proof}
Suppose all the fidelities are estimated to error within $\frac{\mu}{16}$.
Consider a specific qubit $j$.
Let $|\psi\rangle=\argmax_{|\varphi\rangle\in\mathcal{K}\backslash\{|\psi^j\rangle\}}\langle\varphi|\tr_{-j}(\rho)|\varphi\rangle$.
If $\langle\psi|\tr_{-j}(\rho)|\psi\rangle>\frac{1+\sqrt{1-\mu}}{2}$, then by~\Cref{lem:one-in-packing}, $\langle\psi^j|\tr_{-j}(\rho)|\psi^j\rangle\leq\frac{1+\sqrt{1-\mu}}{2}$.
Thus the estimated fidelity of $|\psi^j\rangle$ with $\tr_{-j}(\rho)$ is at most $\frac{1+\sqrt{1-\mu}}{2}+\frac{\mu}{16}$, which must be bigger than or equal to the estimated fidelity of $|\psi\rangle$ with $\tr_{-j}(\rho)$.
Thus $\langle\psi|\tr_{-j}(\rho)|\psi\rangle\leq\frac{1+\sqrt{1-\mu}}{2}+\frac{\mu}{8}$.
Thus we always have $\langle\psi|\tr_{-j}(\rho)|\psi\rangle\leq\theta(\mu)$.
The complexities follow directly from~\Cref{lem:tomography_fidelity}.
\end{proof}

\vspace{0.5em}\noindent \textbf{Step 2: If the family is complete, i.e. if $\bigotimes_{j=1}^n|\psi^j\rangle=\bigotimes_{j=1}^n|\phi^j\rangle$, then directly obtain the answer.}\vspace{0.5em}

\noindent This step is immediate: in this case, we can simply read off $\ket{\phi}$ from the projectors in the family.

\vspace{0.5em}\noindent \textbf{Step 3: If the family is incomplete, sample a low-correlation projector.}\vspace{0.5em}

\noindent If on the other hand $\bigotimes_{j=1}^n|\psi^j\rangle\neq\bigotimes_{j=1}^n|\phi^j\rangle$, then we know that at least one of $|\psi^j\rangle$ is wrong.
We just randomly pick a $j\in[n]$ and randomly pick a $|\psi\rangle\in\mathcal{K}$ so that $|\psi\rangle\neq|\psi^j\rangle$ and assume it is the correct $|\phi^j\rangle$.
Since $\bigotimes_{j=1}^n|\psi^j\rangle$ is a high-correlation family, by~\Cref{lem:unique-high-corr} we know that $\langle\psi|\tr_{-j}(\rho)|\psi\rangle\leq\theta(\mu)$.
Moreover, our guess is correct with probability at least $\frac{1}{n(|\mathcal{K}|-1)}$.

\vspace{0.5em}\noindent \textbf{Step 4: Bootstrap by measuring.}\vspace{0.5em}

\noindent Suppose $|\psi\rangle$ chosen at Step $3$ is correct, then by measuring we can amplify the fidelity.

\begin{lemma}\label{lem:fidelity_amplify_product}
If $\tr(\Pi_{|\psi\rangle}^j\rho)\leq\theta(\mu)$ and $\Pi_{|\psi\rangle}^j=|\psi\rangle\!\langle\psi|_j$ stabilizes $|\phi\rangle$, then $F(|\phi\rangle,\rho')\geq F(|\phi\rangle,\rho)/\theta(\mu)$.
\end{lemma}

\begin{proof}
We have the same calculation as in previous sections: 
\begin{align*}
    \langle\phi|\rho'|\phi\rangle &= \frac{\langle\phi|\rho|\phi\rangle}{\tr(\Pi_{|\psi\rangle}^j\rho)}\geq\frac{\langle\phi|\rho|\phi\rangle}{\theta(\mu)}\,. \qedhere
\end{align*}
\end{proof}

\noindent We have already proven in Lemma~\ref{lem:prepare_rhop} that the post-measurement state can be prepared in subsequent iterations of the algorithm.

\vspace{0.5em}\noindent \textbf{The full algorithm.}\vspace{0.5em}

\noindent Piecing the steps mentioned before together, we get the full algorithm,~\Cref{alg:agnostic_learning_product}.

\begin{algorithm}[htbp]
    \DontPrintSemicolon
    \caption{Agnostic tomography of discrete product states}\label{alg:agnostic_learning_product}
    \KwInput{$\tau>0$, $\mu>0$, $\mu$-packing set $\mathcal{K}$, copies of an $n$-qubit state $\rho$}
    \KwOutput{A product state $\ket{\phi}\in\mathcal{K}^{\otimes n}$}
    \Goal{For every $|\phi\rangle\in\mathcal{K}^{\otimes n}$ with $F(\rho,|\phi\rangle)\geq\tau$, output $|\phi\rangle$ with probability at least $(n|\mathcal{K}|)^{-O(\log(1/\tau)/\mu)}$}
    Set $\mathfrak{P}_0=\varnothing$, $\mathfrak{R}=\varnothing$, $t_{\max}=\lfloor\log_{1/\theta(\mu)}(1/\tau)\rfloor$, $\rho_0=\rho$.\\
    \For{$t=0$ \KwTo $t_{\max}$}{
    Prepare $O(\frac{1}{\mu^2}\log n)$ copies of $\rho_t$ from $\rho$ by \Cref{lem:prepare_rhop} (with $\delta=\frac{1}{6}$, $\mathfrak{P}=\mathfrak{P}_t$). Break the loop if not enough copies are produced.\label{line:product_state_prep}\\
    Run the algorithm in~\Cref{lem:tomography_fidelity} on $\rho_t$ with $\delta=\frac{1}{5}$ and $\epsilon=\frac{\mu}{16}$. For every $j\in[n]$, let $|\psi^j_t\rangle$ be the state in $\mathcal{K}$ that has the highest estimated fidelity with $\tr_{-j}(\rho_t)$.\label{line:product_higest_estimate}\\
    Set $\mathfrak{R}\leftarrow\mathfrak{R}\cup\{\bigotimes_{j=1}^n|\psi^j_t\rangle\}$.\\
    Randomly pick a $j_t\in[n]$ and $|\psi_t\rangle\neq|\psi_t^{j_t}\rangle$.\\
    Define $\mathfrak{P}_{t+1}=\mathfrak{P}_t\cup\{\Pi_{|\psi_t\rangle}^{j_t}\}$, $\rho_{t+1}=\Pi_{|\psi_t\rangle}^{j_t}\rho_t\Pi_{|\psi_t\rangle}^{j_t}/\tr(\Pi_{|\psi_t\rangle}^{j_t}\rho_t)$.
    }
    \Return a uniformly random element $|\phi_r\rangle$ from $\mathfrak{R}$. If $\mathfrak{R}=\varnothing$, return failure.
\end{algorithm}

\subsection{Analysis of the algorithm}
\label{sec:product_analysis}

We now prove that~\Cref{alg:agnostic_learning_product} achieves the guarantee of~\Cref{thm:product_base}.

Fix a state $|\phi\rangle=\bigotimes_{j=1}^n|\phi^j\rangle\in\mathcal{K}$ that has fidelity at least $\tau$ with $\rho$.
We say the algorithm succeeds up to iteration $t$ if, either $\bigotimes_{j=1}^n|\psi^{j_i}\rangle=|\phi\rangle$ for some $0\leq i<t$ (i.e., Step $2$ succeeds at some iteration), or $|\psi_i\rangle=|\phi^{j_i}\rangle$ and $\tr(\Pi_{|\psi_i\rangle}^{j_i}\rho_i)\leq\theta(\mu)$ for all $0\leq i<t$ (i.e., the algorithm reaches iteration $t$ without aborting and Step 3 succeeds at every iteration).
Denote the event by $B_t$.

\begin{lemma}\label{lem:product_middle}
\[\Pr[B_{t+1}|B_t]\geq\frac{2}{3n(|\mathcal{K}|-1)}.\]
\end{lemma}

\begin{proof}
When $B_t$ happens, either $\bigotimes_{j=1}^n|\psi^{j_i}\rangle=|\phi\rangle$ for some $0\leq i<t$, in which case $B_{t+1}$ always happens, or $|\psi_i\rangle=|\phi^{j_i}\rangle$ and $\tr(\Pi_{|\psi_i\rangle}^{j_i}\rho_i)\leq\theta(\mu)$ for all $0\leq i<t$.
In this case, by~\Cref{lem:prepare_rhop}, with probability at least $\frac{5}{6}$, we get enough copies of $\rho_t$.
Then by~\Cref{lem:unique-high-corr}, with probability at least $\frac{4}{5}$, if $|\psi\rangle\in\mathcal{K}$ and $|\psi\rangle\neq|\psi^j_t\rangle$ then $\langle\psi|\tr_{-j}(\rho_t)|\psi\rangle\leq\theta(\mu)$.
Now if $\bigotimes_{j=1}^n|\psi^j_t\rangle=|\phi\rangle$, $B_{t+1}$ happens.
If this is not the case, then with probability at least $\frac{1}{n(|\mathcal{K}|-1)}$, $|\psi_t\rangle=|\phi^{j_t}\rangle$.
When this happens, we are guaranteed that $\tr(\Pi_{|\psi_t\rangle}^{j_t}\rho_t)\leq\theta(\mu)$, so $B_{t+1}$ happens.
Thus
\begin{align*}\Pr[B_{t+1}|B_t]&\geq\min\left\{1,\frac{5}{6}\times\frac{4}{5}\min\left\{1,\frac{1}{n(|\mathcal{K}|-1))}\right\}\right\}=\frac{2}{3n(|\mathcal{K}|-1)}\,.\qedhere
\end{align*}
\end{proof}

\begin{proof}[Proof of~\Cref{thm:product_base}]
Since $B_0$ holds trivially, $\Pr[B_0]=1$.
By~\Cref{lem:product_middle}, we have
\[\Pr[B_{t_{\max}+1}]\geq\Bigl(\frac{2}{3n(|\mathcal{K}|-1)}\Bigr)^{t_{\max}+1}\,,\]
where $t_{\max}$ is defined in Line 1 of~\Cref{alg:agnostic_learning_product}.
Note that $B_{t_{\max}+1}$ means the event that $|\phi\rangle\in\mathfrak{R}$ or $|\psi_i\rangle=|\phi^{j_i}\rangle$ and $\tr(\Pi_{|\psi_i\rangle}^{j_i}\rho_i)\leq\theta(\mu)$ for all $0\leq i\leq t_{\max}$.
But if the later case happens, by~\Cref{lem:fidelity_amplify_product} we have $\langle\phi|\rho_{t_{\max}+1}|\phi\rangle\geq\tau/\theta(\mu)^{t_{\max}+1}>1$, which is impossible.
Thus when $B_{t_{\max}+1}$ happens, $|\phi\rangle\in\mathfrak{R}$.
Since $|\mathfrak{R}|\leq t_{\max}+1$, we have
\[\Pr[|\phi_r\rangle=|\phi\rangle]\geq\frac{\Pr[B_{t_{\max}+1}]}{t_{\max}+1}\geq\frac{1}{1+\log_{1/\theta(\mu)}\frac{1}{\tau}}\Bigl(\frac{2}{3n(|\mathcal{K}|-1)}\Bigr)^{1+\log_{1/\theta(\mu)}\frac{1}{\tau}}=\Bigl(\frac{1}{n|\mathcal{K}|}\Bigr)^{O(\log(1/\tau)/\mu)}\,.\]

During each iteration,~\Cref{line:product_higest_estimate} uses at most $O(\frac{1}{\mu^2}\log n)$ copies of $\rho_t$ by~\Cref{lem:unique-high-corr}.
Hence it suffices to prepare $O(\frac{1}{\mu^2}\log n)$ copies of $\rho_t$ from $O(\frac{1}{\mu^2\tau}\log n)$ copies of $\rho$ at~\Cref{line:product_state_prep}.
The sample complexity is thus
\[(t_{\max}+1)\cdot O\Bigl(\frac{1}{\mu^2\tau}\log n\Bigr)=O\Bigl(\frac{1}{\mu^3\tau}\log\frac{1}{\tau}\log n\Bigr)\,.\]

As for running time,~\Cref{line:product_state_prep} takes $\frac{c_1t}{\mu^2\tau}\log n$ time for some constant $c_1$ by~\Cref{lem:prepare_rhop}.
\Cref{line:product_higest_estimate} takes $\frac{c_2n}{\mu^2}\log n$ time for some constant $c_2$ by~\Cref{lem:unique-high-corr}.
Other lines take constant $c_3$ time.
Hence the time complexity is
\begin{align*}\sum_{t=0}^{t_{\max}}\frac{c_1t}{\mu^2\tau}\log n+\frac{c_2n}{\mu^2}\log n+c_3 &= O\Bigl(\frac{1}{\mu^4\tau}\log^2\frac{1}{\tau}\log n+\frac{n}{\mu^3}\log\frac{1}{\tau}\log n\Bigr)\,. \qedhere\end{align*}
\end{proof}

\section{Agnostic tomography of stabilizer product states}
\label{sec:stabilizer_product} 

Since the set of single-qubit stabilizer states is a $1/2$-packing set, we can use the algorithm in the previous section to agnostically learn the class $\mathcal{SP}$.
However, with the help of Bell difference sampling, we may make a more educated guess for the low-correlation projector at Step $3$, resulting in an improvement in the complexity.

\begin{theorem}\label{thm:stab_product_base}
Fix $\tau>0$ and let $\rho$ be an unknown $n$-qubit state. There is an algorithm with the following guarantee.

Let $|\phi\rangle \in \mathcal{SP}$ be any element of $\mathcal{SP}$ maximizing fidelity with $\rho$, and suppose its fidelity with $\rho$ is at least $\tau$. Given copies of $\rho$, the algorithm outputs $\ket{\phi}$ with probability at least $\tau^{O(\log1/\tau)}$.

The algorithm only performs single-copy and two-copy measurements on $\rho$.
The sample complexity is $O(\frac{1}{\tau}\log\frac{1}{\tau}\log n)$ and the runtime is $O(\frac{1}{\tau}\log^2\frac{1}{\tau}\log n+n\log\frac{1}{\tau}\log n)$.
\end{theorem}

\noindent We give a proof of Theorem~\ref{thm:stab_product_base} in Sections~\ref{sec:stab_prod_alg} and~\ref{sec:stab_prod_analysis}. In the rest of this subsection, we record some consequences of this general result.

Firstly, as in the previous sections, by repeating the algorithm in Theorem~\ref{thm:stab_product_base} sufficiently many times, we obtain a list-decoding guarantee for global maximizers of fidelity. While our proof here applies to global maximizers of fidelity rather than local maximizers, we believe our techniques should carry over to that setting with some more work.

\begin{corollary}\label{cor:stab_prod_list}
Fix $\tau, \delta > 0$, and let $\rho$ be an unknown $n$-qubit state which has fidelity at least $\tau$ with some stabilizer product state. There is an algorithm that, given copies of $\rho$, returns a list of states in $\mathcal{SP}$ of length $\log(1/\delta) \cdot \tau^{O(\log(1/\tau)}$ so that with probability at least $1-\delta$, all states in $\mathcal{SP}$ with maximal fidelity with $\rho$ appear in the list.

The algorithm only performs single-copy and two-copy measurements on $\rho$.
The sample complexity is $O(\log n \log(1/\delta) (1/\tau)^{O(\log 1/\tau)})$ and the runtime is $n\log n \cdot (1/\tau)^{O(\log 1/\tau)}$.
\end{corollary}

\noindent Similarly, this also readily implies an algorithm for proper agnostic tomography of stabilizer product states. The proof details are straightforward and deferred to Appendix~\ref{app:defer_stab_prod_base}.

\begin{corollary}\label{cor:stab_product_base}
Fix $\tau\ge\epsilon>0$, and $\delta>0$. There is an algorithm that, given copies of an $n$-qubit state $\rho$ such that $F_{\mathcal{SP}}(\rho)\ge \tau$, returns a stabilizer product state $\ket{\phi}\in\mathcal{SP}$ that satisfies $F(\rho, \ket{\phi}) \ge F_{\mathcal{SP}}(\rho) - \epsilon$ with probability at least $1 - \delta$. The algorithm only performs single-copy and two-copy measurements on $\rho$. The algorithm uses $\log n\log(1/\delta)(1/\tau)^{O(\log1/\tau)}+(\log^2(1/\tau)+\log(1/\delta))/\epsilon^2$ copies of $\rho$ and $n^2\log^2(1/\delta)(1/\tau)^{O(\log1/\tau)}/\epsilon^2$ time.
\end{corollary}

\subsection{Construction of the algorithm}
\label{sec:stab_prod_alg}

Suppose $|\phi\rangle=\bigotimes_{j=1}^n|\phi^j\rangle\in\argmax_{|\varphi\rangle\in\mathcal{SP}}F(|\varphi\rangle,\rho)$.
Suppose $Q^j$ is the non-trivial single-qubit Pauli operator that stabilizes $|\phi^j\rangle$ up to sign.
That is, $\bigotimes_{j=1}^nQ^j|\phi\rangle=\pm|\phi\rangle$.

Unlike the previous section, the notion of complete family that we will consider here is a set of projectors $\{\Pi^i\}$ where $\Pi^i$ projects the $i$-th qubit using $\frac{I + Q^j}{2}$ and maps the other qubits via the identity. This family is identical to the family of projectors $\Pi^i_{\ket{\psi}}$ with $\ket{\psi}$ ranging over single-qubit stabilizer states, but the parametrization in terms of Paulis will make it more convenient to draw upon our tools related Bell difference sampling. With such a complete family, we can measure in the joint eigenbasis of these projectors and obtain $\ket{\phi}$ with probability at least $\tau$.

\vspace{0.5em}\noindent \textbf{Step 1: Find a high-correlation family.}\vspace{0.5em}

\noindent Instead of finding the single-qubit stabilizer product states that maximize fidelity for each qubit to form the high-correlation family as in the previous section, it is more natural to find the Pauli operators with the highest correlation for each qubit instead.

\begin{lemma}\label{lem:stab_prod_high_unique}
Fix an $n$-qubit state $\rho$.
Run the algorithm in~\Cref{subroutine: estimate Pauli correlation by Bell measurements} to estimate $\tr(P\tr_{-j}(\rho))^2$ for all $P\in\{X,Y,Z\}$ and $j\in[n]$ so that with probability at least $1-\delta$, all fidelities are estimated to error within $0.1$.
Let $P^j$ be the Pauli operator in $\{X,Y,Z\}$ with the highest estimated correlation with $\tr_{-j}(\rho)$.
Then with probability at least $1-\delta$, if $P\in\{X,Y,Z\}$ and $P\neq P^j$, then $\tr(P\tr_{-j}(\rho))^2\leq0.7$.
The sample complexity of this procedure is $O(\log\frac{n}{\delta})$ and the time complexity is $O(n\log\frac{n}{\delta})$.
\end{lemma}

\begin{proof}
Suppose all correlations are estimated to error within $0.1$.
Consider a specific qubit $j$.
Let $P=\argmax_{Q\in\{X,Y,Z\}\backslash\{P^j\}}\tr(Q\tr_{-j}(\rho))^2$.
If $\tr(P\tr_{-j}(\rho))^2>0.5$, then by~\Cref{lem:uncertain}, $\tr(P^j\tr_{-j}(\rho))^2\leq0.5$.
Thus the estimated correlation of $P^j$ with $\tr_{-j}(\rho)$ is at most $0.6$, which must be bigger than or equal to the estimated correlation of $P$ with $\tr_{-j}(\rho)$.
Thus $\tr(P\tr_{-j}(\rho))^2\leq0.7$.
Thus we always have $\tr(P\tr_{-j}(\rho))^2\leq0.7$.
The complexity follows from~\Cref{subroutine: estimate Pauli correlation by Bell measurements}, but note that in this special case, calculating inner products takes $O(1)$ time instead of $O(n)$ time.
\end{proof}

\vspace{0.5em}\noindent \textbf{Step 2: If the family is complete, i.e. if $\bigotimes_{j=1}^nP^j=\bigotimes_{j=1}^nQ^j$, then directly obtain the answer.}\vspace{0.5em}

\noindent If $P^j=Q^j$ for all $j$, then directly measure in the joint eigenbasis $\bigotimes_{j=1}^n\{\frac{I+P^j}{2},\frac{I-P^j}{2}\}$ and we get the result $|\phi\rangle$ with probability at least $\tau$.

\vspace{0.5em}\noindent \textbf{Step 3: If the family is incomplete, sample a low-correlation projector.}\vspace{0.5em}

\noindent If on the other hand $\bigotimes_{j=1}^nP^j\neq\bigotimes_{j=1}^nQ^j$, we know that there exists a qubit $k$ such that $P^k\neq Q^k$.
By~\Cref{thm:prod_progress}, if $|\phi\rangle$ is the stabilizer product state with the highest fidelity, then if we perform Bell difference sampling to get a result $\bigotimes_{j=1}^nR^j$, with probability at least $\frac{1}{4}\tau^4$ we have $\bigotimes_{j=1}^nR^j\in\bigotimes_{1\leq j\leq n,j\neq k}\{I,Q^j\}\otimes\{Q^k\}$.
Assuming this happens, then if we compare $R^j$'s and $P^j$'s, $k$ is a qubit on which the sample is non-identity and different from $P^k$.
There may be many such positions. If we just arbitrarily pick one such position (suppose $k$ is picked and $R^k=R$), then we get $R=Q^k\neq P^k$.
Moreover, since $\bigotimes_{j=1}^nP^j$ has high-correlation, by~\Cref{lem:stab_prod_high_unique}, we have $\tr(R\tr_{-k}(\rho))^2\leq0.7$.
We then guess the sign $\sgn$ for which $R|\phi^k\rangle=\sgn|\phi^k\rangle$.
The correct sign is obtained with probability $\frac{1}{2}$.

\vspace{0.5em}\noindent \textbf{Step 4: Bootstrap by measuring.}\vspace{0.5em}

\noindent Suppose at Step $3$ we get the correct $k$, $R$ and $\sgn$.
Then by measuring we can amplify the fidelity while keeping $|\phi\rangle$ the maximizer of fidelity.

\begin{lemma}\label{lem:stab_prod_ampli}
Given $R\in\{\pm X, \pm Y, \pm Z\}$ and $k\in[n]$, define the projector
\begin{equation*}
    \Pi^k_R \triangleq I^{\otimes k-1} \otimes \Bigl(\frac{I + R}{2}\Bigr)\otimes I^{\otimes n - k}\,.
\end{equation*}
Suppose $\tr(R\tr_{-k}(\rho))^2\leq0.7$, $|\phi\rangle\in\argmax_{|\varphi\rangle\in\mathcal{SP}}F(|\varphi\rangle,\rho)$ and $\Pi_R^k$ stabilizes $|\phi\rangle$.
Consider the post-measurement state $\rho'=\Pi_R^k\rho\Pi_R^k/\tr(\Pi_R^k\rho\Pi_R^k)$.
Then $|\phi\rangle\in\argmax_{|\varphi\rangle\in\mathcal{SP}}F(|\varphi\rangle,\rho')$ and $F(|\phi\rangle,\rho')\geq1.08F(|\phi\rangle,\rho)$.
\end{lemma}

\begin{proof}
Note that $\forall|\varphi\rangle\in\mathcal{SP}$, $\Pi_R^k|\varphi\rangle$ is a stabilizer product state with $\norm{\Pi_R^k|\varphi\rangle}_\infty\leq1$.
Hence 
\[\langle\varphi|\rho'|\varphi\rangle=\frac{\langle\varphi|\Pi_R^k\rho\Pi_R^k|\varphi\rangle}{\tr(\Pi_R^k\rho\Pi_R^k)}\leq\frac{\langle\phi|\rho|\phi\rangle}{\tr(\Pi_R^k\rho\Pi_R^k)}=\langle\phi|\rho'|\phi\rangle,\]
\begin{align*}\langle\phi|\rho'|\phi\rangle=\frac{\langle\phi|\Pi_R^k\rho\Pi_R^k|\phi\rangle}{\tr(\Pi_R^k\rho\Pi_R^k)} &=\frac{\langle\phi|\rho|\phi\rangle}{\frac{1+\tr(R\tr_{-k}(\rho))}{2}}\geq\frac{\langle\phi|\rho|\phi\rangle}{\frac{1+\sqrt{0.7}}{2}}\geq1.08\langle\phi|\rho|\phi\rangle\,.\qedhere\end{align*}
\end{proof}

\vspace{0.5em}\noindent \textbf{The full algorithm.}\vspace{0.5em}

\noindent The full algorithm is shown in~\Cref{alg:agnostic_learning_stabilizer_product_states}.

\begin{algorithm}[htbp]
    \DontPrintSemicolon
    \caption{Agnostic tomography of stabilizer product states}\label{alg:agnostic_learning_stabilizer_product_states}
    \KwInput{$\tau>0$, copies of an $n$-qubit state $\rho$}
    \Promise{$F_{\mathcal{SP}}(\rho)\ge \tau$}
    \KwOutput{A state $\ket{\phi}\in\mathcal{SP}$}
    \Goal{With probability at least $O(1/\tau)^{O(\log 1/\tau)}$, $|\phi\rangle\in\argmax_{|\varphi\rangle\in\mathcal{SP}}F(\rho,|\varphi\rangle)$.}
    Set $\mathfrak{P}_0=\emptyset$, $\mathfrak{R}=\varnothing$, $t_{\max}=\lfloor\log_{1.08}(1/\tau)\rfloor$, $\rho_0=\rho$.\\
    \For{$t=0$ \KwTo $t_{\max}$}{
    Prepare $O(\log n)$ copies of $\rho_t$ by \Cref{lem:prepare_rhop} (with $\delta$ set to $\frac{1}{6}$, $\mathfrak{P}$ set to $\mathfrak{P}_t$). Break the loop if not enough copies are produced.\label{line:stab_prod_prep}\\
    Run the algorithm in~\Cref{subroutine: estimate Pauli correlation by Bell measurements} on $\rho_t$ with $\delta=\frac{1}{5}$ and $\epsilon=0.1$ to estimate the correlations of all single-qubit Pauli operators. Let $P^j_t$ be the operator with the highest estimated correlation with $\tr_{-j}(\rho_t)$.\label{line:stab_prod_find_high}\\
    Measure $\rho$ on the eigenbasis of $\bigotimes_{j=1}^nP^j_t$. Denote the output state by $\ket{\phi_t}$. Set $\mathfrak{R}\leftarrow\mathfrak{R}\cup\{|\phi_t\rangle\}$.\label{line:stab_prod_mea}\\
    Bell difference sampling on $\rho_t$ $1$ time. Suppose the sample is $\bigotimes_{j=1}^nR^j_t$.\label{line:stab_prod_bell}\\
    Find an $j_t\in[n]$ such that $R^{j_t}_t\neq I$ and $R^{j_t}_t\neq P^{j_t}_t$. Break the loop if no such $i$ exists. Let $R_t=R_t^{j_t}$.\label{line:stab_prod_find_diff}\\
    Randomly pick a sign $\sgn_t\in \{\pm 1\}$.\label{line:stab_prod_guess_sgn}

    Define $\mathfrak{P}_{t+1}=\mathfrak{P}_t\cup\{\Pi_{\sgn_tR_t}^{j_t}\}$, $\rho_{t+1}=\Pi_{\sgn_tR_t}^{j_t}\rho_t\Pi_{\sgn_tR_t}^{j_t}/\tr(\Pi_{\sgn_tR_t}^{j_t}\rho_t\Pi_{\sgn_tR_t}^{j_t})$.
    }
    \Return a uniformly random element $|\phi_r\rangle$ from $\mathfrak{R}$. If $\mathfrak{R}=\varnothing$, return failure.
\end{algorithm}

\subsection{Analysis of the algorithm}
\label{sec:stab_prod_analysis}

In this subsection, we prove that~\Cref{alg:agnostic_learning_stabilizer_product_states} satisfies the requirement of~\Cref{thm:stab_product_base}.

Suppose $|\phi\rangle=\bigotimes_{j=1}^n|\phi^j\rangle\in\argmax_{|\varphi\rangle\in\mathcal{SP}}F(|\varphi\rangle,\rho)$.
Suppose $Q^j$ is the non-trivial single-qubit Pauli operator that stabilizes $|\phi^j\rangle$.
That is, $\bigotimes_{j=1}^nQ^j|\phi\rangle=\pm|\phi\rangle$.
We analyze the probability of outputting $|\phi\rangle$.
We say the algorithm succeeds up to iteration $t$ if, either $|\phi_i\rangle=|\phi\rangle$ for some $0\leq i<t$, (i.e., Step $2$ succeeds at some iteration), or $R_i|\phi^{j_t}\rangle=\sgn_t|\phi^{j_t}\rangle$ and $\tr(R_i\tr_{-j_i}(\rho_i))^2\leq0.7$ for all $0\leq i<t$ (i.e., the algorithm reaches iteration $t$ without aborting and Step $3$ succeeds at every iteration).
Denote the event by $B_t$.

\begin{lemma}\label{lem:stab_prod_middle}
Define $\tau_t=1.08^t\tau$.
Then
\[\Pr[B_{t+1}|B_t]\geq\frac{1}{6}\tau\tau_t^3.\]
\end{lemma}

\begin{proof}
When $B_t$ happens, either $|\phi_i\rangle=|\phi\rangle$ for some $0\leq i<t$, in which case $B_{t+1}$ always happens, or $R_i|\phi^{j_t}\rangle=\sgn_t|\phi^{j_t}\rangle$ and $\tr(R_i\tr_{-j_i}(\rho_i))^2\leq0.7$ for all $0\leq i<t$.
In this case, by~\Cref{lem:prepare_rhop}, with probability at least $\frac{5}{6}$, we get enough copies of $\rho_t$.
Then by~\Cref{lem:stab_prod_high_unique}, with probability at least $\frac{4}{5}$, if $R\in\{X,Y,Z\}$ and $R\neq P^j_t$, then $\tr(R\tr_{-j}(\rho))^2\leq0.7$.
Now if $\bigotimes_{j=1}^nP^j_t=\bigotimes_{j=1}^nQ^j$, with probability at least $\tau$, the measurement result $|\phi_t\rangle$ at~\Cref{line:stab_prod_mea} is $|\phi\rangle$, in which case $B_{t+1}$ happens.
If on the other hand $\bigotimes_{j=1}^nP^j_t\neq\bigotimes_{j=1}^nQ^j$, by~\Cref{lem:stab_prod_ampli} we have $F(|\phi\rangle,\rho_t)\geq1.08^t\tau=\tau_t$ and $|\phi\rangle\in\argmax_{|\varphi\rangle\in\mathcal{SP}}F(|\varphi\rangle,\rho_t)$.
Thus, as argued before with probability at least $\frac{1}{4}\tau_t^4$ the algorithm does not exit at~\Cref{line:stab_prod_find_diff} and we have $R_t=Q^{j_t}$.
Then with probability $\frac{1}{2}$ we guessed the correct sign at~\Cref{line:stab_prod_guess_sgn} so that $R_t|\phi^{j_t}\rangle=\sgn_t|\phi^{j_t}\rangle$, and then $B_{t+1}$ happens.
Thus
\[\Pr[B_{t+1}|B_t]\geq\min\left\{1,\frac{5}{6}\times\frac{4}{5}\min\left\{\tau,\frac{1}{4}\tau_t^4\right\}\right\}\geq\frac{1}{6}\tau\tau_t^3.\]
\end{proof}

\begin{proof}[Proof of~\Cref{thm:stab_product_base}]
Since $B_0$ holds trivially, $\Pr[B_0]=1$.
By~\Cref{lem:stab_prod_middle}, we have
\[\Pr[B_{t_{\max}+1}]\geq\prod_{t=0}^{t_{\max}}\frac{1}{6}\tau\tau_t^3\,,\]
where $t_{\max}$ is defined in Line 1 of~\Cref{alg:agnostic_learning_stabilizer_product_states}.
Note that $B_{t_{\max}+1}$ means the event that $|\phi\rangle\in\mathfrak{R}$ or $R_i|\phi^{j_t}\rangle=\sgn_t|\phi^{j_t}\rangle$ and $\tr(R_i\tr_{-j_i}(\rho_i))^2\leq0.7$ for all $0\leq i<t_{\max}$.
But if the later case happens, by~\Cref{lem:stab_prod_ampli} we have $F(|\phi\rangle,\rho_{t_{\max}+1})\geq\tau_{t_{\max}+1}>1$, which is impossible.
Thus when $B_{t_{\max}+1}$ happens, $|\phi\rangle\in\mathfrak{R}$.
Since $|\mathfrak{R}|\leq t_{\max}+1$, we have
\[\Pr[|\phi_r\rangle=|\phi\rangle]\geq\frac{\Pr[B_{t_{\max}+1}]}{t_{\max}+1}\geq\frac{1}{1+\log_{1.08}\frac{1}{\tau}}\prod_{t=0}^{t_{\max}}\frac{1}{6}\tau\tau_t^3=\tau^{O(\log1/\tau)}.\]

During each iteration,~\Cref{line:stab_prod_find_high} uses $O(\log n)$ copies of $\rho_t$ by~\Cref{lem:stab_prod_high_unique}.
\Cref{line:stab_prod_mea} uses $1$ copy of $\rho$.
\Cref{line:stab_prod_bell} uses $4$ copies of $\rho_t$.
Hence it suffices to prepare $O(\log n)$ copies of $\rho_t$ at~\Cref{line:stab_prod_prep} from $O(\frac{1}{\tau}\log n)$ copies of $\rho$.
The sample complexity is thus
\[(t_{\max}+1)O\Bigl(\frac{1}{\tau}\log n\Bigr)=O\Bigl(\frac{1}{\tau}\log\frac{1}{\tau}\log n\Bigr).\]

As for running time,~\Cref{line:stab_prod_prep} takes $\frac{c_1t}{\tau}\log n$ time for some constant $c_1$ by~\Cref{lem:prepare_rhop}. \Cref{line:stab_prod_find_high} takes $c_2n\log n$ time for some constant $c_2$. \Cref{line:stab_prod_mea} takes $c_3n$ time for some constant $c_3$. \Cref{line:stab_prod_bell} takes $c_4n$ time for some constant $c_4$. \Cref{line:stab_prod_find_diff} takes $c_5n$ time for some constant $c_5$. \Cref{line:stab_prod_guess_sgn} takes $c_6$ time for some constant $c_6$. Hence the time complexity is
\begin{align*}\sum_{t=0}^{t_{\max}}\frac{c_1t}{\tau}\log n+c_2n\log n+c_3n+c_4n+c_5n+c_6 &=O\Bigl(\frac{1}{\tau}\log^2\frac{1}{\tau}\log n+n\log\frac{1}{\tau}\log n\Bigr)\,.\qedhere\end{align*}
\end{proof}

\section{Lower bounds for agnostic tomography of stabilizer states}\label{sec:lower}
In this section, we return to the task of agnostic tomography of stabilizer states. 

Recall that in \Cref{sec:stabilizer}, we gave an algorithm which, given an $n$-qubit state $\rho$ with stabilizer fidelity $F_{\cS}(\rho)\geq\tau$, can output a stabilizer state $\ket{\phi}\in\cS$ such that $F(\rho,\ket{\phi})\geq \tau-\epsilon$ with $\poly(n,1/\epsilon)\cdot(1/\tau)^{O(\log(1/\tau))}$ computational and sample complexity (see \Cref{cor:agnostic_learning_stabilizer} and \Cref{alg:agnostic_learning_stabilizer}). This yields an efficient algorithm given that the stabilizer fidelity for the state is sub-polynomially small, i.e. $\tau\geq\mathrm{exp}(-O(\sqrt{\log n}))$. However, the computational complexity will increase to quasipolynomial in $n$ when $\tau=1/\poly(n)$. It is thus natural to ask for a computationally efficient algorithm for agnostic tomography of stabilizer state when $\tau = 1/\poly(n)$, or perhaps even when $\tau = o(1/\poly(n))$. 

In Section~\ref{sec:lower_superpoly}, we observe that the latter is not possible, even information-theoretically. In Section~\ref{sec:lower_poly}, we informally discuss the possibility of showing computational hardness when $\tau = 1/\poly(n)$ based on nonstandard variants of the popular \emph{learning parity with noise}, by connecting a special case of agnostic tomography of stabilizer states to a question about finding dense subspace approximations.

\subsection{Hardness for super-polynomially small stabilizer fidelity}
\label{sec:lower_superpoly}

We first provide an information-theoretic lower bound for agnostic tomography of stabilizer state with input state of stabilizer fidelity $\tau$:

\begin{theorem}\label{thm:info_lower}
Assume $0<\epsilon<\frac{2^n\tau-1}{2(2^n-1)}$, the sample complexity of agnostic tomography of stabilizer states within accuracy $\epsilon$ with high probability for input states $\rho$ which have fidelity at most $\tau$ with respect to any stabilizer state is at least $\Omega(n/\tau)$. 
\end{theorem}

\begin{proof}
We proceed via a Fano's-type argument. We index all stabilizer states by $\ket{\phi_i}\in\cS$ for $i=1,...,\abs{\cS}$ where $\abs{\cS}=2^{\Theta(n^2)}$~\cite{garcia2014geometry}. We consider the set of states $\{\rho_i\}_i$ where
\begin{equation*}
\rho_i=\frac{2^n\tau-1}{2^n-1}\ket{\phi_i}\bra{\phi_i}+\left(1-\frac{2^n\tau-1}{2^n-1}\right)\frac{I}{2^n}\,.
\end{equation*}
The stabilizer fidelity of every $\rho_i$ is $F_{\cS}(\rho_i)=\tau$, and $\ket{\phi_i}$ achieves this fidelity with $\rho_i$. Given $j\neq i$,
\begin{equation*}
F(\ket{\psi_j},\rho)=\bra{\phi_j}\rho_i\ket{\phi_j}\leq\tau-\frac{2^n\tau-1}{2(2^n-1)}=F(\ket{\psi_i},\rho_i)-\frac{2^n\tau-1}{2(2^n-1)}\,,
\end{equation*}
as $\abs{\langle \phi_i|\phi_j\rangle}^2\leq1/2$ for $i\neq j$~\cite{garcia2014geometry}. Therefore, if we consider agnostic tomography of stabilizer states for input state $\rho_i$ within error $\epsilon<\frac{2^n\tau-1}{2(2^n-1)}$, $\ket{\psi_i}$ is the unique solution for the task.

Suppose the sample complexity is $T$.
Consider the following scenario: suppose Alice has a uniformly random bit string $a$ of length $\log_2|\mathcal{S}|=\Theta(n^2)$.
She wants to share the information with Bob, so she sends $T$ copies of the state $\rho_a$ to Bob.
Bob then runs an agnostic tomography algorithm with large constant success probability on the $T$ copies he received. Suppose he gets the result $|\phi_b\rangle$.
Then by the above argument, we have $b=a$ provided the agnostic tomography algorithm succeeds.
Thus, by Fano's inequality, the mutual information between Alice and Bob is $I(a,b)=\Theta(n^2)$.
But by the Holevo bound, $I(a,b)$ is bounded by the Holevo information $\chi \triangleq S(\E_i \rho^{\otimes T})- \E_i S(\rho_i^{\otimes T})$.
Here, $S(\rho)=-\tr(\rho\log_2\rho)$ denotes the von Neumann entropy. 
Since $S(\rho) \le n$ for any $n$-qubit state, and $S(\rho^{\otimes T}) = TS(\rho)$, we have
\begin{equation}
    \chi \le T(n - \E_i S(\rho_i))\,.
\end{equation}
Note that $S(\rho_i)$ can be bounded by
\begin{align*}
S(\rho_i)&=S\left(\frac{2^n\tau-1}{2^n-1}\ket{\phi_i}\bra{\phi_i}+\left(1-\frac{2^n\tau-1}{2^n-1}\right)\frac{I}{2^n}\right)\\
&\geq\frac{2^n\tau-1}{2^n-1}S(\ket{\phi_i}\bra{\phi_i})+\left(1-\frac{2^n\tau-1}{2^n-1}\right)S\left(\frac{I}{2^n}\right)\\
&=\left(1-\frac{2^n\tau-1}{2^n-1}\right)n\,,
\end{align*}
where the second step is by concavity of von Neumann entropy. Thus we have
\begin{align*}
\chi\leq Tn\cdot\frac{2^n\tau-1}{2^n-1}\,.
\end{align*}
In order for $\chi = \Omega(n^2)$, we must have $T=\Omega(n/\tau)$.
\end{proof}

\noindent Although~\Cref{thm:info_lower} applies in the general setting where the input state can be mixed, we can also prove an $\Omega(\tau^{-1})$ sample complexity lower bound for pure states (see \Cref{lem:magic_lower_agnostic} below). A similar construction also yields an $\Omega(\epsilon^{-1})$ sample complexity lower bound for estimating the stabilizer fidelity of quantum states. 

\begin{lemma}\label{lem:magic_lower_agnostic}
For any $\tau>\epsilon+(3/4)^{n/2}+(1/2)^{n/2-1}$, the sample complexity of agnostic tomography of stabilizer states for \emph{pure} input states $\ket{\psi}$ which have fidelity at most $\tau$ with respect to any stabilizer state is at least $\Omega(\tau^{-1})$. 
\end{lemma}

\begin{proof}
Without loss of generality, we assume $n$ is an even number. Consider the task of distinguishing between the following two cases:
\begin{itemize}
    \item The state $|\psi\rangle$ is $\ket{\psi_1}=\sqrt{1-\tau}\ket{\zeta}+\sqrt{\tau}\ket{1^n}$. 
    \item The state $|\psi\rangle$ is $\ket{\psi_2}=\sqrt{1-\tau}\ket{\zeta}+\sqrt{\tau}\ket{+^{n-2}}\otimes\ket{11},$ where $\ket{+}=\frac{1}{\sqrt{2}}(\ket{0}+\ket{1})$.
\end{itemize}
Here $|\zeta\rangle=\Bigl(\frac{|00\rangle+|01\rangle+|10\rangle}{\sqrt{3}}\Bigr)^{\otimes n/2}$, a state with stabilizer fidelity at most $(3/4)^{n/2}$ (See Theorem 22 of Ref.~\cite{garcia2014geometry}).
For any stabilizer state $|\phi\rangle$, we thus have
\begin{align*}
\MoveEqLeft F(|\phi\rangle,|\psi_1\rangle)\\
&=\Bigl|\sqrt{1-\tau}\langle\phi|\zeta\rangle+\sqrt{\tau}\langle\phi|1^n\rangle\Bigr|^2\\
&\leq\Bigl(\sqrt{1-\tau}|\langle\phi|\zeta\rangle|+\sqrt{\tau}|\langle\phi|1^n\rangle|\Bigr)^2\\
&\leq\Bigl(\sqrt{1-\tau}(\sqrt{3}/2)^{n/2}+\sqrt{\tau F(|\phi\rangle,|1^n\rangle)}\Bigr)^2.
\end{align*}
Similarly, we have
\[F(|\phi\rangle,|\psi_2\rangle)\leq\Bigl(\sqrt{1-\tau}(\sqrt{3}/2)^{n/2}+\sqrt{\tau F(|\phi\rangle,|+^{n-2}\rangle\otimes|11\rangle)}\Bigr)^2.\]

Moreover, it is known~\cite{aaronson2004improved} that for oblique stabilizer states $|\phi_1\rangle$ and $|\phi_2\rangle$, their fidelity is $F(|\phi_1\rangle,|\phi_2\rangle)=2^{-s}$ where $s$ is the minimum number of different generators of the stabilizer group of $|\phi_1\rangle$ and $|\phi_2\rangle$.
By this operational meaning of $s$, it is clear that $s$ satisfies triangular inequality, that is, $F(|\phi_1\rangle,|\phi_2\rangle)F(|\phi_2\rangle,|\phi_3\rangle)\leq F(|\phi_1\rangle,|\phi_3\rangle)$ for pairwise oblique stabilizer states $|\phi_1\rangle$, $|\phi_2\rangle$ and $|\phi_3\rangle$.
We thus have \begin{equation*}
    F(|\phi\rangle,|1^n\rangle)F(|\phi\rangle,|+^{n-2}\rangle\otimes|11\rangle)\leq(1/2)^{n-2}\,.
\end{equation*}
Hence
\[\min\{F(|\phi\rangle,|\psi_1\rangle),F(|\phi\rangle,|\psi_2\rangle)\}\leq\Bigl(\sqrt{(1-\tau)(3/4)^{n/2}}+\sqrt{\tau(1/2)^{n/2-1}}\Bigr)^2\leq(3/4)^{n/2}+(1/2)^{n/2-1}.\]
That is, for $o\in\{0,1\}$, if $F(|\phi\rangle,|\psi_o\rangle)>(3/4)^{n/2}+(1/2)^{n/2-1}$, then $F(|\phi\rangle,|\psi_{1-o}\rangle)<F(|\phi\rangle,|\psi_o\rangle)$.

Since $F_{\mathcal{S}}(|\psi_1\rangle)\geq F(|\psi_1\rangle,|1^n\rangle)=\tau$ and $F_{\mathcal{S}}(|\psi_2\rangle)\geq F(|\psi_2\rangle,|+^{n-2}\rangle\otimes|11\rangle)=\tau$, we can use the agnostic tomography algorithm to solve the distinguishing task: run the agnostic tomography algorithm on $|\psi\rangle$ and obtain result $|\phi\rangle$.
By arguments above, $|\phi\rangle$ has higher fidelity with the true $|\psi_o\rangle$, so among $|\psi_1\rangle$ and $|\psi_2\rangle$, output the one with higher fidelity with $|\phi\rangle$ and we will be able to solve the distinguishing task successfully.

On the other hand, note that $F(\ket{\psi_1},\ket{\psi_2})=(1-\tau+\tau/\sqrt{2}^{n-2})^2\geq(1-\tau)^2$.
Therefore, the trace distance between $T$ copies from the two cases is bounded by
\begin{align*}
\MoveEqLeft D_{\tr}(\ketbra{\psi_1}^{\otimes T},\ketbra{\psi_2}^{\otimes T})\\
&\leq\sqrt{1-F(\ket{\psi_1}^{\otimes T},\ket{\psi_2}^{\otimes T})}\\
&=\sqrt{1-F(\ket{\psi_1},\ket{\psi_2})^T}\\
&\leq\sqrt{1-(1-\tau)^{2T}}\,.
\end{align*}
But by Helstrom's theorem~\cite{helstrom1969quantum}, to distinguish with high probability, the trace distance must be of order $\Omega(1)$, thus the agnostic tomography algorithm must have sample complexity $T=\Omega(1/\tau)$.
\end{proof}

\noindent A similar idea yields the following sample complexity lower bound for estimating stabilizer fidelity.

\begin{lemma}\label{lem:magic_lower}
For any $\epsilon\geq(3/4)^{n/2}$, the sample complexity of estimating the stabilizer fidelity of pure input state $\ket{\psi}$ to within accuracy $\epsilon$ is at least $\Omega(1/\epsilon)$.
\end{lemma}

\begin{proof}
With out loss of generality, assume $n$ is even.
We consider the distinguishing task between the following two cases:
\begin{itemize}
    \item The state is $\ket{\psi_1}=\ket{\zeta}$.
    \item The state is $\ket{\psi_2}=\sqrt{1-4\epsilon}\ket{\zeta}+\sqrt{4\epsilon}\ket{1^n}$. 
\end{itemize}
Here again $|\zeta\rangle=\Bigl(\frac{|00\rangle+|01\rangle+|10\rangle}{\sqrt{3}}\Bigr)^{\otimes n/2}$.
While $F_{\cS}(\ket{\psi_1})\leq (3/4)^{n/2}$, we have $F_{\cS}(\ket{\psi_2})\ge 4\epsilon$. As $\epsilon\geq(3/4)^{n/2}$, we have $F_{\cS}(\ket{\psi_2})-F_{\cS}(\ket{\psi_1})\geq3\epsilon$. We can thus distinguish between these two cases by estimating stabilizer fidelity to accuracy $\epsilon$.

On the other hand, the fidelity between $\ket{\psi_1}$ and $\ket{\psi_2}$ is $F(\ket{\psi_1},\ket{\psi_2})=1-4\epsilon$. Therefore, the trace distance between $T$ copies from the two cases is bounded by
\begin{align*}
\MoveEqLeft D_{\tr}(\ketbra{\psi_1}^{\otimes T},\ketbra{\psi_2}^{\otimes T})\\
&=\sqrt{1-F(\ket{\psi_1}^{\otimes T},\ket{\psi_2}^{\otimes T})}\\
&=\sqrt{1-F(\ket{\psi_1},\ket{\psi_2})^T}\\
&=\sqrt{1-(1-4\epsilon)^{2T}}\,.
\end{align*}
Any algorithm thus requires sample complexity at least $T=\Omega(\epsilon^{-1})$ to distinguish between these two cases.
Thus the sample complexity for estimating stabilizer fidelity is $\Omega(1/\epsilon)$.
\end{proof}

\begin{remark}[Comparison to pseudo-magic states]\label{remark:compare_pseudo}
The previous work~\cite{gu2024pseudomagic} on constructing pseudo-magic states already implied a slightly weaker $\Omega(\epsilon^{-1/2})$ \emph{sample} complexity for estimating the stabilizer fidelity of input states within $\epsilon$. 

We briefly outline the argument therein. Their construction uses the subset phase states~\cite{aaronson2024quantum,gu2024pseudomagic}: for any function $f:\{0,1\}^n\to\{0,1\}$ and subset $S\subseteq\{0,1\}^n$, the associated subset phase state is $\ket{\psi_{f,S}}\triangleq\frac{1}{\sqrt{\abs{S}}}\sum_{x\in S}(-1)^{f(x)}\ket{x}$.
One can consider the task of distinguishing whether the state is randomly chosen from the Haar random state ensemble $\mathcal{E}_{\text{Haar}}$ or the state is randomly chosen from the ensemble $\mathcal{E}=\{\ket{\psi_{f,S}}\}$ containing all subset phase states with a fixed $\abs{S}=K$. Regarding these two ensembles, the trace distance between $T$ copies from the two cases is bounded by 
\[D_{\tr}(\mathbb{E}_{\ket{\psi}\sim\mathcal{E}_{\text{Haar}}}[\ketbra{\psi}^{\otimes T}], \mathbb{E}_{\ket{\psi}\sim\mathcal{E}}[\ketbra{\psi}^{\otimes T}]) \le O(T^2/K)\]
for any $T\leq K\leq 2^n$~\cite{aaronson2024quantum}.
For a state $\ket{\psi}$ randomly chosen from the subset phase state ensemble $\mathcal{E}_{\text{Haar}}$, we have $F_{\cS}(\ket{\psi})=\mathrm{exp}(-\Theta(n))$ with high probability~\cite{grewal2023low}. However, for a state $\ket{\psi_{f,S}}$ randomly chosen from the Haar random ensemble $\mathcal{E}=\{\ket{\psi_{f,S}}\}$, we trivially have $F_{\cS}(\ket{\psi_{f,S}})\geq|S|^{-1}$. We can then use an algorithm for estimating stabilizer fidelity to within error $\epsilon$ to distinguish states from the two ensembles, provided $|S|=K=(2\epsilon)^{-1}$. Hence at least $T=\Omega(\epsilon^{-1/2})$ copies are necessary to distinguish the two cases, and thus to estimate the stabilizer fidelity within $\epsilon$.

While this implies a lower bound for estimating stabilizer fidelity (as well as other notions of magic), it does not immediately imply a lower bound for agnostic tomography of stabilizer states. Consider the following naive approach for reducing from the above distinguishing task to agnostic tomography. If we try to run an agnostic tomography algorithm with error $\epsilon=\tau/2=O(\tau)$, we will obtain an output state whose fidelity with $\rho$ is close to the true stabilizer fidelity. But to estimate this fidelity to sufficient accuracy to distinguish subset phase states from Haar-random, one would need $\Theta(\tau^{-1})$ copies of $\rho$ just to estimate the fidelity between the output state and $\rho$. The sample complexity for estimating fidelity already exceeds the lower bound $\Omega(\tau^{-1/2})$, rendering this reduction invalid.
\end{remark}

\subsection{On the potential hardness for polynomially small stabilizer fidelity}\label{sec:lower_poly}
\Cref{thm:info_lower} rules out the possibility of efficient algorithms for super-polynomially small stabilizer fidelity. There still exists a gap between this lower bound and our algorithm which has $\mathrm{poly}(n,1/\epsilon)$ runtime for $\tau\ge \mathrm{exp}(-c\sqrt{\log n})$. It is natural to ask whether there exists an efficient algorithm even when $\tau = 1/\mathrm{poly}(n)$. Here we examine a natural barrier to improving upon our runtime.

First, just to fix terminology, let us give the following name to the task of proper agnostic tomography of $n$-qubit stabilizer states with parameters $\tau, \epsilon$:

\begin{problem}[$(n, \tau, \epsilon)$-Closest Stabilizer State]\label{problem: closest stabilizer state}
    Given copies of an $n$-qubit state $\rho$ with stabilizer fidelity $F_{\cS}(\rho)\ge \tau$, output an $n$-qubit stabilizer state $\ket{\phi}\in \cS$ such that $F(\rho, \ket{\phi})\ge F_{\cS}(\rho)-\epsilon$ with probability at least $2/3$.
\end{problem}

\noindent By specializing the unknown state $\rho$ to a subset state, we find that this problem simplifies into the following problem. For two finite set $A, B$, define the relative size of intersection as $I(A, B)=\abs{A\cap B}^2/(\abs{A}\abs{B})$.

\begin{problem}[$(n, \epsilon)$-Densest Affine Subspace]\label{problem: densest affine space}
    Given a polynomially-sized subset $A$ of $\mathbb{F}_2^n$, output an affine subspace $V\subseteq\mathbb{F}^{2n}_2$ such that $I(A, V)\ge \max_U I(A, U)-\epsilon$, 
    where $U$ ranges over all affine subspaces of $\mathbb{F}^n_2$.
\end{problem}

\begin{lemma}\label{lem: poly lower bound reduction}
    For $\epsilon = 1/\poly(n)$, if there is a $\mathrm{poly}(n)$-time algorithm that solves $(n, \epsilon, \epsilon)$-Closest Stabilizer State, then there is a $\mathrm{poly}(n)$-time quantum algorithm that solves $(n, \epsilon)$-Densest Affine Subspace.
\end{lemma}
\begin{proof}
    Let $A\subseteq \mathbb{F}_2^n$ be an instance of $(n, \epsilon)$-Densest Affine Subspace. We consider the corresponding subset state
    \begin{equation*}
        \ket{\psi_{A}}=\frac{1}{\sqrt{\abs{A}}}\sum_{x\in A}\ket{x}.
    \end{equation*}
    Since the size of $A$ is polynomial, we can prepare $\ket{A}$ efficiently. Indeed, when $A=\{0, 1, \cdots, \abs{A}-1\}$, this is equivalent to preparing the uniform superposition of the first $\abs{A}$ integers, which has an efficient algorithm \cite{shuklaEfficientQuantumAlgorithm2024}. For general $A$, we only need to apply at most $\abs{A}$ additional basis swap gates. Here a basis swap gate is a unitary that maps $\ket{i}\leftrightarrow \ket{j}$ for some $i\neq j$ and fixes other basis states. When $i=1^{n-1}0,j=1^{n-1}1$, the basis swap gate is just the multi-qubit controlled-NOT gate, and thus can be implemented efficiently~\cite{nielsen2010quantum,barenco1995elementary}. For general $(i, j)$, we permute the basis states so that $(i,j)\to (1^{n-1}0, 1^{n-1}1)$, apply the multi-qubit controlled-NOT gate, and permute back.

    Assume there exists an efficient algorithm $\mathcal{A}$ that solves $(n, \epsilon, \epsilon)$-Closest Stabilizer State. 
    Running $\mathcal{A}$ on state $\ket{\psi_A}$, we obtain a stabilizer state $\ket{\phi}$. A well-known fact is that $\ket{\phi}$ has a canonical form up to a global phase \cite{dehaeneCliffordGroupStabilizer2003} 
    \begin{equation}\label{eq: stabilizer state canonical form}
        \ket{\phi}=\frac{1}{\sqrt{\abs{B}}}\sum_{y\in B}i^{y^T Q y}(-1)^{c^T y}\ket{y},
    \end{equation}
    where $B\subseteq\mathbb{F}^{n}_2$ is an affine subspace, $Q$ is a symmetric binary matrix, and $c$ is a binary vector. It is easy to see that $F(\ket{\psi_A}, \ket{\phi})\leq I(A, B)$ and $F_{\cS}(\ket{\psi_A})=\max_{U}I(A, U)$ (because for every affine space $U$, the corresponding subset state is a stabilizer state). 

    If $\max_U I(A, U)\leq \epsilon$, then $I(A, B)\ge 0\ge \max_U I(A, U)-\epsilon$. Otherwise if $F_{\cS}(\ket{\psi_A})=\max_U I(A, U)\geq \epsilon$, by definition of $\mathcal{A}$, with probability at least $2/3$, $I(A, B)\ge F(\ket{\psi_A}, \ket{\phi})\ge F_{\cS}(\ket{\psi_A})-\epsilon=\max_U I(A, U)-\epsilon$. Therefore, $B$ is a desired output and we obtain an efficient algorithm that solves $(n, \epsilon)$-Densest Affine Subspace.
\end{proof}

\noindent Note that the affine subspace $U$ that maximizes $I(A, U)$ must have dimension at most $2\log_2(\abs{A})$ because otherwise
\begin{equation*}
    I(A, U)=\frac{\abs{A\cap U}^2}{\abs{A}\abs{U}}\leq \frac{\abs{A}}{\abs{U}}< \frac{1}{\abs{A}} = I(A, \{a\}),~\forall a\in A.
\end{equation*}
Therefore, there is a trivial quasipolynomial time algorithm for $(n, 0)$-Densest Affine Subspace: simply enumerate every $(r+1)$-tuple $(a_0,a_1,\cdots, a_r)$ of $A$ for $r\leq 2\log(\abs{A})$, calculate $I(A, a_0+\spn(a_1-a_0, \cdots, a_r-a_0))$ for each tuple, and output the affine subspace with the largest value. In this special case of subset states, this matches the quasipolynomial runtime of our agnostic tomography algorithm. To the best of our knowledge, there is no known classical algorithm that outperforms this trivial runtime.

While the hardness of Densest Affine Subspace is not a standard assumption, it bears resemblance to some well-studied cryptographic assumptions, as we discuss next. Consider the following closely related but possibly harder problem.

\begin{problem}[$(n, t, \alpha)$-Max Intersection Affine Subspace]
    Given a polynomially-sized subset $A$ of $\mathbb{F}_2^n$, output a $t$-dimensional affine subspace $V\subseteq\mathbb{F}^n_2$ such that $\abs{A\cap V}\ge (1-\alpha)\max_U \abs{A\cap U}$, 
    where $U$ ranges over all \emph{$t$-dimensional} affine subspaces of $\mathbb{F}^n_2$.
\end{problem}

The main difference between Max Intersection Affine Subspace and Densest Affine Subspace is that the dimension of the affine subspace is fixed beforehand so that we can ignore the denominator $\sqrt{\abs{A}\abs{U}}$ in the definition of $I(A,U)$. We can efficiently solve $(n, 1/\poly(n))$-Densest Affine Subspace 
if there is an efficient algorithm for $(n, O(\log n), 1/\poly(n))$-Max Intersection Affine Subspace, by enumerating $t\leq 2\log(\abs{A})=O(\log n)$. So far we are not aware of any reduction from the other direction so it is possible that the latter problem is harder. 

Note that an efficient algorithm for $(n, n-1, 1/\poly)$-Max Intersection Affine Subspace would break the \emph{learning parity with noise (LPN)} assumption, a standard cryptographic assumption that is believed to be quantumly secure\cite{pietrzak2012cryptography}. Similarly, an efficient solution to $(n, \beta n, 1/\poly(n))$-Max Intersection Affine Subspace for any constant $\beta$ would break the \emph{learning subspace with noise (LSN)} assumption \cite{dodis2009cryptography}, which is less standard yet no attack is known. From this perspective, Densest Affine Subspace is a natural extension of these problems to the regime where $t = O(\log n)$. An efficient solution to $(n, 1/\poly(n), 1/\poly(n))$-Closest Stabilizer State may thus shed light on the quantum security, or absence thereof, of this line of cryptographic assumptions.

\section*{Acknowledgments}
We thank Anurag Anshu, Sabee Grewal, Jonas Haferkamp, Xingjian Li, and Chenyi Zhang for illuminating discussions about agnostic tomography, stabilizer states, magic estimation, and pseudoentanglement. We thank Salvatore F.E. Oliviero and Tobias Haug for helpful feedback on an earlier version of this manuscript and pointers to the resource theory literature. We thank the authors of~\cite{bakshi2024learninga} for sharing details of their ongoing investigations in agnostic tomography, including their concurrent work on agnostic tomography of discrete product states.


\bibliographystyle{MyRefFont}
\bibliography{NoteRef}
\clearpage
\appendix

\section{Deferred proofs}
\label{app:defer}

\subsection{Proofs from Section \ref{sec:geometry}}
\subsubsection{Proof of Lemma \ref{lem: fidelity high dimension}}\label{app:defer_fidelity}
$\sqrt{\sigma}=\ketbra{s}\otimes \sqrt{\sigma_0}$. By definition of fidelity,
\begin{align*}
    F(\rho, \sigma)&=\tr(\sqrt{\sqrt{\sigma}\rho\sqrt{\sigma}})^2=\tr(\sqrt{(\ketbra{s}\otimes \sqrt{\sigma_0})\rho(\ketbra{s}\otimes \sqrt{\sigma_0})})^2\\
    &=\tr(\sqrt{\ketbra{s}\otimes (\sqrt{\sigma_0}~\braket{s|\rho|s}~\sqrt{\sigma_0})})^2=\tr(\braket{s|\rho|s})\tr(\ketbra{s}\otimes \sqrt{\sqrt{\sigma_0}\rho_s\sqrt{\sigma_0}})^2\\
    &=\tr(\braket{s|\rho|s})\tr(\sqrt{\sqrt{\sigma_0}\rho_s\sqrt{\sigma_0}})^2\\
    &=\tr(\braket{s|\rho|s})F(\rho_s, \sigma_0).
\end{align*}

\subsubsection{Proof of Lemma~\ref{lem:packing}}
\label{app:defer_packing}

We use the well-known isomorphism between $\mathbb{CP}^1$ and the two-dimensional sphere $\mathbb{S}^2$ (the Bloch sphere).
For $|\psi\rangle\in\mathcal{K}\subseteq\mathbb{CP}^1$, map it to the point $u=(\langle\psi|X|\psi\rangle,\langle\psi|Y|\psi\rangle,\langle\psi|Z|\psi\rangle)$.
Since the Pauli operators form an orthogonal basis, one can verify that $|\psi\rangle\!\langle\psi|=\frac{I+u_1X+u_2Y+u_3Z}{2}$ and that $1=\tr(|\psi\rangle\!\langle\psi|^2)=\frac{1+u\cdot u}{2}$, so $u\in \mathbb{S}^2$ is a unit vector.
Similarly, suppose another state $|\psi'\rangle\in\mathcal{K}$ is mapped to $u'\in \mathbb{S}^2$.
Then we have $1-\mu\geq F(|\psi\rangle,|\psi'\rangle)=\frac{1+u\cdot u'}{2}$.
Hence, the geodesic distance between $u$ and $u'$ is at least $\arccos(1-2\mu)$.
Thus if we draw a circle centering at each of the mapped points with radius $\alpha=\frac{1}{2}\arccos(1-2\mu)$, then the circles will not overlap with each other.
The area of one circle is $2\pi(1-\cos\alpha)$, and since the total area cannot exceed the area of the unit sphere, we have $|\mathcal{K}|\leq\frac{4\pi}{2\pi(1-\cos\alpha)}=\frac{2}{1-\sqrt{1-\mu}}=O\left(\frac{1}{\mu}\right)$.

\subsection{Proofs from Section~\ref{sec:subroutines}}
\label{app:subroutinesproof}

\subsubsection{Proof of Lemma~\ref{subroutine: fidelity high dimension}}

Since $\tr(\langle0^{n-t}|C_i^\dagger\rho C_i|0^{n-t}\rangle)=\sum_{s\in \{0, 1\}^t}\braket{0^{n-t}s|C_i^\dagger\rho C_i|0^{n-t}s}$, we only need to estimates $\braket{\phi|\rho|\phi}$ for every $\ket{\phi}\in \{C_i\ket{0^{n-t}s}: s\in \{0, 1\}^t, i\in [M]\}$ to additive error $\epsilon/2^t$. By \Cref{lem:classical_shadow}, the sample complexity is $O(\frac{2^{2t}}{\epsilon^2}\log\frac{2^{t}M}{\delta})$ and the time complexity is $O(\frac{2^{3t}M}{\epsilon^2}n^2\log\frac{2^{t}M}{\delta})$.

\subsubsection{Proof of Lemma~\ref{lem:tomography_fidelity}}

The algorithm works as follows: measure $\frac{9}{2\epsilon^2}\log\frac{6n}{\delta}$ times in the $X$ basis for all qubits on $\rho$ to obtain an estimate $x_j$ of $\tr(X\tr_{-j}(\rho))$ for each $j\in[n]$.
Similar for the observables $Y$ and $Z$, suppose the estimates are $y_j$ and $z_j$ for $j\in[n]$.
Then for each $i\in[M]$, use $\langle\phi_i|\frac{I+x_jX+y_jY+z_jZ}{2}|\phi_i\rangle$ as the estimation of $\langle\phi_i|\tr_{-j}(\rho)|\phi_i\rangle$.

By Hoeffding's inequality, with probability at least $1-\frac{\delta}{3n}$, $|x_j-\tr(X\tr_{-j}(\rho))|\leq\frac{2\epsilon}{3}$.
The same goes for $y$ and $z$.
Hence by union bound, with probability at least $1-\delta$, $|x_j-\tr(X\tr_{-j}(\rho))|\leq\frac{2\epsilon}{3}$, $|y_j-\tr(Y\tr_{-j}(\rho))|\leq\frac{2\epsilon}{3}$ and $|z_j-\tr(Z\tr_{-j}(\rho))|\leq\frac{2\epsilon}{3}$ for all $j\in[n]$.
Conditioned on this, $\forall i\in[M]$ and $\forall j\in[n]$,
\begin{multline*}
\Bigl|\langle\phi_i|\frac{I+x_jX+y_jY+z_jZ}{2}|\phi_i\rangle-\langle\phi_i|\tr_{-j}(\rho)|\phi_i\rangle\Bigr|
\le\frac{|x_j-\tr(X\tr_{-j}(\rho))|}{2}|\langle\phi_i|X|\phi_i\rangle| \\
+\frac{|y_j-\tr(Y\tr_{-j}(\rho))|}{2}|\langle\phi_i|Y|\phi_i\rangle|+\frac{|z_j-\tr(Z\tr_{-j}(\rho))|}{2}|\langle\phi_i|Z|\phi_i\rangle|\le \epsilon\,.
\end{multline*}
The sample complexity is $\frac{27}{2\epsilon^2}\log\frac{6n}{\delta}$.
Obtaining $x_j$, $y_j$ and $z_j$ takes $O(\frac{n}{\epsilon^2}\log\frac{n}{\delta})$ time, and obtaining the fidelity estimations requires $O(Mn)$ time.
Thus the total time complexity is $O(\frac{n}{\epsilon^2}\log\frac{n}{\delta}+Mn)$.

\subsubsection{Proof of Lemma~\ref{subroutine: estimate Pauli correlation by Bell measurements}}

The algorithm works as follows: perform Bell measurement on $\rho^{\otimes2}$ for $m_{\mathsf{Bell}}=2\log(2M/\delta)/\epsilon^2$ times.
Suppose the samples are $\{|\Psi_{x_i}\rangle:i=1,\ldots,m_{\mathsf{Bell}}\}$.
Then, for $y\in S$, define $a_i$ and $b_i$ to be the first and second half of $y$, respectively.
Then $\tr(W_y\rho)^2$ is estimated as $\frac{1}{m_{\mathsf{Bell}}}\sum_{i=1}^{m_{\mathsf{Bell}}}(-1)^{\langle x_i,y\rangle+a\cdot b}$.

The correctness follows from the fact that the Bell basis is an eigenbasis of the operator $W_y^{\otimes2}$, with $|\Psi_x\rangle$ corresponding to the eigenvalue $(-1)^{\langle x,y\rangle+a\cdot b}$.
\[(W_y^{\otimes2})(W_x\otimes I)|\Omega\rangle=(W_yW_xW_y^T\otimes I)|\Omega\rangle=(-1)^{\langle x,y\rangle+a\cdot b}(W_x\otimes I)|\Omega\rangle.\]
Thus by Hoeffding's inequality, for a particular $y$, the probability of estimating $\tr(W_y\rho)^2$ to error $\geq\epsilon$ is at most $2e^{-\frac{2m_{\mathsf{Bell}}\epsilon^2}{4}}=\frac{\delta}{M}$.
Thus by union bound, the probability of the estimation for some $y\in S$ being of error $\geq\epsilon$ is at most $\delta$.
For each $y\in S$ and for each sample $x_i$, computing the inner products takes $O(n)$ time, and thus the total time complexity is $O(Mn\log(M/\delta)/\epsilon^2)$.

\subsubsection{Proof of Lemma~\ref{lem:Clifford_Synthesis}}

Consider the algorithm in Lemma 3.2 of Ref.~\cite{grewal2023efficient}.
It involves first performing Gaussian elimination on the stabilizer tableau to find a basis for $A$.
Then some further manipulation is performed to output the circuit $C$.
The algorithm runs in $O(mn\min\{m,n\})$ time and $C$ contains $O(nd)$ number of elementary gates.
The only difference in our case is that instead of guaranteeing $A$ to be isotropic, we need to decide whether this is the case.
This can be done after performing Gaussian elimination.
We check whether the $d$ basis states obtained are pairwise commuting.
This requires an additional $O(d^2n)$ time, which is dominated by $O(mn\min\{m,n\})$ since $m,2n\geq d$.

\subsubsection{Proof of Lemma~\ref{lem: sample heavy-weight subspace}}

For $0\leq i\leq m$, define $A_i\triangleq \spn(x_1,\cdots, x_i)$ ($A_0\triangleq \{0^d\}$ by convention). Define the indicator random variable $X_i$ as
\begin{equation*}
    X_i = \begin{cases}
        1\quad \text{if $x_i\in \mathbb{F}_2^d\backslash A_{i-1}$ or $\cD(A_{i-1})\ge 1-\epsilon$},\\
        0\quad \text{otherwise}.
    \end{cases}
\end{equation*}
    Observe the following facts:
\begin{enumerate}
    \item For any $x_1,\cdots, x_{i-1}$, $\Pr[X_i=1|x_1,\cdots, x_{i-1}]\ge \epsilon$. Indeed, if $\cD(A_{i-1})\ge 1-\epsilon$, then $X_i$ must be 1. Otherwise if $\cD(A_{i-1})\leq 1-\epsilon$, then the probability of $X_i=1$ is $1-\cD(A_{i-1})\ge\epsilon$.
    \item $\E[X_i]\ge \epsilon$, obvious from the first observation.
    \item Whenever $X_1+\cdots + X_m\ge d$, we have $\cD(A_m)\ge 1-\epsilon$. This is because otherwise $\cD(A_i)<1-\epsilon$ for all $i\leq m$, so $X_i=1$ implies $x_i\in\mathbb{F}_2^d\backslash A_{i_1}$. There are at least $d$ such $i$, so $x_1,\cdots, x_m$ must span the whole space $\mathbb{F}_2^n$, contradicting the assumption that $\cD(A_m)< 1-\epsilon$.
\end{enumerate}
Let $Y_i$ be the event that $\cD(A_i)\ge 1-\epsilon$. (a) follows from the direct calculation
\begin{equation}
    \Pr[Y_d]\ge \Pr[X_1=1,X_2=1,\cdots, X_d=1]\ge \epsilon^d,
\end{equation}
where the two inequalities are from the third observation and the first observation, respectively.

(b) is from the concentration inequality. A caveat is that the $X_i$'s are not independent. To address the issue, consider $X_i'$ ($i=1,2,\cdots, m$) as $m$ i.i.d. samples from a Bernoulli distribution with $\Pr[X_i'=1]=\epsilon$. Then $\Pr[X_i'=1|X_1',\cdots, X_{i-1}']=\epsilon\leq \Pr[X_i=1|X_1,\cdots, X_{i-1}]$ from the first observation. Then it's easy to see that $\Pr[X_1+\cdots +X_m<d]\leq \Pr[X_1'+\cdots +X_m'<d]$. Let $\Upsilon=1-\frac{d}{m\epsilon}$. By the third observation and the Chernoff bound,
\begin{align*}
    \Pr[Y_m]&\ge \Pr[X_1+X_2+\cdots+X_m\ge d]\\
    &\ge 1-\Pr[X_1'+\cdots +X_m'<d]\\
    &\ge 1-\Pr[X_1'+\cdots + X_m'<(1-\Upsilon)m\epsilon]\\
    &=1- \mathrm{exp}(-\frac12 \Upsilon^2m\epsilon)\\
    &\ge 1- \mathrm{exp}(-\frac{m\epsilon}{2}+d)\ge 1-\delta\,.
\end{align*}

\subsection{Proofs from Section~\ref{sec:stabilizer}}

\subsubsection{Proof of Corollary~\ref{cor:agnostic_learning_stabilizer}}
\label{sec:defer_agnostic_stabilizer_cor}

Suppose the stabilizer that maximizes fidelity with $\rho$ is $|\phi_0\rangle$  (breaking ties arbitrarily), that is, $\langle\phi_0|\rho|\phi_0\rangle=F_\cS(\rho)$.
Then $|\phi_0\rangle$ is an $1$-approximate local maximizer of fidelity with $\rho$, with fidelity at least $\tau$.

We can run the algorithm given by~\Cref{cor:stab_list} for $\gamma = 1$ to obtain a list of stabilizer states of length at most $M \triangleq O(\log(1/\delta))\cdot ((\gamma-1/2)\tau)^{-O(\log\frac{1}{\tau})}$ which contains $\ket{\phi_0}$ with probability at least $1 - \delta/2$. 
Then we may use classical shadows (\Cref{lem:classical_shadow}) to estimate the fidelities of the stabilizers in the list so that with probability at least $1-\frac{\delta}{2}$, every fidelity is estimated to error at most $\frac{\epsilon}{2}$.
The output $|\phi\rangle$ is chosen to be the one with the highest estimated fidelity.
Conditioned on $|\phi_0\rangle$ appearing in the list and the fidelities all being estimated to error at most $\frac{\epsilon}{2}$, the fidelity of $|\phi\rangle$ is overestimated by at most $\frac{\epsilon}{2}$ while the fidelity of $|\phi_0\rangle$ is underestimated by at most $\frac{\epsilon}{2}$, and thus $|\phi\rangle$ has fidelity at least $F_{\cS}(\rho)-\epsilon$.
By union bound, we conclude that $F(\rho,|\phi\rangle)\geq F_{\cS}(\rho)-\epsilon$ with probability at least $1-\delta$.

The sample complexity and runtime are given by that of the algorithm in~\Cref{cor:stab_list} plus that of the classical shadows protocol. The former takes $O(n\log(1/\delta))\cdot (1/\tau)^{O(\log 1/\tau)}$ copies and $O(n^3\log(1/\delta)\cdot (1/\tau)^{O(\log 1/\tau)}$ time. The latter takes $O((\log^2(1/\tau) + \log(1/\delta))/\epsilon^2)$ copies and $O(n^2\log^2(1/\delta)/\epsilon^2)\cdot (1/\tau)^{O(\log 1/\tau)}$ time.
Summing these yields the claimed sample complexity and runtime bounds.

\subsubsection{Proof of Corollary~\ref{cor:fid_est}}
\label{sec:defer_fid_est}

Suppose the stabilizer that maximizes fidelity with $\rho$ is $|\phi_0\rangle$ and hence $F_{\cS}(\rho)=\langle\phi_0|\rho|\phi_0\rangle$.
The algorithm is similar to the one in~\Cref{cor:agnostic_learning_stabilizer}.
First run the algorithm from~\Cref{cor:stab_list} with $\tau=\epsilon$ to output a list which is guaranteed to contain $\ket{\phi_0}$ with probability $1 - \delta/2$. Then use the classical shadows (\Cref{lem:classical_shadow}) to estimate the fidelity of all returned stabilizer states to within error at most $\epsilon$, with probability at least $1-\delta/2$. Finally, return the largest estimate, call it $\tau_{\mathrm{est}}$.

If $F_{\cS}(\rho)\geq\epsilon$, then with probability at least $1-\frac{\delta}{2}$, $|\phi_0\rangle$ is one of the stabilizers in the list. Conditioned on the fidelity estimates being accurate, we have $F_{\cS}(\rho)-\epsilon\leq\tau_{\mathrm{est}}\leq F_{\cS}(\rho)+\epsilon$.
Thus if $F_{\cS}(\rho)\geq\epsilon$, the returned estimate $\tau_{\mathrm{est}}$ is within error $\epsilon$ of the true stabilizer fidelity with probability at least $1-\delta$.

If $F_{\cS}(\rho)<\epsilon$, then conditioned on the fidelity estimates being accurate, we have $F_{\cS}(\rho)-\epsilon<0\leq\tau_{\mathrm{est}}\leq F_{\cS}(\rho)+\epsilon$. So in this case, the returned estimation is still within error $\epsilon$ with probability at least $1-\frac{\delta}{2}$.

\subsection{Proofs from Section~\ref{sec:dimension}}

\subsubsection{Proof of Lemma~\ref{lem: full tomography given C}}
\label{sec:defer_tomo_C}
Take $N=2^{O(t)}\log(1/\delta)/\epsilon^2$.
Apply $C$ to $\rho$ and measure the first $n-t$ qubits on the computational basis. With probability at least $\tr(\braket{0^{n-t}|C\rho C^\dagger|0^{n-t}})\ge \tau$, the outcome is $0^{n-t}$ and the post-measurement state is $\rho^{C}_{n-t}$. Repeat the process $2N/\tau$ times, by Chernoff bound, with probability at least $1-e^{-N/4}\geq1-\delta/2$, we will obtain $N$ copies of $\rho^{C}_{n-t}$. Applying full tomography on $\rho^{C}_{n-t}$ (\Cref{lem: full tomography} with $\epsilon$ set to $\epsilon/2$, $\delta$ set to $\delta/2$), we obtain a density matrix $\sigma_0$ such that $D_{\tr}(\rho_{n-t}^C, \sigma_0)\leq \epsilon/2$ with probability at least $1-\delta/2$. Then by the well-known relation between fidelity and trace distance (see, e.g., Ref.~\cite{nielsen2010quantum}), $F(\rho_{n-t}^C, \sigma_0)\ge (1-D_{\tr}(\rho_{n-t}^C, \sigma_0))^2\ge 1-\epsilon$.
The overall success probability is thus at least $1-\delta$.

The sample complexity is $O(N/\tau)=2^{O(t)}\log(1/\delta)/\epsilon^2\tau$ and the time complexity is $O(n^2N/\tau+2^{O(t)}\log(1/\delta)/\epsilon^2)=2^{O(t)}n^2\log(1/\delta)/\epsilon^2\tau$, where $n^2$ comes from the implementation of $C$.

\subsubsection{Proof of Lemma~\ref{lem: reduce to find Clifford}}
\label{sec:defer_reduce_cliff}
(a) Run $\mathcal{A}$ for $M=\log(2/\delta)/p$ times, obtaining a sequence of Clifford unitaries $C_1,\cdots, C_M$. With probability at least $1-(1-p)^{M}\ge 1-e^{-pM}=1-\delta/2$, there exists $C_i$ such that $\tr(\braket{0^{n-t}|C_i\rho C_i^\dagger|0^{n-t}})\ge F(\rho, \cS^{n-t})-\epsilon/3$. Apply \Cref{subroutine: fidelity high dimension} to estimate $\tr(\braket{0^{n-t}|C_i\rho C_i^\dagger|0^{n-t}})$ within error $\epsilon/6$ with probability at least $1-\delta/2$ and output the $C_i$ with the largest estimated value. With an overall probability at least $1-\delta$, the output $C$ satisfies $\tr(\braket{0^{n-t}|C\rho C^\dagger|0^{n-t}})\ge F(\rho, \cS^{n-t})-2\epsilon/3$. The sample complexity is $O(MS+\frac{2^{2t}}{\epsilon^2}\log\frac{2^t M}{\delta})=O(\frac{\log(1/\delta)S}{p}+\frac{2^{O(t)}}{\epsilon^2}\log\frac{1}{p\delta})$. The time complexity is $O(MT+\frac{2^{3t}Mn^2}{\epsilon^2}\log\frac{2^t M}{\delta})=O(\frac{\log(1/\delta)T}{p}+\frac{2^{O(t)}n^2\log(1/\delta)}{\epsilon^2p}\log\frac{1}{p\delta})$.

(b) By (a) (with $\delta$ set to $\delta/2$), we can find a Clifford gate $C$ such that $\tr(\braket{0^{n-t}|C\rho C^\dagger|0^{n-t}})\ge F(\rho, \cS^{n-t})-2\epsilon/3\ge \tau/3$ with probability at least $1-\delta/2$. Apply \Cref{lem: full tomography given C} to obtain the density matrix of a $t$-qubit state $\sigma_0$ such that $F(\rho_{n-t}^C, \sigma_0)\ge 1-\epsilon/3$ with probability at least $1-\delta/2$. With an overall probability at least $1-\delta$, both things happen. Then \Cref{lem: fidelity high dimension} implies
\begin{align*}
F(\rho, C^\dagger(\ketbra{0^{n-t}}\otimes \sigma_0)C)&=\tr(\braket{0^{n-t}|C\rho C^\dagger|0^{n-t}})F(\rho_{n-t}^C, \sigma_0)\\
&\ge (F(\rho, \cS^{n-t})-2\epsilon/3)(1-\epsilon/3)\\
&\ge F(\rho, \cS^{n-t})-\epsilon.
\end{align*}
The sample complexity is $O(\frac{\log(1/\delta)S}{p}+\frac{2^{O(t)}}{\epsilon^2}\log\frac{1}{p\delta}+\frac{2^{O(t)}}{\epsilon^2\tau}\log\frac{1}{\delta})$. The time complexity is $O(\frac{\log(1/\delta)T}{p}+\frac{2^{O(t)}n^2\log(1/\delta)}{\epsilon^2p}\log\frac{1}{p\delta}+\frac{2^{O(t)}n^2}{\epsilon^2\tau}\log\frac{1}{\delta})$.

\subsubsection{Proof of Lemma \ref{lem: high stabilizer exponential time}}\label{sec: few qubits case high dimension}
Here we prove \Cref{lem: high stabilizer exponential time} by constructing an agnostic tomography algorithm for $\cS^{n-t}$ with exponential sample and time complexity. This will be an ingredient of the efficient algorithm (\Cref{thm:agnostic_learning_high_stabilizer_dimension_states}). The workflow is similar to \Cref{thm:agnostic_learning_high_stabilizer_dimension_states}. To avoid the inadequacy of Bell difference sampling (as in \Cref{lem:weaker_evenly_distribution}), here we uniformly sample the Pauli string in Step 3. In addition, with an exponential budget, we can afford to select all high-correlation Pauli strings in Step 1, see \Cref{lem: high dim select all high correlation Pauli strings}. We directly write down the algorithm in \Cref{alg:agnostic_learning_states_high_stabilizer_dimension_weaker}.

\begin{lemma}\label{lem: high dim select all high correlation Pauli strings}
    Given copies of an $n$-qubit state $\rho$, there exists a algorithm that outputs a basis $H$ of a stabilizer family such that with probability at least $2/3$, $\spn(H)$ contains all Pauli strings $P$ with $\tr(P\rho)^2\ge 0.7$. The algorithm uses $1600\log(6\times 2^{2n})$ copies and $2^{O(n)}$ time.
\end{lemma}
\begin{proof}
    By \Cref{subroutine: estimate Pauli correlation by Bell measurements}, we can estimate $\tr(P\rho)^2$ for all $n$-qubit Pauli strings $P$ within error $0.05$ with probability at least $2/3$ via Bell measurements using $1600\log(6\times 2^{2n})$ samples. Denote the estimate of $\tr(P\rho)^2$ by $\hat{E}_P$. Let $S$ be the set of Pauli strings $P$ with $\hat{E}_P\ge 0.6$. It is easy to see that $S$ contains all Pauli strings $P$ with $\tr(P\rho)^2\ge 0.7$. Furthermore, every $P$ in $S$ has $\tr(P\rho)^2>0.5$, so $S$ is commuting by \Cref{lem:high-corr-commute}. Therefore, we can find a stabilizer family that contains $S$. The algorithm outputs a basis of the stabilizer family. The runtime is $2^{O(n)}$.
\end{proof}

\begin{algorithm}[htbp]
    \DontPrintSemicolon
    \caption{agnostic tomography of states with high stabilizer dimension, weaker version}\label{alg:agnostic_learning_states_high_stabilizer_dimension_weaker}
    \KwInput{$t\in \mathbb{N}, \tau> \epsilon>0$, copies of an $n$-qubit state $\rho$}
    \Promise{$F(\rho, \cS^{n-t})\ge \tau$}
    \KwOutput{A Clifford gate $C$.}
    \Goal{With probability at least $(6\times 2^{2n})^{-k_{\max}}\epsilon^{t+1}(\tau-\epsilon)/(k_{\max}+1)$, $\tr(\braket{0^{n-t}|C\rho C^\dagger|0^{n-t}})\ge F(\rho, \cS^{n-t})-\epsilon$, where $k_{\max}=\left\lfloor\log_{1.08}\frac{1}{\tau}\right\rfloor+1$.}
    Set $\mathfrak{R}=\emptyset$, $k_{\max}=\left\lfloor\log_{1.08}\frac{1}{\tau}\right\rfloor+1$, $C_0=I^{\otimes n}$.\\
    \For{$k=0$ \KwTo $k_{\max}$}{
    Define $\tau_k=1.08^{k}\tau, \epsilon_k=1.08^k \epsilon, \rho_k=\braket{0^k|C_k \rho C_k^\dagger|0^k}/\tr(\braket{0^k|C_k \rho C_k^\dagger|0^k}), m_k=1600\log(6\times 2^{2(n-k)})$.\\
    Use $\frac{2}{\tau}(m_k+t+2+\log(\frac{3}{2}))$ copies of $\rho$ to prepare $\rho_k$. Break the loop if the number of $\rho_k$ is less than $m_k+t+2$.\label{line: weaker line 4}\\
    Run \Cref{lem: high dim select all high correlation Pauli strings} using $m_k$ copies of $\rho_k$. Denote the result by $H_k$.\label{line: weaker line 5}\\
    Run \Cref{lem:find_heavy_subspace_S} on $\rho_k$ and $H_k$ (with $(n, t', t, \tau, \epsilon)$ set to $(n-k, t, t, \tau_k, \epsilon_k)$. The sample complexity is $t+2$.). The output is an $(n-k)$-qubit Clifford gate $U_k$. Define $R_k=(I^{\otimes k}\otimes U_k)C_k$. Add $R_k$ to $\mathfrak{R}$.\label{line: weaker line 6}\\
    Uniformly randomly pick a $(n-k)$-qubit Pauli string $Q_k$ and a sign $\sgn_k\in \{0, 1\}$.\label{line: weaker line 7}\\
    Find a $(n-k)$-qubit Clifford gate $V_k$ such that $V_k \sgn_k Q_k V_k^\dagger = Z_1$.\\
    Define $C_{k+1}=(I^{\otimes k}\otimes V_k)C_k$.\\
    }
    \Return a uniformly random element from $\mathfrak{R}$. If $\mathfrak{R}=\emptyset$, return $I^{\otimes n}$.
\end{algorithm}

Fix a state $\sigma^*\in \cS^{n-t}$ such that $F(\rho, \sigma^*)=F(\rho, \cS^{n-t})$. Define $\sigma^*_k=\braket{0^k|C_k\sigma^* C_k^\dagger|0^k}/\tr(\braket{0^k|C_k\sigma^* C_k^\dagger|0^k})$. To described the expected behavior of \Cref{alg:agnostic_learning_states_high_stabilizer_dimension_weaker}, we define the following events for each $0\leq k\leq k_{\max}$:
\begin{enumerate}
    \item We say the algorithm correctly proceeds to iteration $k$ if it does not break the loop before iteration $k$ and $\tr(\rho_r Q_r)^2\leq 0.7$ and $Q_r\sigma_r^*=\sgn_r \sigma_r^*$ for every $0\leq r\leq k$. Denote the event by $A_k$.
    \item We say the algorithm succeeds at iteration $k$ if it correctly proceeds to iteration $k-1$, does not break the loop at iteration $k$, and $\dim(\spn(H_k)\cap \Weyl(\sigma^*_k))\ge n-k-t$. Denote the event by $B_k$.
    \item Define $E_k=B_0\vee B_1\cdots \vee B_k$.
\end{enumerate}
For convenience, we define $A_{-1}$, $B_{-1}$, and $E_{-1}$ to be the whole probability space. We prove the following facts:
\begin{lemma}\label{lem: high stabilizer exponential time events}
    For the events defined above, we have:
    \begin{enumerate}[label=(\alph*)]
        \item If $A_{k-1}$ happens, then $C_k\sigma^* C_k^\dagger=\ketbra{0^{k}}\otimes \sigma_k^*$ and $F(\rho_k, \sigma^*_k)= F(\rho_k, \cS_{n-k}^{n-k-t})\ge 1.08^k F(\rho, \cS_{n}^{n-t})\ge 1.08^k\tau$.
        \item $\Pr[A_k\vee E_k|A_{k-1}\vee E_{k-1}]\ge \frac{1}{6\times 2^{2n}},~\forall 0\leq k\leq k_{\max}$.
        \item $\Pr[A_{k_{\max}-1}]=0$.
        \item $\Pr[E_{k_{\max}-1}]\ge (6\times 2^{2n})^{-k_{\max}}$.
        \item If $E_{k_{\max}-1}=B_0\vee B_1\cdots \vee B_{k_{\max}-1}$ happens, with probability at least $\epsilon^{t+1}(\tau-\epsilon)/(k_{\max}+1)$, the output of the algorithm is a Clifford gate $C$ such that $\tr(\braket{0^{n-t}|C\rho C^\dagger|0^{n-t}})\ge F(\rho, \sigma^*)-\epsilon$.
        \item The output of the algorithm is a Clifford gate $C$ such that $\tr(\braket{0^{n-t}|C\rho C^\dagger|0^{n-t}})\ge F(\rho, \sigma^*)-\epsilon$ with probability at least $(6\times 2^{2n})^{-k_{\max}}\epsilon^{t+1}(\tau-\epsilon)/(k_{\max}+1)$.
    \end{enumerate}
\end{lemma}
\begin{proof}
    (a) If $A_{k-1}$ happens, we prove by induction that for every $0\leq r\leq k$,
    \begin{equation}\label{eq: high stabilizer exponential time events 1}
        C_r\sigma^* C_r^\dagger=\ketbra{0^r}\otimes \sigma_r^*,~F(\rho_r, \sigma^*_r)=F(\rho_r, \cS_{n-r}^{n-r-t})\ge 1.08^r F(\rho, \cS_{n}^{n-t})\ge 1.08^r\tau=\tau_{r}.
    \end{equation}
    When $r=0$, \eqref{eq: high stabilizer exponential time events 1} is trivial. Assume \eqref{eq: high stabilizer exponential time events 1} is true for $r-1$. By definition of $\rho_{r-1}, \rho_r$, we have
    \begin{align*}
        \rho_r&\propto \braket{0^r|C_r\rho C_r^\dagger|0^r}=\braket{0^r|(I^{\otimes r-1}\otimes V_{r-1})C_{r-1}\rho C_{r-1}^\dagger (I^{\otimes r-1}\otimes V_{r-1}^\dagger)|0^k}\\
        &\propto \braket{0|V_{k-1}\rho_{r-1}V_{k-1}^\dagger|0}.
    \end{align*}
    Similarly, $\sigma^*_r\propto \braket{0|V_{k-1}\sigma^*_{r-1}V_{k-1}^\dagger|0}$.
    So $\rho_r=\braket{0|V_{k-1}\rho_{r-1}V_{k-1}^\dagger|0}/\tr(\braket{0|V_{k-1}\rho_{r-1}V_{k-1}^\dagger|0})$. By the induction hypothesis, $F(\rho_{r-1}, \sigma^*_{r-1})= \tau_{r-1}$. Since $A_{r-1}$ happens, $V_{r-1}$ satisfies $V_{r-1}^\dagger Z_1 V_{r-1}\sigma^*_{r-1}=\sgn_{r-1}Q_{r-1}\sigma^*_{r-1}=\sigma^*_{r-1}$ and $\tr(V_{r-1}^\dagger Z_1 V_{r-1}\rho_{r-1})^2=\tr(Q_{r-1}\rho_{r-1})^2\leq 0.7$. \eqref{eq: high stabilizer exponential time events 1} follows from \Cref{lem: high dim Step 4} (where $(\rho, \sigma^*, C)\leftarrow (\rho_{r-1}, \sigma^*_{r-1}, V_{r-1})$), $F(\rho_r, \sigma^*_r)=F(\rho_r, \cS_{n-r}^{n-r-t})\ge 1.08 F(\rho_{r-1}, \cS_{n-(r-1)}^{n-(r-1)-t})$.

    (b) Conditioned on $A_{k-1}\vee C_{k-1}$ happens, the goal is to lower bound the probability of $A_k\vee C_k$. If $C_{k-1}$ happens, then $A_k\vee C_k$ happens for sure. Now assume $C_{k-1}$ does not happen and $A_{k-1}$ happens. We go through the iteration $k$. By \Cref{lem: fidelity high dimension} and (a), \begin{equation*}
        \tr(\braket{0^k|C_k\rho C_k^\dagger|0^k})\ge F(C_k\rho C_k^\dagger, \ketbra{0^k}\otimes \sigma_k^*)=F(\rho, \sigma^*)\ge \tau.
    \end{equation*}
    So the state preparation (line \ref{line: weaker line 4}) succeeds with probability at least $2/3$ according to \Cref{lem: prepare state high dimension}. By \Cref{subroutine: estimate Pauli correlation by Bell measurements}, line \ref{line: weaker line 5} succeeds with probability at least $2/3$. From now on we condition on the success of both lines, which happens with probability at least $1/3$.

    $\spn(H_k)$ contains all high-correlation Pauli strings. If $\dim(\spn(H_k)\cap \Weyl(\sigma^*_k))\ge n-k-t$, $B_k$ happens for sure by definition. Otherwise if $\dim(\spn(H_k)\cap \Weyl(\sigma^*_k))< n-k-t\leq \Weyl(\sigma^*_k)$, there exists a low-correlation Pauli string $Q$ in $\Weyl(\sigma_k^*)$. In this case, the random guess in Line \ref{line: weaker line 7} will find $Q_k=Q$ with the correct sign with probability at least $\frac{1}{2\times 2^{2n}}$. In other words, $A_k$ happens with probability at least $\frac{1}{2\times 2^{2n}}$ when $\dim(\spn(H_k)\cap \Weyl(\sigma^*_k))< n-k-t$. In both cases, $A_k\vee C_k$ happens with probability at least $\frac{1}{6\times 2^{2n}}$. Hence, $\Pr[A_k\vee C_k|A_{k-1}\vee C_{k-1}]\ge \frac{1}{6\times 2^{2n}}$.

    (c) If $A_{k_{\max}-1}$ happens, by (a), $F(\rho_{k_{\max}}, \sigma^*_{k_{\max}})\ge 1.08^{k_{\max}}\tau>1$, a contradiction.

    (d) By (b), (c),
    \begin{equation*}
        \frac{1}{6\times 2^{2n}}\leq \Pr[A_k\vee E_{k}|A_{k-1}\vee E_{k-1}]=\frac{\Pr[A_k\vee E_k, A_{k-1}\vee E_{k-1}]}{\Pr[A_{k-1}\vee E_{k-1}]}\leq \frac{\Pr[A_k\vee E_k]}{\Pr[A_{k-1}\vee E_{k-1}]}.
    \end{equation*}
    So $Pr[A_k\vee E_k]\ge (6\times 2^{2n})^{-k-1}$. In particular, by (c), $\Pr[E_{k_{\max}-1}]=\Pr[A_{k_{\max}-1}\vee E_{k_{\max}-1}]\ge (6\times 2^{2n})^{-k_{\max}}$.

    (e) Suppose $B_{k}$ happens. Write $H=H_k$ for simplicity. 
    By definition of $B_{k}$, $\dim(\spn(H_k)\cap \Weyl(\sigma_k^*))\ge n-k-t$, i.e., $\sigma_k^*\in H_{n-k-t}^{n-k-t}\subseteq \cS_{n-k}^{n-k-t}$. We have shown in (a) that $\sigma_k^*$ is the closest state to $\rho_k$ in $\cS_{n-k}^{n-k-t}$. So $F(\rho_k, H_{n-t}^{n-t})=F(\rho_k, \cS_{n-k}^{n-k-t})=F(\rho_k, \sigma_k^*)\ge \tau_t$. By \Cref{lem:find_heavy_subspace_S}, with probability at least $\epsilon_t^{t+1}(\tau_t-\epsilon_t)$, the output of line \ref{line: weaker line 6} is a Clifford gate $U_k$ such that $\tr(\braket{0^{n-k-t}|U_k\rho_k U_{k}^\dagger|0^{n-k-t}})\ge F(\rho_k, H_{n-t}^{n-t})-\epsilon_k= F(\rho_k, \sigma_k^*)-\epsilon_k$ (here we take $t'=t$ in \Cref{lem:find_heavy_subspace_S}, so the $F(\rho_{n-t'}^C, \cS_{t'}^{t'-t}))$ is simply 1). Therefore,
    \begin{align*}
        \tr(\braket{0^{n-t}|R_k\rho R_k^\dagger|0^{n-t}})&=\tr(\braket{0^{n-t}|(I^{\otimes k}\otimes U_k)C_k\rho C_{k}^\dagger (I^{\otimes k}\otimes U_k^\dagger)|0^{n-t}})\\
        &=\tr(\braket{0^{n-k-t}|U_k\rho_k U_k^\dagger|0^{n-k-t}})\tr(\braket{0^k|C_k\rho C_k^\dagger|0^k})\\
        &\ge (F(\rho_k, \sigma_k^*)-\epsilon_k)\tr(\braket{0^k|C_k\rho C_k^\dagger|0^k})\\
        &= (F(\rho_k, \sigma_k^*)-\epsilon_k)\frac{F(\rho, \sigma^*)}{F(\rho_k, \sigma_k^*)}\\
        &\ge F(\rho, \sigma^*)-\epsilon.
    \end{align*}
    In the fourth line, we use $F(\rho, \sigma^*)=F(C_k\rho C_k^\dagger, \ketbra{0^k}\otimes \sigma_k^*)=F(\rho_k, \sigma^*_k)\tr(\braket{0^k|C_k\rho C_k^\dagger|0^k})$. In the last line, we use $F(\rho, \sigma^*)/F(\rho_k, \sigma_k^*)\leq 1/1.08^{k}$ from (a) (recall that $B_k$ implies $A_{k-1}$ by definition). 

    We have proved that if $B_k$ happens, $R_k$ is a desired output with probability at least $\epsilon_k^{t+1}(\tau_k-\epsilon_k)\ge \epsilon^{t+1}(\tau-\epsilon)$. 
    Since $\abs{\mathfrak{P}}\leq k_{\max}+1$, the algorithm returns a desired output with probability at least $\epsilon^{t+1}(\tau-\epsilon)/(k_{\max}+1)$.

    (f) This is straightforward from (d), (e).
\end{proof}

The (f) in \Cref{lem: high stabilizer exponential time events} indeed establishes the correctness of \Cref{alg:agnostic_learning_states_high_stabilizer_dimension_weaker}. To prove \Cref{lem: high stabilizer exponential time}, we only need to amplify the success probability to $1-\delta$ by repeating the algorithm as in \Cref{lem: reduce to find Clifford}.

\begin{proof}[Proof of \Cref{lem: high stabilizer exponential time}]
    Replace the $\epsilon$ in \Cref{alg:agnostic_learning_states_high_stabilizer_dimension_weaker} with $\epsilon/2$. By (f) in \Cref{lem: high stabilizer exponential time events}, the algorithm outputs a Clifford gate $C$ such that $\tr(\braket{0^{n-t}|C\rho C^\dagger|0^{n-t}})\ge F(\rho, \cS_n^{n-t})-\epsilon/2$ with probability at least $(6\times 2^{2n})^{-k_{\max}}\epsilon^{t+1}(\tau-\epsilon/2)/(k_{\max}+1)\ge (6\times 2^{2n})^{-k_{\max}}\epsilon^{t+1}(\tau/2)/(k_{\max}+1)$. We now count the sample complexity $S$ and the time complexity $T$.

    At iteration $k$, we use $\frac{2}{\tau}(m_k+t+2+\log(\frac{3}{2}))=O(\frac{n}{\tau})$ copies of $\rho$ to prepare $\rho_k$. There are $k_{\max}+1=O(\log(\frac{2}{\tau}))$ iterations in each algorithm. So $S=O(\frac{n}{\tau}\log(\frac{2}{\tau}))$.

    At iteration $k$, state preparation (line \ref{line: weaker line 4}) takes $O(\frac{n^3}{\tau})$ time, selecting high-correlation Pauli strings (line \ref{line: weaker line 5}) takes $2^{O(n)}$ time according to \Cref{lem: high dim select all high correlation Pauli strings}, and other lines take $O(n^3)$ time. There are $(k_{\max}+1)$ such iterations. So $T=2^{O(n)}\frac{1}{\tau}\log(\frac{2}{\tau})$.

    In summary, \Cref{alg:agnostic_learning_states_high_stabilizer_dimension_weaker} outputs a Clifford unitary $C$ such that $\tr(\braket{0^{n-t}|C\rho C^\dagger|0^{n-t}})\ge F(\rho, \cS_n^{n-t})-\epsilon/2$ with probability at least $p=(6\times 2^{2n})^{-k_{\max}}\epsilon^{t+1}(\tau/2)/(k_{\max}+1)$ using $S=O(\frac{n}{\tau}\log(\frac{2}{\tau}))$ samples and $T=2^{O(n)}\frac{1}{\tau}\log(\frac{2}{\tau})$ time. With the same repeating argument as \Cref{lem: reduce to find Clifford}(a), we can amplify the success probability to $1-\delta$. This argument introduces another error of $\epsilon/2$, so the overall error is $\epsilon$. The sample complexity is $$O\left(\frac{S}{p}\log\frac{1}{\delta}+\frac{2^{O(t)}}{\epsilon^2}\log\frac{1}{p\delta}\right)=\left(\frac{2}{\tau}\right)^{O(n)}\left(\frac1\epsilon\right)^{O(t)}\log\frac1\delta$$ and the time complexity is 
    \begin{align*}O\left(\frac{T}{p}\log\frac{1}{\delta}+\frac{2^{O(t)}n^2\log(1/\delta)}{\epsilon^2p}\log\frac{1}{p\delta}\right)&=\left(\frac{2}{\tau}\right)^{O(n)}\left(\frac1\epsilon\right)^{O(t)}\log^2\frac1\delta\,.\qedhere
    \end{align*}
\end{proof}

\subsubsection{Proof of Lemma \ref{lem: principle of inclusion-exclusion}}\label{sec: proof of inclusion-exclusion}
Let $P_2=\Weyl(\sigma_2)$. We first prove that, there exists a stabilizer group $P$ such that $P_1\subseteq P$ and $\dim(P)\ge \dim(P_2)$. Let $A_1=P_1\cap P_2$ and write $P_1=A_1\oplus A_2, P_2=A_1\oplus B$ for some subspaces $A_2\subseteq P_1, B\subseteq P_2$. Let $A_3=A_2^\perp\cap B$, where $A_2^\perp$ is the space of Pauli strings that are commuting with every element in $A_2$. By definition, $A_1, A_2, A_3$ are disjoint (i.e., the intersection of any two of them is the zero subspace) and commuting (indeed, $A_1,A_2$ are commuting because they are in the same stabilizer group $P_1$. $A_1, A_3$ are commuting because they are in $P_2$. $A_2, A_3$ are commuting because $A_3\subseteq A_2^\perp$). Hence $P=A_1\oplus A_2\oplus A_3$ is a stabilizer group with dimension $\dim(A_1)+\dim(A_2)+\dim(A_3)$. 
By the principle of inclusion-exclusion, 
\begin{equation*}
    \dim(A_3)\ge \dim(A_2^\perp)+\dim(B)-2n=\dim(B)-\dim(A_2).
\end{equation*}
So $\dim(P)=\dim(A_1)+\dim(A_2)+\dim(A_3)\ge \dim(A_1)+\dim(B)=\dim(P_2)$. In other words, $P$ is a stabilizer group that contains $P_1$ and has dimension at least $\dim(P_2)$.
Hence, $\Stab(P)\subseteq \Stab(P_1)\cap \cS^{\dim(P_2)}\subseteq \Stab(P_1)\cap \cS^{n-t}$. 

Denote $a_i\triangleq \dim(A_i)$. Divide $n$ qubit into $I_1=\{1, 2\cdots, a_1\}$, $I_2=\{a_1+1,\cdots, a_1+a_2\}$, $I_3=\{a_1+a_2+1,\cdots, a_1+a_2+a_3\}$, and $I_0$ for the rest. There exists a Clifford gate $C$ that maps $A_i$ to $\{I, Z\}^{\otimes I_i}$ for every $i=1,2,3$ (here we omit the identity on other regions). 
Then $C\sigma_1 C^\dagger$ has the form $\ketbra{s_1s_2}\otimes \sigma_1'$ for some $s_1\in \{0, 1\}^{a_1}, s_2\in \{0, 1\}^{a_2}$. By \Cref{lem: fidelity high dimension}, $\tr(\braket{s_1s_2|C\rho C^\dagger|s_1s_2})\ge F(\rho, \sigma_1)$. Similarly, there exists a state $\ket{s_2's_3}$ on qubits $I_2\cup I_3$ such that $\tr(\braket{s_2's_3|C\rho C^\dagger|s_2's_3})\ge F(\rho, \sigma_2)$. Let $J_i=\{s\in \{0, 1\}^{n}: s^{I_i}=s_i\}$ be the set of bit-strings that agree with $s_i$ on $I_i$. 
By the principle of inclusion-exclusion,
\begin{align*}
    \tr(\braket{s_1s_2s_3|C\rho C^\dagger|s_1s_2s_3})&=\sum_{s\in J_1\cap J_2\cap J_3}\braket{s|C\rho C^\dagger|s}\\
    &\ge \sum_{s\in J_1\cap J_2}\braket{s|\rho|s}+\sum_{s\in J_3}\braket{s|C\rho C^\dagger|s}-1\\
    &= \tr(\braket{s_1s_2|C\rho C^\dagger|s_1s_2})+\tr(\braket{s_3|C\rho C^\dagger|s_3}_{I_3})-1\\
    &\ge \tr(\braket{s_1s_2|C\rho C^\dagger|s_1s_2})+\tr(\braket{s_2's_3|C\rho C^\dagger|s_2's_3})-1\\
    &\ge F(\rho, \sigma_1)+F(\rho, \sigma_2)-1.
\end{align*}
By \Cref{lem: fidelity high dimension}, there exists a state $\sigma_0$ on $I_0$ ($\sigma_0$ is just $\braket{s_1s_2s_3|C\rho C^\dagger|s_1s_2s_3}$ with normalization) such that $F(\rho, C^\dagger(\ketbra{s_1s_2s_3}\otimes \sigma_0)C)\ge F(\rho, \sigma_1)+F(\rho, \sigma_2)-1$. $C^\dagger(\ketbra{s_1s_2s_3}\otimes \sigma_0)C$ is a state in $\Stab(P)\subseteq \Stab(P_1)\cap \cS^{n-t}$. The lemma follows.

\subsection{Proof of Corollary~\ref{cor:product_base}}
\label{sec:defer_product_cor}

Suppose $|\phi_0\rangle\in\argmax_{|\varphi\rangle\in\mathcal{K}^{\otimes n}}F(|\varphi\rangle,\rho)$ (breaking ties arbitrarily).
That is, $F(\rho,|\phi_0\rangle)=F_{\mathcal{K}^{\otimes n}}(\rho)\geq\tau$.

We can run the algorithm given by Corollary~\ref{cor:prod_list} to obtain a list of states from $\mathcal{K}^{\otimes n}$ of length at most $M \triangleq \log(2/\delta)\cdot (n|\mathcal{K}|)^{O(\log(1/\tau)/\mu)}$ with probability at least $1 - \delta/2$. For each of the $\ket{\phi} = \bigotimes_{j=1}^n\ket{\phi^j}$ in the list, measure in the basis $\bigotimes_{j=1}^n\{|\phi^j\rangle\!\langle\phi^j|,I-|\phi^j\rangle\!\langle\phi^j|\}$ for $\frac{2}{\epsilon^2}\log(4M/\delta)$ times to estimate the fidelity $F(|\phi\rangle,\rho)$.
By Hoeffding's inequality, with probability at most $\delta/2M$, the estimation has an error greater than $\frac{\epsilon}{2}$.
Thus by union bound, with probability at least $1-\frac{\delta}{2}$, every fidelity is estimated to error within $\frac{\epsilon}{2}$.
The state $|\phi\rangle$ with the highest estimated fidelity is returned.
Conditioned on $|\phi_0\rangle$ appears in the list and the fidelities are estimated to error at most $\frac{\epsilon}{2}$, the fidelity of $|\phi\rangle$ is overestimated by at most $\frac{\epsilon}{2}$ while the fidelity of $|\phi_0\rangle$ is underestimated by at most $\frac{\epsilon}{2}$, and thus $|\phi\rangle$ has fidelity at least $F_{\cS}(\rho)-\epsilon$.
By union bound, we conclude that $F(\rho,|\phi\rangle)\geq F_{\mathcal{K}^{\otimes n}}(\rho)-\epsilon$ with probability at least $1-\delta$.

The sample complexity is
\[M\cdot \Bigl\{O\Bigl(\frac{1}{\mu^3\tau}\log\frac{1}{\tau}\log n\Bigr)+\frac{2}{\epsilon^2}\log(4M/\delta)\Bigr\}=\frac{\log^2(1/\delta)}{\epsilon^2}(n|\mathcal{K}|)^{O(\log(1/\tau)/\mu)}.\]
The time complexity is
\[M\cdot \Bigl\{O\Bigl(\frac{1}{\mu^4\tau}\log^2\frac{1}{\tau}\log n+\frac{n}{\mu^3}\log\frac{1}{\tau}\log n\Bigr)+\frac{2n}{\epsilon^2}\log(4M/\delta)\Bigr\}=\frac{\log^2\frac{1}{\delta}}{\epsilon^2}(n|\mathcal{K}|)^{O(\log(1/\tau)/\mu)}\,.\]

\subsection{Proof of Corollary~\ref{cor:stab_product_base}}
\label{app:defer_stab_prod_base}

Suppose $|\phi_0\rangle\in\argmax_{|\varphi\rangle\in\mathcal{SP}}F(|\varphi\rangle,\rho)$.
That is, $F(\rho,|\phi_0\rangle)=F_{\mathcal{SP}}(\rho)\geq\tau$.

We can run the algorithm given by Corollary~\ref{cor:stab_prod_list} to obtain a list of states from $\mathcal{SP}$ of length at most $M \triangleq \log(2/\delta) \cdot (1/\tau)^{O(\log 1/\tau)}$ which contains all stabilizer product states with fidelity at least $\tau$ with $\rho$ with probability at least $1 - \delta/2$. We then use the classical shadows algorithm in~\Cref{lem:classical_shadow} to estimate the fidelities of all returned stabilizer states so that with probability at least $1-\frac{\delta}{2}$, every fidelity is estimated to error at most $\frac{\epsilon}{2}$.
The output is chosen to be the one with the highest estimated fidelity.
Conditioned on $|\phi_0\rangle$ appears in the list and the fidelities are estimated to error at most $\frac{\epsilon}{2}$, the fidelity of $|\phi\rangle$ is overestimated by at most $\frac{\epsilon}{2}$ while the fidelity of $|\phi_0\rangle$ is underestimated by at most $\frac{\epsilon}{2}$, and thus $|\phi\rangle$ has fidelity at least $F_{\cS}(\rho)-\epsilon$.
By union bound, we conclude that $F(\rho,|\phi\rangle)\geq F_{\cS}(\rho)-\epsilon$ with probability at least $1-\delta$.

Running the algorithm in~\Cref{thm:stab_product_base} takes $\log(2/\delta) \cdot (1/\tau)^{O(\log 1/\tau)} \cdot O(\frac{1}{\tau}\log\frac{1}{\tau}\log n)$ copies.
Running the classical shadows algorithm takes $\frac{4}{\epsilon^2}\log\frac{(\log(2/\delta)\cdot (1/\tau)^{O(\log 1/\tau)}}{\delta/2}$ copies.
Hence the sample complexity is
\begin{multline*}\log(2/\delta)\cdot (1/\tau)^{O(\log 1/\tau)}\cdot \log(2/\delta)\cdot O\Bigl(\frac{1}{\tau}\log\frac{1}{\tau}\log n\Bigr)+\frac{4}{\epsilon^2}\log\frac{\log(2/\delta)\cdot (1/\tau)^{O(\log 1/\tau)}}{\delta/2} \\
=\log n\log\frac{1}{\delta}(1/\tau)^{O(\log 1/\tau)}+\frac{1}{\epsilon^2}\Bigl(\log^2\frac{1}{\tau}+\log\frac{1}{\delta}\Bigr).
\end{multline*}

Running the algorithm in~\Cref{thm:stab_product_base} takes $\log(2/\delta)\cdot (1/\tau)^{O(\log 1/\tau)}\cdot O(\frac{1}{\tau}\log^2\frac{1}{\tau}\log n+n\log\frac{1}{\tau}\log n)$ time.
Running the classical shadows algorithm takes $\log(2/\delta)\cdot (1/\tau)^{O(\log 1/\tau)}\cdot O(\frac{4n^2}{\epsilon^2}\log\frac{\log(2/\delta)\cdot (1/\tau)^{O(\log 1/\tau)}}{\delta/2})$ time.
Hence the time complexity is
\begin{multline*}
    \log(2/\delta)\cdot (1/\tau)^{O(\log 1/\tau)} \cdot O\Bigl(\frac{1}{\tau}\log^2\frac{1}{\tau}\log n+n\log\frac{1}{\tau}\log n\Bigr) \\
    +\log(2/\delta)\cdot (1/\tau)^{O(\log 1/\tau)}O\Bigl(\frac{4n^2}{\epsilon^2}\log\frac{\log(2/\delta)\cdot (1/\tau)^{O(\log 1/\tau)}}{\delta/2}\Bigr)
    =\frac{n^2}{\epsilon^2}\log^2(1/\delta)\cdot (1/\tau)^{O(\log 1/\tau)}\,,
\end{multline*}
as claimed.

\end{document}